\documentclass[twoside,11pt]{article}
\usepackage{extarrows}   %

\newcommand*{\bbC}{\mathbb{C}}

\newcommand*{\bbE}{\mathbb{E}}

\newcommand*{\bbP}{\mathbb{P}}

\newcommand*{\bbR}{\mathbb{R}}

\newcommand*{\mcalF}{\mathcal{F}}

\newcommand*{\mcalH}{\mathcal{H}}

\newcommand*{\mcalL}{\mathcal{L}}
\newcommand*{\mcalM}{\mathcal{M}}
\newcommand*{\mcalN}{\mathcal{N}}

\newcommand*{\mcalT}{\mathcal{T}}

\newcommand*{\mcalV}{\mathcal{V}}

\newcommand*{\bfA}{\boldsymbol{A}}

\newcommand*{\qqquad}{\quad\,\quad\,\quad}

\newcommand*{\pair}[2]                   %
  {
    \langle
      #1,#2
    \rangle
  }

\newcommand*{\Pair}[2]                   %
  {
    \left\langle
      #1,#2
    \right\rangle
  }

\newcommand*{\ppair}[2]                   %
  {
    \langle\!\!
      \langle
        #1,#2
      \rangle\!\!
    \rangle
  }

\newcommand*{\Ppair}[2]                   %
  {
    \left\langle\!\!
      \left\langle
        #1,#2
      \right\rangle\!\!
    \right\rangle
  }

\newcommand*{\norm}[1]                   %
  {
    |\!|#1|\!|
  }

\newcommand*{\Norm}[1]                   %
  {
    \left|\!
      \left|
        #1
      \right|\!
    \right|
  }

\newcommand*{\nnorm}[1]                   %
  {
    |\!|\!|
      #1
    |\!|\!|
  }

\newcommand*{\Nnorm}[1]                   %
  {
    \left|\!
      \left|\!
        \left|
          #1
        \right|\!
      \right|\!
    \right|
  }

\newcommand*{\bra}[1]   %
 {
   \left\langle #1
   \right|
 }

\newcommand*{\ket}[1]   %
 {
   \left| #1
   \right\rangle
 }

\newcommand*{\braket}[2] %
 {
  \left\langle
   #1
   \middle|
   #2
  \right\rangle
 }

\newcommand*{\bracket}[3] %
 {
  \left\langle
   #1
   \middle|
   #2
   \middle|
   #3
  \right\rangle
 }

\newcommand*{\vacexp}[1] %
 {\left\langle #1 \right\rangle}

\newcommand*{\normal}[1] %
{\,:#1:\,}

\newcommand*{\Choose}[2]                  %
  {\left[{
      #1 \atop
      #2
  }\right]}

\newcommand*{\dd}[1]                       %
  {\frac{\mathrm{d}   }
        {\mathrm{d} #1}
  }

\newcommand*{\mulint}[1]                   %
  {
    \underbrace{
      \int\int\cdots\int
    }_{#1}\,\!
  }

  {
    \begin{equation*}
    \addtocounter{mythm}{1}
  }
  {
    \tag{\textbf{\themythm}}
    \end{equation*}
  }

\usepackage{amsthm}
\usepackage{subcaption}
\usepackage{enumitem}
\setlist[enumerate,1]{label=(\arabic*),ref=(\arabic*)}
\usepackage{latexsym,amsmath,amssymb,amsfonts,epsfig,graphicx,psfrag}
\usepackage{eepic,color,colordvi,amscd,xcolor}
\usepackage{float} 
\usepackage{comment}
\usepackage{placeins}
\usepackage{blindtext}
\usepackage[normalem]{ulem}

\usepackage[preprint]{jmlr2e}

\usepackage{lastpage}

\ShortHeadings{Inference on Dynamic Signed Network}{Shang et al.}
\firstpageno{1}

\begin{document}

\title{Inference for Balance in Dynamic Signed Networks\textsuperscript{\dag}}

\author{\name Ergan Shang \email eshang@andrew.cmu.edu \\
       \addr Department of Statistics and Data Science\\
       Carnegie Mellon University\\
       Pittsburgh, PA 15213, USA
       \AND
       \name Yuan Zhang \email yzhanghf@stat.osu.edu \\
       \addr Department of Statistics\\
       The Ohio State University\\
       Columbus, OH 43210, USA
       \AND
       \name Weijing Tang\thanks{Corresponding author.} \email weijingt@andrew.cmu.edu \\
       \addr Department of Statistics and Data Science\\
       Carnegie Mellon University\\
       Pittsburgh, PA 15213, USA}


\maketitle
{\renewcommand{\thefootnote}{\fnsymbol{footnote}}%
\footnotetext[2]{Codes for the simulation study (\Cref{Simulation})
and real data analysis (\Cref{Sec:real_data}) are available on GitHub at
\url{https://github.com/ergan-shang/inference_for_balance_in_dynamic_signed_networks}}}

\begin{abstract}%

Signed networks consist of both positive and negative relations, and structural balance theory provides an important conceptural framework for understanding their global tension structure. 
While existing statistical methods mainly focus on assessing empirical evidence of balance in a single observed network, many real-world signed relations evolve over time. 
This paper develops nonparametric inference for the population degree of structural balance at specified time points in dynamic signed networks, where the target time may or may not coincide with an observed snapshot. 
We consider a dynamic signed graphon model in which both edge formation and sign generation are governed by smoothly time-varying graphon functions. 
To exploit temporal smoothness, we construct a kernel-smoothed estimator that borrows information from snapshots near the target time point. 
Our theoretical analysis establishes a studentized inference procedure and a higher-order distributional approximation based on Edgeworth expansion, showing that temporal smoothing improves inference in sparse networks by reducing variance of observation noise, up to smoothing bias and time-discretization errors. 
We demonstrate the finite-sample performance and practical usefulness of the proposed method through extensive simulation studies and an application to a dynamic international relation network in political science.

\end{abstract}

\setcounter{tocdepth}{-1}

\section{Introduction}

Signed networks arise naturally in many social, political, and relational systems, where interactions may be cooperative or antagonistic, trusting or distrustful, friendly or hostile~\citep{leskovec2010signed, tang2016survey, kunegis2014applications}. 
Examples include alliances and rivalries in international relations, trust and distrust in online platforms, and positive and negative interpersonal ties in social groups. One important conceptural framework for understanding connectivity patterns in such signed networks is \textit{structural balance theory}~\citep{heider1946attitudes, harary1953notion}, which suggests that signed relations tend to organize into locally stable configurations. 
At the level of a triangle (i.e., a configuration of three nodes conntected to each other), balance is characterized by the product of the three edge signs: a triangle is called \textit{balanced} when this product is positive and \textit{unbalanced} otherwise.
This formalizes familiar heuristics such as ``the friend of my friend is my friend'' and ``the enemy of my enemy is my friend.'' 
The prevalence of balanced versus unbalanced triangles provides a natural summary statistics of the global tension structure of a signed network. 

A growing literature has developed statistical inference tools to assess empirical evidence for balance theory in observed signed networks. 
One line of work develops hypothesis tests against \textit{balance-free} null models, where an empirical balance statistic is compared with its distribution under a reference model that removes explicit preference for balanced triangles while preserving certain network features~\citep{facchetti2011computing, leskovec2010signed, feng2022testing}.
For example, \citet{leskovec2010signed} considered a null model that fixes the positions of edges while reshuffling their signs, whereas \citet{facchetti2011computing} assigned edge signs through an independent Bernoulli model.
Another line of work uses parametric random network models, including signed exponential random graph models (ERGMs),  to assess whether balance-related local configurations have significant effects on tie formation~\citep{gallo2024testing, fritz2025exponential}. 
In a different direction, \citet{chen2024nonparametric} developed an inference procedure for the population degree of balance under a  nonparametric signed graphon model. 
These works provide important tools for studying balance in real-world signed networks, but they primarily focus on a single observed network snapshot.

In many applications, signed relations are not static. 
Social relations, political alliances, and antagonistic interactions often evolve  over time.
As a result, the degree of structural balance may strengthen, weaken, or fluctuate across periods. 
This setting is related to the broader literature on  statistical modeling and inference tools for dynamic network data, which has developed a range of tools for capturing temporal dependence, evolving latent structure, and time-varying connectivity patterns~\citep[e.g.,][]{kim2018review,sewell2015latent,matias2017statistical,pensky2019dynamic,wang2021optimal,JMLR:v26:23-0444,athreya2025euclidean,lin2026dynamic}. 
Recent work has begun to study balance in dynamic signed networks from a mechanism-testing perspective. 
For example, \citet{fritz2025exponential} extended parametric ERGMs to dynamic signed networks, with the goal of assessing whether balance-related local configurations, such as common enemies or common friends, are associated with tie formation and sign evolution over time. 
In this work, we study a different inferential target: the population degree of balance at given time points, which may or may not coincide with observed network snapshots. 
This target is useful in practice because it provides a quantitative summary of the global tension structure of an evolving signed network. 
For example, in international relations, a lower degree of balance may reflect a less stable configuration of alliances and disputes, whereas a higher degree of balance may suggest the emergence of more coherent blocs. 
Thus, inference for the time-specific degree of balance can help researchers describe how balanced a system is at a given time and quantify its uncertainty.

When the target time point coincides with an observed snapshot, one could in principle apply existing one-snapshot network-moment inference methods for the degree of balance \citep{10.1214/11-AOS904, 10.1214/15-AOS1338, zhang2022edgeworth, chen2024nonparametric}. However, such methods only rely on the information available in a single network. 
In sparse signed networks, one snapshot may contain limited triadic information for reliable uncertainty quantification, and the validity of one-snapshot inference typically requires restrictions on the sparsity level. 
The availability of dynamic observations of signed relations provide an opportunity to overcome this limitation.  
In many systems, the mechanisms governing tie formation and tie sign do not change abruptly from one moment to the next, but instead evolve smoothly. 
When this temporal smoothness holds, snapshots near the target time point contain relevant information about the balance degree. 
Rather than treating repeated observations as isolated networks, one can borrow information from neighboring time points to stabilize inference. 

To formalize this idea without imposing a parametric model for tie formation or sign evolution, we consider a nonparametric dynamic graphon model for signed networks. 
This model accommodates node heterogeneity through latent variables and captures smooth temporal evolution in both edge formation and sign generation via two time-varying graphon functions. 
It can be viewed as a dynamic extension of time-static signed graphon models and as a signed-network counterpart of dynamic graphon models for binary networks~\citep{chen2024nonparametric, pensky2019dynamic}. 
Within this framework, the smooth temporal evolution of the underlying graphon functions implies that the population degree of balance itself also varies smoothly over time.

Motivated by this perspective, we propose an inference procedure based on edge-level kernel smoothing. 
We first smooth the edges observations across nearby time points and then construct the signed triangle moment from the smoothed edges. This edge-level construction is important because it reduces observation noise before the nonlinear triadic aggregation step. Compared with directly smoothing time-specific triangle moments computed from each network snapshot, edge-level smoothing leads to stronger variance reduction in sparse networks. We also develop a local leave-one-out cross-validation rule for bandwidth selection.

Our theoretical analysis establishes a studentized inference procedure and a higher-order distributional approximation based on Edgeworth expansion. 
The key insight is that edge-level temporal smoothing reduces the contribution of observation noise by a factor of order $(Th)^{-1}$, where $T$ is the number of observed time points and $h$ is the bandwidth. 
Consequently, compared to single-snapshot inference, the proposed method enlarges the sparsity regime under which valid inference is possible and improves the accuracy of the distributional approximation, up to the bias and discretization terms introduced by smoothing. 

The remainder of the paper is organized as follows. 
\Cref{sec:methodology} introduces the dynamic signed graphon model, the kernel-smoothed estimator, and the bandwidth selection procedure. \Cref{Theoretical Results} develops the studentized inference method and establishes the Edgeworth approximation. \Cref{Simulation} examines the finite-sample validity of the inference procedure using extensive simulation studies, followed by an application to a dynamic international relation network in \Cref{Sec:real_data}.

\paragraph{Notation} We use the standard asymptotic notation $O(\cdot)$, $o(\cdot)$, $\Omega(\cdot)$, and $\omega(\cdot)$, together with their probabilistic counterparts $O_p(\cdot)$ and $o_p(\cdot)$. 
For a random variable $Z$ and a deterministic sequence $\{\alpha_n\}$, we write $Z=\tilde O_{p,1}(\alpha_n)$ if $\bbP(|Z|\ge C\alpha_n)=O(n^{-1})$ for some constant $C>0$.
We write $Z_1 \gtrsim Z_2$ if there exists a constant $C>0$ such that $Z_1 \geq C Z_2$ for all sufficiently large $n$. We write $Z_1 \asymp Z_2$ if $Z_1 \gtrsim Z_2$ and $Z_2 \gtrsim Z_1$.

\section{Methodology} \label{sec:methodology}
We analyze dynamic interactions among $n$ nodes over a time interval $\mcalT\subset\bbR$, within which $T$ timestamps are observed, denoted by $\{t_\ell\in\mcalT\}_{\ell=1}^T$. At each timestamp $t_\ell$, we observe an undirected signed network represented by a signed symmetric adjacency matrix $\bfA(t_\ell)=[A_{ij}(t_\ell)]_{i,j=1}^n\in\{0, 1, -1\}^{n\times n}$. For $1\leq i<j\leq n$, $A_{ij}(t_\ell)=1$ if a positive edge exists between the nodes $i$ and $j$ at timestamp $t_\ell$, $A_{ij}(t_\ell)=-1$ if a negative edge exists, and $A_{ij}(t_\ell)=0$ if no edge is present. We assume the network has no self-loops, so $A_{ii}(t_\ell)=0$ for $1\leq i \leq n$.

Our goal is to infer the degree of structural balance in the dynamic signed network at a given time point $t^* \in \mcalT$ based on the observed sequence of signed networks. The target time point$t^*$ may or may not be one of the observed timestamps $\{t_\ell\in\mcalT\}_{\ell=1}^T$. We quantify the degree of balance through the expectation 
\begin{equation}
\label{mu_ijk}
\mu_{ijk}(t^*):=\bbE[A_{ij}(t^*)A_{jk}(t^*)A_{ki}(t^*)],
\end{equation}
which captures the expected sign of a triangle formed by nodes $i,j,k\in [n]$. Recall that the triangle is called balanced when the product $A_{ij}(t^*)A_{jk}(t^*)A_{ki}(t^*)$ is positive. Thus, $\mu_{ijk}(t^*)>0$ indicates that the triangle formed between nodes $i,j,k$ is more likely to be balanced, while $\mu_{ijk}(t^*)<0$ indicates the opposite. Its average over all node triplets, i.e.,
\[\mu(t^*):=\frac{1}{\binom{n}{3}}\sum_{i<j<k}\mu_{ijk}(t^*), \]
serves as a population-level global measure of balance. For example, if unbalanced triangles dominate the network at time $t^*$, the average will tend to be negative, whereas in a ``balance-free'' network, where positive and negative edges occur with equal probability, we expect the average to be near zero.

To estimate and quantify the uncertainty of the global measure of balance, we propose a dynamic graphon model for signed networks in \Cref{model}. This model captures node heterogeneity through latent variables associated with each node and accommodates flexible connectivity patterns via nonparametric graphon functions. When the graphon evolves smoothly over time, we further introduce a kernel-smoothed network moment estimator in \Cref{procedure}, which adaptively incorporates information from observed networks near the target time, along with a data-driven bandwidth selection method described in \Cref{Tuning}.

\subsection{Dynamic Graphon Model for Signed Networks}
\label{model}
Our dynamic graphon model is characterized by two measurable functions $F(\cdot, \cdot, \cdot), G(\cdot, \cdot, \cdot):[0, 1]^2\times\mcalT\rightarrow[0, 1]$, which are symmetric in the first two arguments, i.e., $F(x, y, t)=F(y, x, t)$ and $G(x, y, t)=G(y, x, t)$ for $x, y\in[0, 1]$ and $t\in\mcalT$. 

\begin{definition}(Dynamic Graphon Model for Signed Networks)
Let $X_1, \dots, X_n$ be i.i.d.\ latent variables drawn from the uniform distribution on $[0, 1]$. Conditional on $\{X_i\}_{i=1}^n$, for $t\in \mcalT$ and $1\leq i<j\leq n$, an edge between nodes $i$ and $j$ is independently drawn with probability $\rho_nF(X_i, X_j, t)$, i.e., $$\bbP(\vert A_{ij}(t)\vert=1\big\rvert X_i, X_j)=\rho_nF(X_i, X_j, t).$$ Then further conditional on an edge existing between nodes $i$ and $j$ at time $t$, independently of all others, it takes the negative sign with probability $G(X_i, X_j, t)$ and the positive sign otherwise, i.e., $$\bbP(A_{ij}(t)=-1\big\rvert\vert A_{ij}(t)\vert=1, X_i, X_j)=G(X_i, X_j, t).$$
Besides, $A_{ij}(t)=A_{ji}(t)$ and $A_{ii}=0$ for $i\neq j$ and $t \in \mcalT$. \label{def_dynamic_signed}
\end{definition}

In this model, the function $F$ governs the probability of edge formation between two nodes given their latent variables and time, while the function $G$ determines the distribution of edge signs. We impose no parametric assumptions on $F$ and $G$, which allows the model to flexibly capture complex time-varying connectivity and sign-generation patterns. To accommodate the increasing sparsity observed in many real-world networks as their size grows, we include a sparsity parameter $\rho_n=o(1)$ that diminishes as the network size $n$ grows. 
Since $\rho_n$ and $F$ are not separately identifiable, without loss of generality, we fix the scale by imposing the normalization $\int_{\mcalT}\int_{0}^1 \int_{0}^1 F(x,y,t)dx dy dt=1.$ Under this normalization, the parameter $\rho_n$ characterizes the global scale of network sparsity with respect to the network size $n$ and is assumed to be invariant over time. This corresponds to a setting where the overall scale of edge probabilities is governed by network growth, while temporal variation in connectivity is captured locally through the time-varying function $F(x,y,t)$. 

This model generalizes existing models in two directions. Compared to the static graphon model for signed networks~\citep{chen2024nonparametric}, both functions now include an additional time argument $t\in \mcalT$, which accommodates evolving connectivity and sign patterns over time. Moreover,  if the second layer of sign generation is omitted, the model reduces to a dynamic graphon model for binary networks~\citep{pensky2019dynamic}.

Under the proposed dynamic graphon model in \Cref{def_dynamic_signed}, the expectation $\mu_{ijk}(t^*)$ depends only on two functions $F$ and $G$ and is invariant to the specific node indices $i, j, k$, as the node latent variables are drawn i.i.d.\ from the same distribution. For clarity,  we therefore omit these subscripts in what follows. Our inference target, the global measure of balance, then simplifies to: $$\mu(t^*)= \bbE[A_{ij}(t^*)A_{jk}(t^*)A_{ki}(t^*)].$$

\subsection{Network Moment Estimation via Kernel Smoothing}
\label{procedure}
Our estimation procedure is motivated by the observation that network structures in many real-world systems evolve gradually. For example, relationships in social networks tend to develop or dissolve in response to an accumulation of events over extended periods. When the underlying probability of network connections varies smoothly over time, this temporal smoothness provides an opportunity to improve efficiency in estimation and inference for the balance measure $\mu(t^*)$. 
In particular, instead of relying solely on a single network snapshot at time $t^*$, we can borrow information from connectivity patterns observed at nearby time points. This is especially useful in sparse signed networks, where each snapshot contains limited observations and may not provide sufficient information for reliable estimation on its own.

To formalize this intuition, we assume that the underlying graphon functions $F$ and $G$ are smooth with respect to the time argument $t$ (see Assumption~\ref{assump: time obs} in Section~\ref{Theoretical Results} for a formal characterization). Under this condition, the signed network moment $\mu(t^*)$ naturally inherits the same temporal smoothness.
We propose a network moment estimator based on kernel smoothing, a widely used nonparametric technique for leveraging local neighborhood information~\citep{chen2017tutorial, zambom2013review}. 
Specifically, we estimate $\mu(t^*)$ using the empirical moment
\begin{equation}\tilde U_{nh}(t^*):=\binom{n}{3}^{-1}\sum_{i<j<k}\tilde A_{ij}(t^*)\tilde A_{jk}(t^*)\tilde A_{ki}(t^*).\label{estimator}
\end{equation} 
Unlike the standard network moment, this estimator is constructed from the kernel-smoothed edges $\tilde{A}_{ij}(t^*)$, which are defined as a weighted average of the observed edges over the time interval $\mcalT$: 
\[\tilde A_{ij}(t^*):=\frac{\sum_{\ell=1}^TK_h(t^*-t_\ell)A_{ij}(t_\ell)}{\sum_{\ell=1}^TK_h(t^*-t_\ell)},\]
where $K_h(\cdot)=K(\cdot / h)/h$ is the scaled kernel function with bandwith $h$. 

The kernel assigns weights to observed edges based on their scaled temporal distance to~$t^*$. 
We consider an order-$\nu$ kernel $K(\cdot):\mathbb{R}\rightarrow \mathbb{R}$ satisfying standard regularity conditions:
    \begin{enumerate}\label{assump: kernel}
        \item[(K1)] $K(\cdot)$ is even and satisfies $\int K(u) d u = 1$ with vanishing moments up to order $\nu-1$, i.e., $\kappa_{j1}=\int u^jK(u) d u=0$ for $1\le j\le \nu-1$ and $\kappa_{\nu1}=\int u^\nu K(u) d u \neq 0$;
        \item[(K2)] $K(\cdot)$ is bounded with $M:=\|K\|_{\infty} <\infty$; 
        \item[(K3)] There exist constants $C', c>0$, and $\alpha \ge 1$ such that for all $x \ge 0$ and $j\leq \nu$, \[\int_{x}^\infty |K(u)| d u \le C' e^{-c x^\alpha}\quad \text{and}\quad \int_x^\infty u^j|K(u)|du\leq C'x^{j-1}e^{-cx^\alpha};\]
        \item[(K4)] The function $u^jK(u)$ has bounded total variation for $0\leq j \le \nu$.
    \end{enumerate}
These conditions ensure bias control and exponential tail decay, which are satisfied by common compact-support kernels (e.g., uniform and Epanechnikov) and by the Gaussian kernel. For simplicity, we adopt the Gaussian kernel in subsequent sections.

The bandwidth parameter $h$ controls the degree of the smoothing. Intuitively, larger $h$ reduces the variance of the estimate by incorporating a wider range of observations but may introduce bias by over-smoothing local patterns. We provide a rigorous theoretical analysis of the impact of the bandwidth on the estimator in \Cref{Theoretical Results} and develop a practical data-driven selection procedure in \Cref{Tuning}.

\subsection{Bandwidth Selection}
\label{Tuning}

A common approach to bandwidth selection in nonparametric smoothing is \textit{leave-one-out} (LOO) cross-validation. The idea is to choose a bandwidth $h$ such that the estimator with bandwidth $h$ predicts each observation well when that observation is excluded from the smoothing procedure. Following this idea, we develop an LOO-based procedure for data-driven bandwidth selection in dynamic network moment estimation.

Specifically, for a given time point $s$ and each pair of nodes, we construct the LOO kernel-smoothed edge using all time points except $s$:
$$\overset{\circ}{A}_{ij}(h, s)=\frac{\sum_{\ell=1, t_\ell\neq s}^TK_h(t_\ell-s)A_{ij}(t_\ell)}{\sum_{\ell=1, t_\ell\neq s}^TK_h(t_\ell-s)} \quad \text{for } 1\leq i<j \leq n.$$
This serves as the LOO kernel smoothing predictor of the edge $A_{ij}(s)$. 
To evaluate prediction accuracy, we compare the network moment estimator based on $\overset{\circ}{A}_{ij}(h, s)$ with that based on the observed edges $A_{ij}(s)$ at time $s$. Recall that the population parameter is $\mu(s)=\bbE[A_{ij}(s)A_{jk}(s)A_{ki}(s)]$ with the empirical network moment 
\[\hat \bbE(A(s)):=\binom{n}{3}^{-1}\sum_{i<j<k}A_{ij}(s)A_{jk}(s)A_{ki}(s).\] 
The corresponding LOO network moment estimator is 
\[\hat \bbE(\overset{\circ}{A}(h, s)):= \binom{n}{3}^{-1}\sum_{i<j<k}\overset{\circ}{A}_{ij}(h,s)\overset{\circ}{A}_{jk}(h,s)\overset{\circ}{A}_{ki}(h,s).\]
The cross-validation criterion is then defined as the aggregated squared prediction error over time points in a neighborhood  $N(t^*, \tau)$ of the target inference point $t^*$, and the optimal bandwidth is selected as
\begin{equation}
    h=arg\min_h\sum_{s\in \mcalN(t^*, \tau)}\left(\hat E(\overset{\circ}{A}(h, s))-\hat E(A(s))\right)^2.\label{opt_select_h}
\end{equation}
Here, we restrict the summation to a local neighhorhood of $t^*$ because the smoothness of $\mu(t^*)$ may vary across different inference points. The neighborhood is defined by a user-specified fraction $\tau\in(0, 1]$ and contains the $\lfloor T\cdot\tau\rfloor$ time points closest to $t^*$. 
The parameter~$\tau$ therefore controls the extent to which local information is used for bandwidth selection. 
Our sensitivity analysis in \Cref{mcalF} suggests that larger values of $\tau$ are preferred when the number of observed time points $T$ is small, to stabilize optimization, while for large $T$, smaller $\tau$ values also yield stable results.

\section{Inference Procedure and its Theoretical Validity}
\label{Theoretical Results}
In this section, we develop an inference procedure for the global measure of balance $\mu(t^*)$ at the target time point. 
Our approach is based on analyzing the distribution of the studentized estimator $\tilde U_{nh}(t^*)$ defined in~\eqref{estimator}. 
Specifically, we construct a variance estimator and use it to form a studentized statistic.
As we shall show later, by leveraging connectivity patterns from nearby time points, kernel smoothing effectively reduces the edge-level noise variance by a factor of order $(Th)^{-1}$. 
The additional temporal information brings two benefits for network moment inference: 
(i) it broadens the validity regime of the inference procedure by relaxing the network sparsity condition, and (ii) it improves the accuracy of the distributional approximation, thereby yielding more reliable inference. 
After introducing the studentized statistic, we establish a higher-order distributional approximation using an Edgeworth expansion, which forms the basis for our confidence interval construction.

To simplify notation, we omit the superscript ${}^*$ for the target time point and use $t$ throughout this section. 
Additional technical details explaining the source of the $(Th)^{-1}$ variance reduction are deferred to \Cref{subsec: variance_tech}.

\subsection{Studentization and Distributional Approximation}
\label{subsec: studentization}
\paragraph{Studentization.}
To derive the studentization, we first need a variance estimator. Recall that the point estimator $\tilde U_{nh}(t)$ in~\eqref{estimator} involves two sources of randomness: one arises from the latent variables $\{X_i\}_{i=1}^n$ and another arises from observation error when sampling $\{\tilde{A}_{ij}(t)\}_{i,j\in [n]}$ given latent variables $\{X_i\}_{i=1}^n$. Different from one-snapshot method, for every pair of nodes, the observation error in $\tilde{A}_{ij}(t)$ is itself a kernel-weighted average of independent noise terms, which effectively reduces its variance. 
Guided by the analysis of the leading term in the variance decomposition of $\tilde U_{nh}(t^*)$ (see \citet{zhang2022edgeworth} and \Cref{subsec: variance_tech} for further details), we estimate the variance by
\begin{equation}
\hat S_{nh}^2(t)=\frac{9}{n^2}\sum_{i=1}^n\left\{\binom{n-1}{2}^{-1}\sum_{\substack{j<k\\ j, k\neq i}}\tilde A_{ij}(t)\tilde A_{jk}(t)\tilde A_{ki}(t)-\tilde U_{nh}(t)\right\}^2,\label{denominator}
\end{equation}
and define the studentized statistics 
\[\hat T_{nh}(t):=\frac{\tilde U_{nh}(t)-\mu(t)}{\hat S_{nh}(t)}.\]

\paragraph{Distributional approximation.}
Having constructed the studentized statistic, we next approximate its sampling distribution for inference.
To obtain a higher-order accurate approximation to the distribution of $\hat T_{nh}(t)$, we use an Edgeworth expansion based on the Hoeffding decomposition \citep{hoeffding1992class} of the underlying noiseless U-statistic, which captures the dominant source of variation.
Let $\tilde W_{ij}(t):= \bbE[\tilde A_{ij}(t)\mid X_i, X_j]$ and define 
\begin{align}\label{eq: noiseless U}
    U_{nh}(t):=\bbE[\tilde U_{nh}(t)|X_1, \cdots, X_n]=\binom{n}{3}^{-1}\sum_{i<j<k}\tilde W_{ij}(t)\tilde W_{jk}(t)\tilde W_{ki}(t).
\end{align}
The quantity $U_{nh}(t)$ is a noiseless U-statistics whose randomness arises only from the latent variables $X_1, \cdots, X_n$. 
Let $\mu_{h}(t):=\bbE[U_{nh}(t)]$.
The first- and second-order Hoeffding projections are defined as
\begin{align*}
g_{1h}(x,t)
&:=
\bbE\!\left[
\tilde W_{12}(t)\tilde W_{23}(t)\tilde W_{31}(t)
\mid X_1=x
\right]
-\mu_h(t),\\
g_{2h}(x_1,x_2,t)
&:=
\bbE\!\left[
\tilde W_{12}(t)\tilde W_{23}(t)\tilde W_{31}(t)
\mid X_1=x_1,X_2=x_2
\right]
-
g_{1h}(x_1,t)-g_{1h}(x_2,t)-\mu_h(t).
\end{align*}
Using these projection kernels, the population Edgeworth expansion of the studentized statistic is
\[G_{nh}(x)=\Phi(x)+\frac{\varphi(x)}{\sqrt n\xi_{1h}^3}\cdot\bigg\{\frac{2x^2+1}{6}\bbE[g_{1h}^3(X_1, t)]+(x^2+1)\bbE[g_{1h}(X_1, t)g_{1h}(X_2, t)g_{2h}(X_1, X_2, t)]\bigg\},\]
where $\xi_{1h}^2:=var(g_{1h}(X_1, t))$, and $\varphi(\cdot)$ and $\Phi(\cdot)$ denote the density and distribution functions of the standard normal distribution.

The population version $G_{nh}(x)$ involves coefficients that depend on the unknown data-generation distribution. We estimate these coefficients by their
empirical counterparts:
\begin{align*}
	\hat\xi_{1h}^3:=&\left(\frac{1}{n}\sum_{i=1}^n\hat g_{1h}^2(X_i, t)\right)^{3/2},\\
	\hat\bbE[g_{1h}^3(X_1, t)]:=&\frac{1}{n}\sum_{i=1}^n\hat g_{1h}^3(X_i, t),\\
	\hat\bbE[g_{1h}(X_1, t)g_{1h}(X_2, t)g_{2h}(X_1, X_2, t)]:=&\binom{n}{2}^{-1}\sum_{i<j}\hat g_{1h}(X_i, t)\hat g_{1h}(X_j, t)\hat g_{2h}(X_i, X_j, t),
\end{align*} 
where 
\begin{align*}
	\hat g_{1h}(X_i, t):=& \frac{1}{\binom{n-1}{2}}\sum_{j<k, i\neq j, i\neq k}\tilde A_{ij}(t)\tilde A_{jk}(t)\tilde A_{ki}(t)-\tilde U_{nh}(t), \\
	\hat g_{2h}(X_i, X_j, t):=&\frac{1}{n-2}\sum_{k\neq i,j}\tilde A_{ij}(t)\tilde A_{jk}(t)\tilde A_{ki}(t)-\hat g_{1h}(X_i,t)-\hat g_{1h}(X_j,t)-\tilde U_{nh}(t).
\end{align*}
The empirical Edgeworth expansion $\hat G_{nh}(x)$ is obtained by replacing the population quantities 
$\xi_{1h}^3$, $\bbE[g_{1h}^3(X_1,t)]$, and $\bbE[g_{1h}(X_1,t)g_{1h}(X_2,t)g_{2h}(X_1,X_2,t)]$
in $G_{nh}(x)$ with the corresponding empirical estimates above.

\paragraph{Smoothing the distribution.}
Note that the distribution of $\hat T_{nh}(t)$ may exhibit minor lattice effects due to the discreteness of network statistics. These small jump points can affect higher-order accuracy of distributional approximation. 
Following~\cite{Shao05082025}, we add a small independent Gaussian perturbation $\delta_T\sim\mcalN(0, c_\delta n^{-1}\log n)$ with a sufficiently large constant $c_\delta$ and conduct inference using the perturbed statistic $\tilde T_{nh}(t):=\delta_T+\hat T_{nh}(t)$.
The perturbation variance vanishes asymptotically and does not affect the limiting distribution, but it smooths the distribution sufficiently for the Edgeworth approximation.

\paragraph{Theoretical validity.}
We now provide theoretical guarantees for the accuracy of the Edgeworth expansion as an approximation to the distribution of $\tilde T_{nh}(t)$. 
We assume the following assumptions.
\begin{assumption}[Temporal resolution and bandwidth]\label{assump: time obs}
    Let $0=t_0<t_1<\cdots <t_T=\mcalL$. Assume $\Delta_T:=\max_{1\le \ell \le T}(t_\ell - t_{\ell-1}) \le C/T$ for some constant $C>0$, and that $t$ is an interior evaluation point satisfying $\min\{\mcalL-t, t\} \ge \delta$ for some constant $\delta>0$. 
    In addition, we assume $h \to 0$ and $Th\to \infty$.
\end{assumption}
\begin{assumption}[Graphon smoothness and boundedness]
\label{assump: smooth graphon}
Let $\beta = \nu' + \alpha$ with integer $\nu' \ge \nu$ and $\alpha \in (0,1]$.
Assume that for all $(x,y) \in [0,1]^2$, 
the functions $s \mapsto F(x,y,s)$ and 
$s \mapsto G(x,y,s)$ belong to the Hölder class 
$\mathcal{H}^\beta(L)$ on $\mcalT$; that is, they are $\nu'$-times differentiable in $s$, and their $\nu'$-th derivatives in $s$ satisfy, for all $s,s' \in \mcalT$, 
\[
\left| \partial_t^{\nu'} F(x,y,s) - \partial_t^{\nu'} F(x,y,s') \right|
\le L |s-s'|^{\alpha},
\quad
\left| \partial_t^{\nu'} G(x,y,s) - \partial_t^{\nu'} G(x,y,s') \right|
\le L |s-s'|^{\alpha}.
\]
Assume $0\le F(x,y,s)\le C_F$ uniformly over \((x,y,s)\in[0,1]^2\times\mcalT\) for some constant $C_F>0$.
\end{assumption}
\begin{assumption}[Non-degeneracy for U-statistics]\label{non-u}
Let $W_{ij}(t)=\bbE[A_{ij}(t)|X_i, X_j]$ and $g_1(x,t):=\bbE[W_{12}(t)W_{23}(t)W_{31}(t)|X_1=x]-\mu(t)$. Assume that $\xi_1^2:=var(g_{1}(X_1,t))\geq C_2\rho_n^{6}$ for some  constant $C_2>0$.
\end{assumption}
\begin{assumption}[Network sparsity regime] \label{assump: sparsity}
   Assume that $\rho_n\gtrsim n^{-2/3+\epsilon}(Th)^{-1}\log^{5/3}n$ for some sufficiently small constant $\epsilon>0$. %
\end{assumption}

The following theorem establishes the approximation error of the Edgeworth expansion, the proof of which is provided in Section~\ref{Proof of main thm}.  

\begin{theorem}\label{main thm}
Suppose \Cref{assump: time obs,assump: smooth graphon,assump: sparsity,non-u} hold, then we have 
 $$\left\Vert \mcalF_{\tilde T_{nh}}(\cdot) -  G_{nh}(\cdot) \right\Vert_\infty= O(\mcalM(n, \rho_n, T, h)),$$ 
 $$\left\Vert \mcalF_{\tilde T_{nh}}(\cdot) - \hat G_{nh}(\cdot) \right\Vert_\infty=\tilde O_{p, 1}(\mcalM(n, \rho_n, T, h)),$$  
where $\mcalF_{\tilde T_{nh}}$ is the cumulative distribution function of the perturbed statistics $\tilde T_{nh}$ and $$\mcalM(n, \rho_n, T, h)=n^{-1}\log^{3/2}(n)\cdot\max\left\{1, (Th\rho_n)^{-3/2}\right\}+\sqrt nh^\nu+\frac{\sqrt n}{Th}.$$
Here $\|\cdot\|_\infty$ denotes the Kolmogorov--Smirnov distance.
\end{theorem}

Assumption~\ref{assump: time obs} requires that the temporal observations have sufficient resolution and that the target evaluation point lies in the interior of the time interval.
Assumption~\ref{assump: smooth graphon} imposes Hölder smoothness on the time evolution of the graphon functions. The requirement $\nu'\ge \nu$ guarantees that the graphon is sufficiently smooth relative to the kernel order so that the kernel-smoothed estimator achieves bias of order $h^\nu$. These two assumptions are standard regularity conditions in nonparametric kernel smoothing; see, for example,~\citet{kim2012robust}.
Assumption~\ref{non-u} ensures non-degeneracy of the first-order projection in the Hoeffding decomposition. Similar conditions are commonly imposed in the literature on network moment inference~\citep{10.1214/11-AOS904,zhang2022edgeworth}.
Assumption~\ref{assump: sparsity} specifies the admissible sparsity regime of the network. 
Compared with the one-snapshot-based inference procedure (see Remark~\ref{rmk:one_snapshot} for details), this assumption allows a wider range of network sparsity. 
In particular, when inference is based only on a single observed network at time $t$, the attainable sparsity regime is fundamentally limited by the information contained in that snapshot, leading to the condition $\rho_n \gtrsim n^{-2/3+\epsilon}$ up to logarithmic factors~\citep{10.1214/11-AOS904, 10.1214/15-AOS1338, zhang2022edgeworth}. 
In the dynamic setting considered here, kernel smoothing allows us to borrow information from nearby time points. As a result, our condition contains an additional factor $(Th)^{-1}$, reflecting the effective increase in sample size due to temporal smoothing; see~\Cref{subsec: variance_tech} for further technical details.
Under the standard nonparametric regime $Th\to\infty$, this substantially enlarges the range of sparsity levels under which valid inference can be achieved.

Theorem~\ref{main thm} shows that the empirical Edgeworth expansion $\hat G_{nh}$ provides a uniformly accurate approximation to the distribution of the perturbed statistic $\tilde T_{nh}$. 
The approximation error is controlled by $\mcalM(n,\rho_n,T,h)$, which consists of three components. 
The first term reflects stochastic error arising from the network observations. 
The second term $\sqrt{n}h^\nu$ corresponds to the bias induced by kernel smoothing. 
The third term $\sqrt{n}(Th)^{-1}$ captures the discretization error due to observing the network only at discrete time points. 
Together, these terms quantify the trade-off introduced by kernel smoothing: incorporating temporal neighborhood information reduces stochastic variability, while introducing additional smoothing bias and discretization error.

\begin{remark}[One-snapshot-based inference]\label{rmk:one_snapshot}
As a special case, when the target time point~$t$ coincides with an observed timestamp $t_\ell$, one may conduct inference using only the single observed network $\mathbf A(t_\ell)$. 
This one-snapshot method replaces $\tilde A_{ij}(t)$ in \Cref{estimator,denominator} by the observed adjacency entries $A_{ij}(t_\ell)$.
In this case, the corresponding studentized statistic yields an Edgeworth approximation error of order $n^{-1}\,\rho_n^{-3/2}$ up to logarithmic factors.
Compared with this baseline, the proposed kernel-smoothed procedure improves the approximation accuracy through the factor $(Th)^{-3/2}$ and allows inference under sparser network regimes, at the cost of the additional bias and discretization error described above.
\end{remark}

\begin{remark}[Optimal bandwidth and the corresponding approximation error]\label{rmk:optimal_bandwidth}

The optimal bandwidth $h$ and the corresponding approximation error depend on the interplay between the network sparsity $\rho_n$, the temporal resolution $T$, and the network size $n$. Ignoring logarithmic factors, the optimal $h_{\mathrm{opt}}$ is determined by minimizing the bound $\mcalM (n,\rho_n,T,h) \approx n^{-1}(Th\rho_n)^{-3/2} + \sqrt{n}h^\nu + \sqrt{n}(Th)^{-1}$.We discuss two primary regimes of interest:
\paragraph{Case 1: Extremely sparse regime.} 
When the network is extremely sparse such that $(n\rho_n)^3 Th \ll 1$, the stochastic error dominates the temporal discretization error, i.e., $\sqrt{n}(Th)^{-1}$ $\ll n^{-1}(Th\rho_n)^{-3/2}$.
Therefore, the primary trade-off is between the network stochastic error and the smoothing bias. Balancing $n^{-1}(Th\rho_n)^{-3/2} \asymp \sqrt{n}h^\nu$ yields the optimal bandwidth:
$$h_{\mathrm{opt}} \asymp (n\rho_n T)^{-\frac{3}{2\nu+3}}.$$
Substituting $h_{\mathrm{opt}}$ back into the bound yields an optimal approximation error of $\mcalM_{\mathrm{opt}} \asymp n^{\frac{3-4\nu}{4\nu+6}}(T\rho_n)^{-\frac{3\nu}{2\nu+3}}$.
Importantly, this regime implies that $n\rho_n \to 0$, under which the sparsity condition $\rho_n \gtrsim n^{-2/3+\epsilon}$ required for valid one-snapshot-based inference is violated.
In contrast, the kernel-smoothed statistics aggregates edges across nearby time points, effectively increasing the amount of available signal by a factor of $Th$. As long as the sparsity condition in Assumption~\ref{assump: sparsity} holds, our kernel-smoothed procedure remains valid.

\paragraph{Case 2: Moderately sparse to dense regime.}
In denser networks where $(n\rho_n)^3 Th \gg 1$, which includes the common regime $n\rho_n \to \infty$, the temporal discretization error dominates the stochastic network error, i.e., $\sqrt{n}(Th)^{-1} \gg n^{-1}(Th\rho_n)^{-3/2}$. Here, the optimal bandwidth is determined by balancing the discretization error and the smoothing bias ($\sqrt{n}(Th)^{-1} \asymp \sqrt{n}h^\nu$), which yields:$$h_{\mathrm{opt}} \asymp T^{-\frac{1}{\nu+1}}.$$
This results in an optimal approximation error of $\mcalM_{\mathrm{opt}} \asymp \sqrt{n}T^{-\frac{\nu}{\nu+1}}$.
In this regime, the sparsity condition $\rho_n \gtrsim n^{-2/3+\epsilon}$ required for valid one-snapshot-based inference can hold and both procedures become applicable. 
The comparison between our kernel-smoothed procedure and the one-snapshot baseline ($\mcalM_{\mathrm{one-snapshot}} \asymp n^{-1}\rho_n^{-3/2}$) can be quantified by the ratio:
$$\frac{\mcalM_{\mathrm{opt}}}{\mcalM_{\mathrm{one-snapshot}}} \asymp \frac{(n\rho_n)^{3/2}}{T^{\frac{\nu}{\nu+1}}}.$$
This comparison suggests that the kernel-smoothed method achieves a strictly smaller approximation error provided $T^{\frac{\nu}{\nu+1}} \gg (n\rho_n)^{3/2}$, or equivalently, $T \gg (n\rho_n)^{\frac{3(\nu+1)}{2\nu}}$. 
On the other hand, if $T$ is small and the network is very dense, relying on a single snapshot may be preferable to avoid temporal discretization errors.
\end{remark}

\paragraph{Construction of confidence intervals.}
Theorem \ref{main thm} establishes that the maximum deviation between the cumulative distribution function of $\tilde T_{nh}$ and the empirical Edgeworth expansion $\hat G_{nh}$ is bounded by $\mcalM(n, \rho_n, T, h)$ with high probability. This result can be utilized to construct Cornish-Fisher confidence intervals \citep{fisher1960percentile} with provable coverage probabilities as follows.

\begin{corollary}
\label{upper-lower-CI}
    Assuming the conditions of Theorem \ref{main thm} hold, for any $\alpha\in (0,1)$, the two-sided confidence interval for $\mu(t)$ given by \begin{equation}\left(\tilde U_{nh}(t)-\hat q_{\hat T_n, 1-\alpha/2}\cdot \hat S_n(t), \tilde U_{nh}(t)-\hat q_{\hat T_n, \alpha/2}\cdot \hat S_n(t)\right)\label{interval}\end{equation} has a $1-\alpha+O(\mcalM(n, \rho_n, T, h))$ coverage probability, where \begin{equation}\hat q_{\hat T_n, \alpha}:=z_\alpha - \frac{1}{\sqrt n\hat\xi_{1h}^3}\left(\frac{2z_\alpha^2+1}{6}\hat\bbE[g_{1h}^3(X_1, t)]+(z_\alpha^2+1)\hat\bbE[g_{1h}(X_1, t)g_{1h}(X_2, t)g_{2h}(X_1, X_2, t)]\right)-\delta_T\label{quantile}\end{equation} and $z_\alpha=\Phi^{-1}(\alpha)$. 

Similarly, a one-sided upper confidence bound for $\mu(t)$ with $1-\alpha + O(\mcalM(n, \rho_n, T, h))$ coverage is defined as 
\[\left(-\infty, \tilde U_{nh}(t) - \hat q_{\hat T_n, \alpha} \cdot \hat S_n(t)\right).\] 
The one-sided lower bound is defined analogously.

\end{corollary}

\begin{remark}
\label{null_val}
The results in Theorem \ref{main thm} also provide a foundation for performing one-sided or two-sided hypothesis tests for the balance measure $\mu(t)$. However, how to choose a ``balance-free'' null value $\mu_0$, corresponding to a baseline under which balanced triangles are not preferred over unbalanced ones, is itself a non-trivial problem and has been discussed in the literature~\citep{facchetti2011computing, leskovec2010signed, feng2022testing}. A simple choice is to adopt a null model that assumes equal probabilities for positive and negative edges, which implies $\mu_0=0$. Alternatively, following \cite{facchetti2011computing}, one may adopt a conditional null model that preserves the overall edge density and sign proportion of the observed network.
Specifically, each observed edge is assigned a negative sign independently with
probability $s$, where $s$ is taken to be the observed fraction of negative edges
at time $t$. Under this null model, the corresponding null balance measure is determined by 
\begin{equation*}
\mu_0 = \rho_n^3 [ \underbrace{\left((1-s)^3 + 3s^2(1-s)\right)}_{\text{Expected positive triangles}} - \underbrace{\left(s^3 + 3(1-s)^2s\right)}_{\text{Expected negative triangles}}],
\end{equation*}
where $\rho_n$ denotes the observed edge density.
This choice of $\mu_0$ provides a baseline that reflects the expected degree of balance in a random signed network with the same edge density and sign proportion.

Given a null value $\mu_0$, consider the test $H_0: \mu(t) \leq \mu_0$ versus $H_1: \mu(t) > \mu_0$. The empirical $p$-value can be evaluated as:
\begin{equation}
\label{p_value}
\hat p = 1 - \hat G_{nh} \left( \frac{\tilde U_{nh}(t) - \mu_0}{\hat S_{nh}(t)} \right).
\end{equation}
Following \cite{zhang2022edgeworth}, rejecting $H_0$ if $\hat p < \alpha$ controls Type-I error at $\alpha + O(\mcalM(n, \rho_n, T, h))$ and remains sensitive to deviations of order $\omega(n^{-1/2})$.
\end{remark}

\subsection{Variance Decomposition and Technical Insights}
\label{subsec: variance_tech} 
To see where the $(Th)^{-1}$ factor in~\Cref{assump: sparsity} comes from, we decompose the estimation error into three terms
$$\tilde U_{nh}(t) - \mu(t) = \underbrace{{\tilde U_{nh}(t) - U_{nh}(t)}}_{(I) \text{ Observation Noise}} + \underbrace{{U_{nh}(t) - \mu_{h}(t)}}_{(II) \text{ Latent Variable Variation}} + \underbrace{{\mu_h(t) - \mu(t)}}_{(III) \text{ Smoothing Bias}},$$ 
where $U_{nh}(t)$ is the noiseless U-statistics defined in (\ref{eq: noiseless U}), and $\mu_{h}(t):=\bbE[U_{nh}(t)]$ denotes its expectation. 
Term (I) represents the observation noise from edge sampling conditional on the latent variables.  
Term (II) captures the variation induced by the latent variables $X_1,\ldots,X_n$.  
Term (III) is the deterministic bias introduced by kernel smoothing.

The effect of kernel smoothing can be seen clearly through the observation-noise component in the above decomposition. 
For the smoothed estimator, term (I) takes the form
\begin{align*}
	\tilde U_{nh}(t)-U_{nh}(t) = & \binom{n}{2}^{-1} \sum_{i<j}  \frac{3\sum_{k \neq i,j} \tilde W_{ik}(t) \tilde W_{jk}(t)}{n-2}\tilde \eta_{ij}(t) \\
	& + \binom{n}{3}^{-1} \sum_{i<j;k\neq i,j}\tilde W_{ij}(t)\tilde \eta_{ik}(t) \tilde \eta_{jk}(t) + \binom{n}{3}^{-1} \sum_{i<j<k}\tilde \eta_{ij}(t)\tilde \eta_{jk}(t) \tilde \eta_{ki}(t),
\end{align*}
which is a polynomials of independent mean-zero noise terms $\tilde \eta_{ij}(t):=  \tilde A_{ij}(t)-\tilde W_{ij}(t)$.
The key difference from the one-snapshot setting lies in the variance of these edge-noise terms. 
In the one-snapshot setting, the corresponding noise variables are $\eta_{ij}(t):=A_{ij}(t)-\bbE[A_{ij}(t)|X_i, X_j]$.
Under kernel smoothing, each $\tilde\eta_{ij}(t)$ is itself a weighted average of independent noise terms across nearby time points: $\tilde\eta_{ij}(t)=\sum_{\ell} w_\ell(t)\eta_{ij}(t_{\ell})$ with $w_\ell(t): = K_h(t-t_\ell) / \sum_{\ell'} K_h(t-t_{\ell'})$. 
Because the edge noises are independent across time, averaging over approximately $Th$ snapshots reduces the variance of $\tilde\eta_{ij}(t)$ by a factor of order $(Th)^{-1}$ relative to $\eta_{ij}(t)$. 
Under the sparsity condition in~\Cref{assump: sparsity}, the dominant source of variation in the smoothed estimator is governed by the latent-variable component in (II). 
The leading contribution of this component is captured by the first-order Hoeffding projection, which motivates the variance estimator in~\eqref{denominator}. 
The proof is given in~\Cref{Proof of variance prop}.

\begin{remark}[Comparison with snapshot-level smoothing]
    An alternative approach is to smooth the balanced triangle statistics at the snapshot level. Specifically, one may first compute the triangle statistic $\hat U_n(t_\ell)=\binom{n}{3}^{-1}\sum_{i<j<k}A_{ij}(t_\ell)A_{jk}(t_\ell)A_{ki}(t_\ell)$ at each snapshot $t_\ell$, and then apply kernel smoothing to obtain $\hat U_{nh}(t)=\frac{1}{\sum_{\ell=1}^TK_h(t-t_\ell)}\sum_{\ell=1}^TK_h(t-t_\ell)\hat U_n(t_\ell)$. 
    However, we choose edge-level smoothing because it can achieve a smaller variance by exploiting a higher degree of averaging. This advantage can be clearly seen in the cubic remainder term of the observation noise. 
    For the snapshot-level estimator, the cubic remainder takes the form
    \[\binom{n}{3}^{-1}\sum_{i<j<k}\sum_\ell w_\ell(t)\eta_{ij}(t_\ell)\eta_{jk}(t_\ell)\eta_{ki}(t_\ell),\]
    where $w_\ell(t): = K_h(t-t_\ell) / \sum_{\ell'}K_h(t-t_{\ell'})$.
    In contrast, for the proposed edge-smoothed stimator, the corresponding cubic remainder is 
    \begin{align*}
    	& \binom{n}{3}^{-1} \sum_{i<j<k}\tilde \eta_{ij}(t)\tilde \eta_{jk}(t) \tilde \eta_{ki}(t) \\
    	= & \binom{n}{3}^{-1} \sum_{i<j<k} \left(\sum_{\ell_1} w_{\ell_1}(t) \eta_{ij}(t_{\ell_1})\right) \left(\sum_{\ell_2} w_{\ell_2}(t) \eta_{jk}(t_{\ell_2})\right) \left(\sum_{\ell_3} w_{\ell_3}(t) \eta_{ki}(t_{\ell_3})\right).
    \end{align*}
After expanding this product, we see it averages over separate time indices $(\ell_1,\ell_2,\ell_3)$ for the three edges in the triangle.
Because edge-level noises are mean-zero and independent across both time and node pairs, most cross products vanish when computing the variance. Only a small subset of index-matching terms contributes.
Because the kernel window contains an effective number of time points of order $Th$, this yields a stronger averaging effect than snapshot-level smoothing.
In particular, for the cubic remainder in the observation noise, the variance contribution is reduced by a factor of order $(Th)^{-3}$ under edge-level smoothing, compared with $(Th)^{-1}$ under snapshot-level smoothing. 
\end{remark}

\section{Simulation Studies}
\label{Simulation}
In this section, we conduct simulation studies to compare the empirical performance of our proposed method with three alternative methods:
\begin{itemize}
    \item \textbf{Our Method (Normal Approx.)}: Uses the same kernel-smoothed studentized estimator $\hat T_{nh}$ as our proposed method, but approximates the distribution of $\hat T_{nh}$ by the standard normal distribution, omitting the higher-order correction term in the Edgeworth expansion. 
    \item \textbf{One Snapshot}: Uses only a single snapshot of the dynamic network observed at the target time point, without aggregating information across time. As described in Remark \ref{rmk:one_snapshot}, the distributional approximation is based on the Edgeworth expansion, but all quantities are computed using only this single snapshot.
    \item \textbf{One Snapshot (Normal Approx.)}: A simplified version of the one-snapshot method that replaces the Edgeworth-expansion-based quantile with that from the standard normal approximation.
\end{itemize}

We evaluate the empirical coverage relative to the nominal level of $90\%$, as well as the average lengths of both two-sided and one-sided confidence intervals constructed by all methods.
We generate data from the dynamic graphon model in~\Cref{def_dynamic_signed} under two settings of graphon functions: 
\begin{enumerate}[label=(\arabic*)]
    \item \label{two-poly}$F(x, y, t)=1.6(1-t^2)(x^2+y^2)$ and 
    $G(x, y, t)=0.4(1-t^2)(x^2+y^2)$; 
    \item \label{balance} $F(x, y, t)\equiv1$ and $G(x, y, t)=\left(1+\exp\left(500\left(t+0.5\right)(x-0.4)(y-0.4)\right)\right)^{-1}$. 
\end{enumerate}
The second setting represents a two-community structure consistent with balance theory, where edges between nodes with $x,y<0.4$ (first block) or $x,y>0.4$ (second block) are more likely to be positive. The time-dependent term further captures how the degree of structural balance evolves over time.
Unless otherwise specified, we set the network size $n=200$, the number of timestamps $T=200$, and the sparsity parameter $\rho_n=0.25$. 
The observed timestamps are equally spaced over $[-0.5, 0.5]$. %
For kernel-smoothed estimators, we adopt the Gaussian kernel and select the bandwidth using the LOO procedure proposed in \Cref{Tuning} with $\tau=0.1$.

Below, we analyze how performance varies with the network size, the sparsity, the number of observed snapshots, and the degree of balance. In \Cref{CDF approx}, we further provide a comparison of the cumulative distribution functions (CDFs) between the Edgeworth expansion-based approximation and the standard normal approximation, which demonstrates that the Edgeworth expansion yields a more accurate approximation to the true CDF than the standard normal approximation. In \Cref{mcalF}, we also provide a sensitivity analysis of the tuning parameter $\tau$, where the results remain stable across different choices of~$\tau$ when~$T$ is large, while larger values of~$\tau$ are preferred when~$T$ is small.

\subsection{Impact of Network Size}

In this subsection, we evaluate how the network size affects the coverage proportion and the average length of confidence intervals (CIs) constructed under different methods. 
Specifically, we vary the network size with a range of $n\in\{50, 100, 200, 400\}$ in two settings  and run $2,000$ Monte Carlo replications for each setting.
\begin{figure}[!h]
    \centering
    \includegraphics[width=\linewidth]{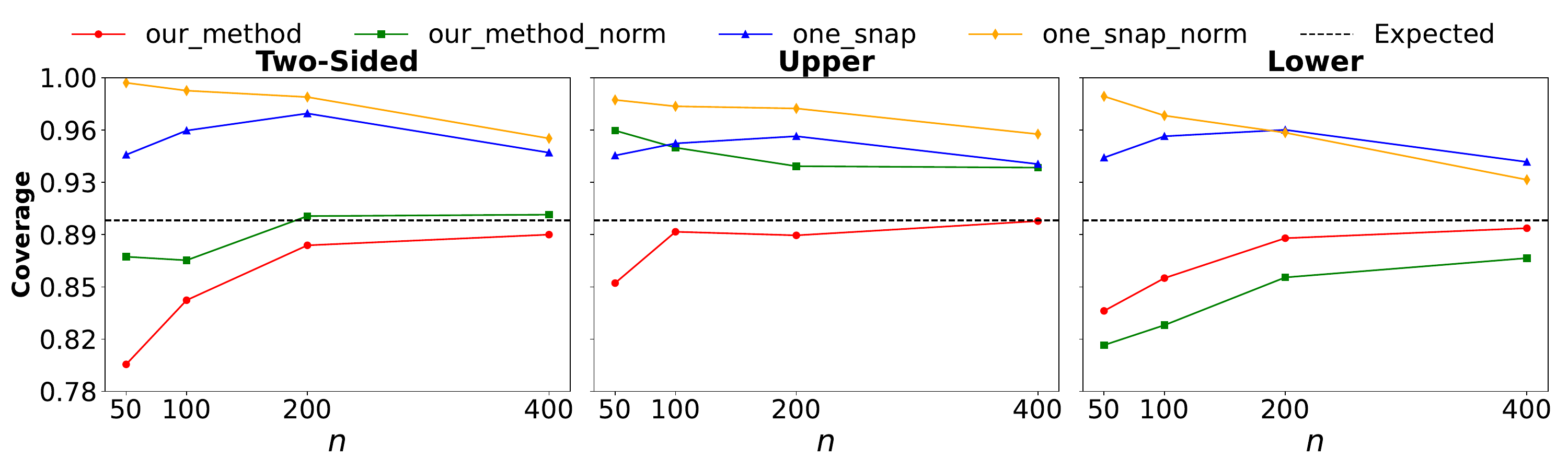}
    \caption{
    Coverage proportion  across varying network sizes under Setting~\ref{two-poly}, where the target inference time point is $t^*=0$.}
    \label{CoverageAcrossNUnderTwoPolyT1000RhoS08T0Frac10}
\end{figure}
\begin{figure}[!h]
    \centering
    \includegraphics[width=\linewidth]{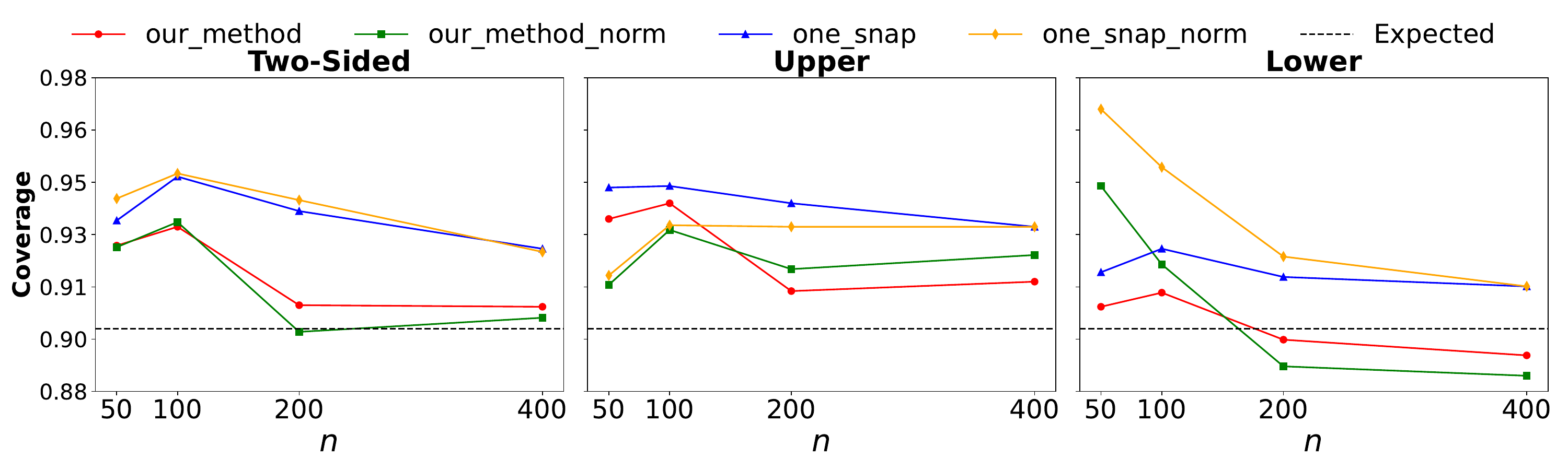}
    \caption{Coverage proportion  across varying network sizes under Setting~\ref{balance}, where the target inference time point is $t^*=-0.1$.}
    \label{CoverageAcrossNUnderBalanceT200Rho025T200Frac10}
\end{figure}
\begin{figure}[!h]
    \centering
    \includegraphics[width=\linewidth]{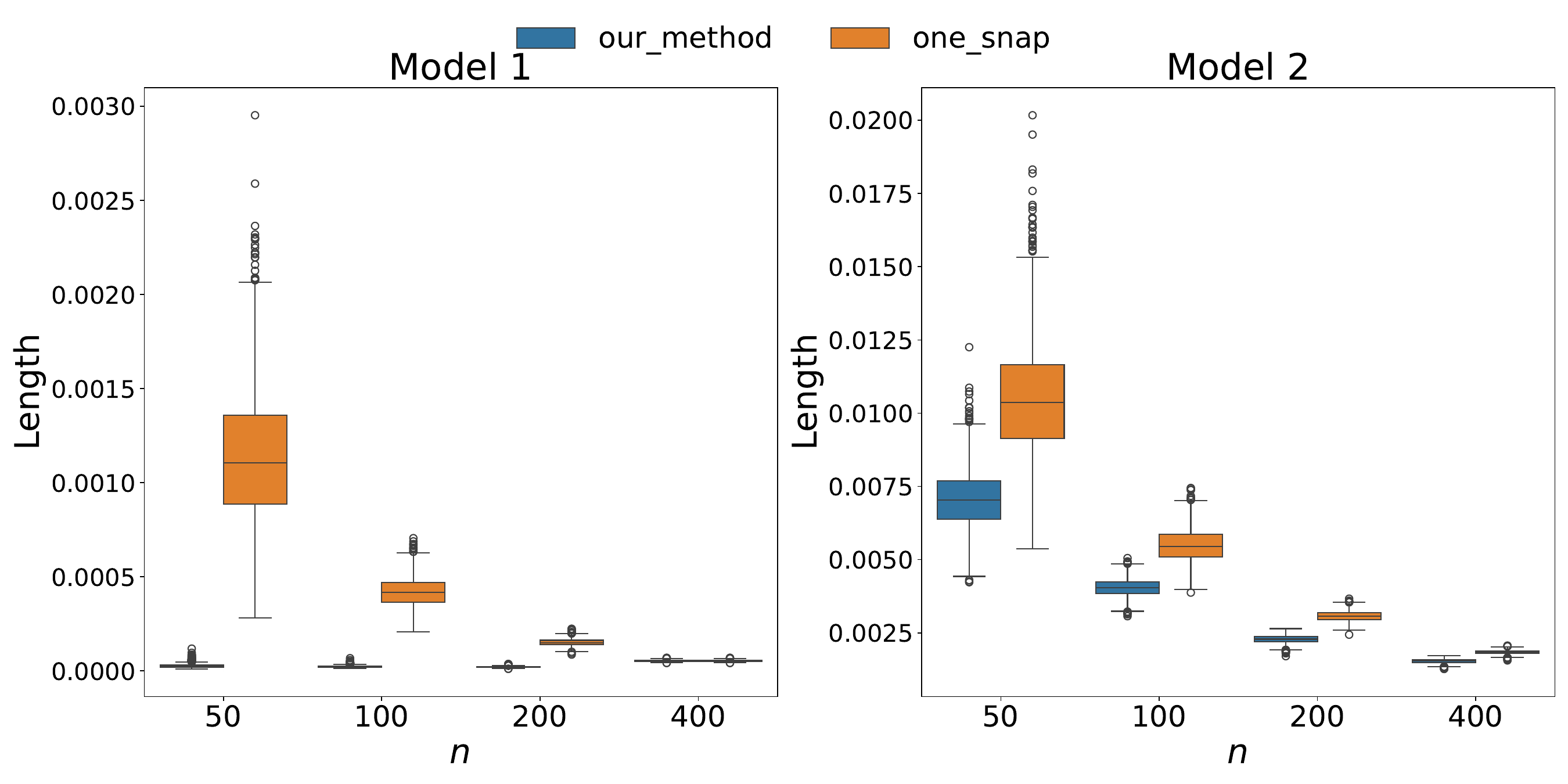}
    \caption{Length of CIs across varying network sizes. The left panel corresponds to Setting~\ref{two-poly} with $t^*=0$ and the right panel corresponds to Setting~\ref{balance} with $t^*=-0.1$.}
    \label{BoxplotOf2MethodsAcrossDifferentNetworkSizesT200Rho025}
\end{figure}
As shown in \Cref{CoverageAcrossNUnderTwoPolyT1000RhoS08T0Frac10,CoverageAcrossNUnderBalanceT200Rho025T200Frac10}, as the network size increases, the empirical coverage for all methods converges toward the nominal level $90\%$. Compared with the standard normal approximation and the one-snapshot methods, our method achieves coverage proportions closer to the nominal level in most cases. 
In addition, as shown in \Cref{BoxplotOf2MethodsAcrossDifferentNetworkSizesT200Rho025}, the lengths of the CIs for all methods consistently decrease as the network size increases. Importantly, given comparable coverage, our method that benefits from leveraging information from neighboring time points  consistently achieves shorter intervals than the one-snapshot alternatives.

\subsection{Impact of Network Sparsity}
\label{sec:one-poly}
To investigate how network sparsity affects the inference validity, we modify Setting~\ref{two-poly}
such that the probability of edge formation is entirely determined by the sparsity parameter $\rho_n$:
\begin{enumerate}[label=(\arabic*), start=3]
    \item \label{one-poly} $F(x, y, t)\equiv1$ and $G(x, y, t)=0.4(1-t^2)(x^2+y^2)$.
\end{enumerate}
We vary $\rho_n \in \{n^{-0.4}, n^{-0.25}, n^{-0.2}, n^{-0.1}\}$,
which yields network densities of approximately 12.0\%, 26.6\%, 34.7\%, 58.9\%.

\begin{figure}[!h]
    \centering
    \includegraphics[width=\linewidth]{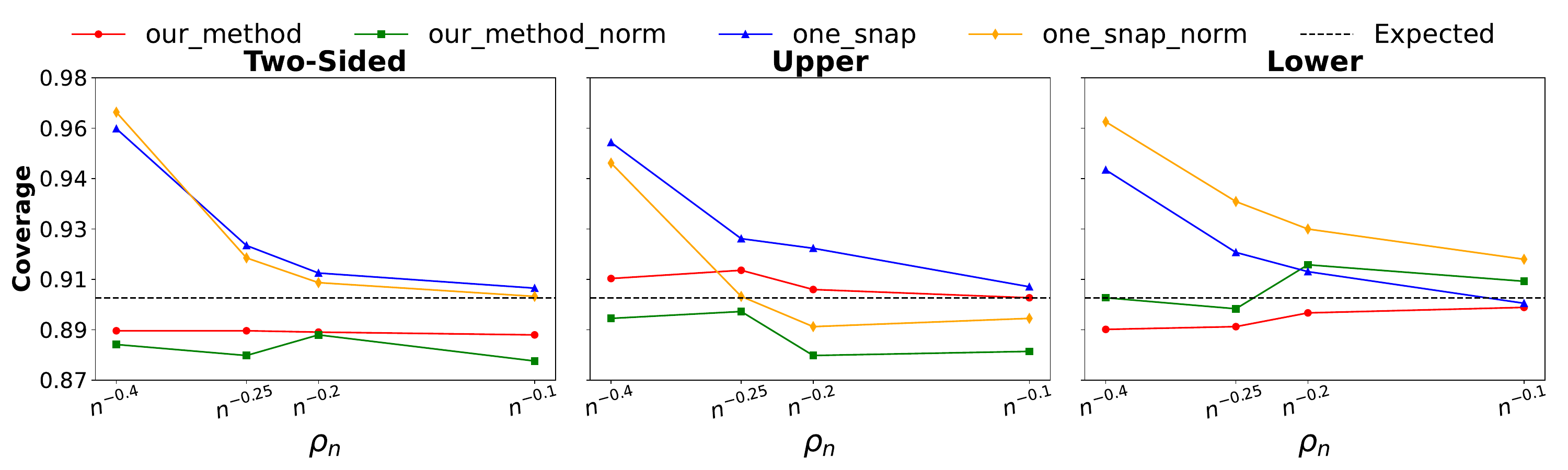}
    \caption{Coverage proportion  across varying network sparsity levels under Setting~\ref{one-poly}, where the target inference time point is $t^*=0$.
    }
    \label{CoverageAcrossRhoUnderOnePolyN200T200Frac10}
\end{figure}
\begin{figure}[!h]
    \centering
    \includegraphics[width=\linewidth]{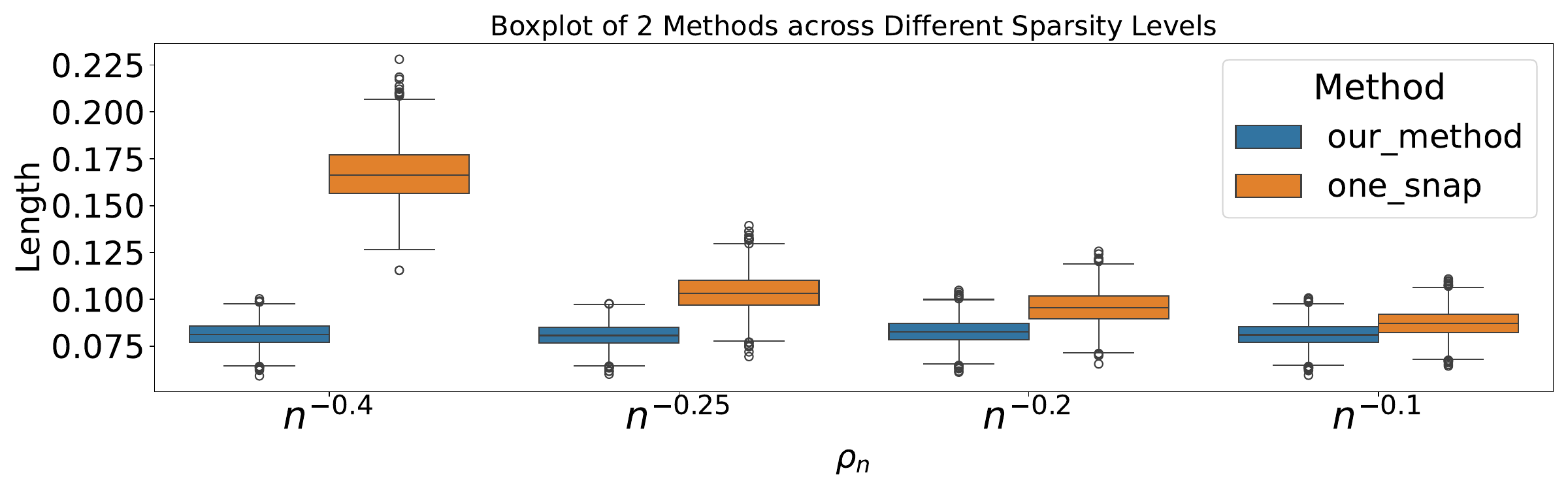}
    \caption{
    Length of CIs (normalized by $\rho_n^3$) across varying network sparsity levels under Setting~\ref{one-poly}, where the target inference time point is $t^*=0$.
    }
    \label{BoxplotOf2MethodsAcrossDifferentRhoNOnePolyN200T200T0Frac10}
\end{figure}
In terms of coverage proportion, as shown in \Cref{CoverageAcrossRhoUnderOnePolyN200T200Frac10}, our methods, based on both the Edgeworth expansion and the standard normal approximation, remain stable across varying sparsity levels and consistently achieve coverage close to the nominal level $90\%$. In contrast, the coverage proportion of the one-snapshot methods deviates substantially from the nominal level as the network becomes sparser.

Regarding CI length, note that the inference target $\mu(t^*)$, its empirical estimator, and its variance estimator all scale with $\rho_n^3$. To make a fair comparison across varying sparsity levels, we therefore normalize CI lengths by $\rho_n^3$. 
As shown in \Cref{BoxplotOf2MethodsAcrossDifferentRhoNOnePolyN200T200T0Frac10}, our methods yield consistently short and stable CIs across all sparsity regimes, whereas the one-snapshot methods lead to longer CIs as $\rho_n$ decreases. These results indicate that by incorporating information from neighboring timestamps, our method effectively mitigates the instability that arises in sparse network regimes.

\subsection{Impact of the Number of Timestamps}
\label{change T}

To further investigate the impact of available temporal information, we fix the sparsity level at $\rho_n = n^{-0.4}$ (an edge density of approximately 12.0\%) and vary the number of timestamps $T \in \{50, 75, 100, 125, 200\}$ under Setting~\ref{one-poly}. 
As shown in \Cref{CoverageAcrossTUnderOnePolyN200RhoN04T0Frac10}, the coverage of our methods improves steadily with larger $T$ and converges toward the nominal level. When $T$ is small, coverage proportion may deviate due to bias introduced by kernel smoothing. In contrast, under this sparse regime, the one-snapshot methods deviate from nominal level and cannot benefit from leveraging neighborbood information with increasing~$T$.
\begin{figure}[!h]
    \centering
    \includegraphics[width=\linewidth]{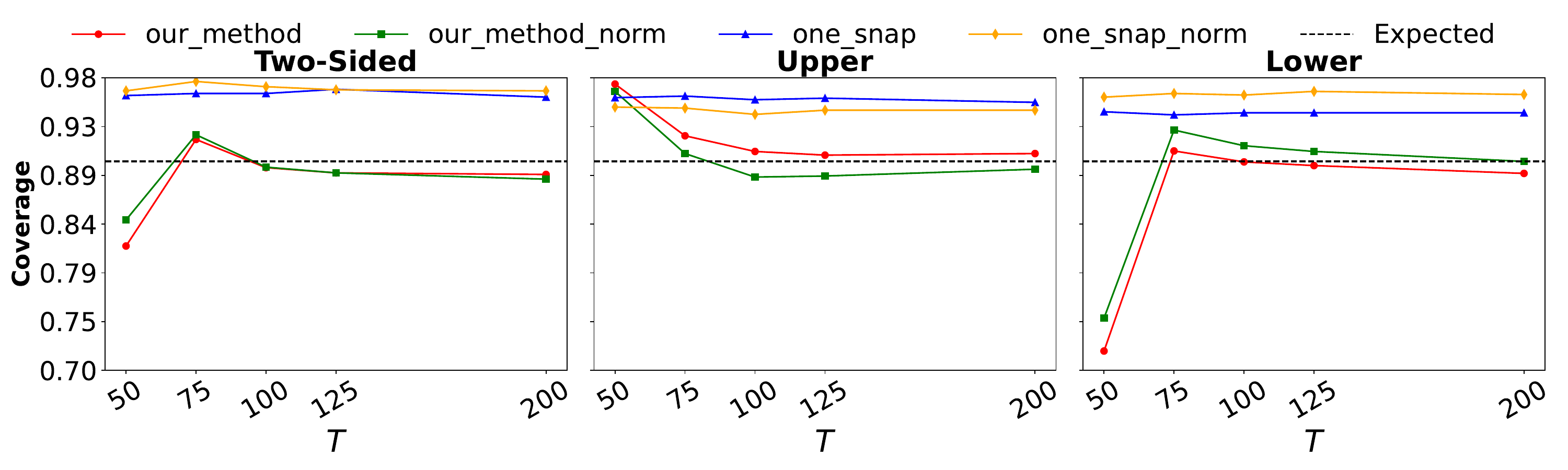}
    \caption{Coverage proportion  across varying number of timestamps under Setting~\ref{one-poly}, where  $\rho_n= n^{-0.4}$ and $t^*=0$.
    }
    \label{CoverageAcrossTUnderOnePolyN200RhoN04T0Frac10}
\end{figure}

\subsection{Impact of the Degree of Balance}
Finally, we assess the performance of our method under varying degrees of balance using the graphon model defined in Setting \ref{balance}.
In this model, the degree of balance evolves over time. When $t=-0.5$, each edge is equally likely to be positive or negative, representing a balance-free scenario. As $t$ increases, edges within the same block are more likely to be positive, while those across blocks are more likely to be negative, which leads to a higher proportion of balanced triangles. Thus, larger $t$ corresponds to stronger degree of balance. 
We vary the inference target $t^*\in\{-0.46, -0.4, -0.3, -0.1, 0.3\}$, for which the proportion of balanced triangles ranges from approximately $60\%$ to $90\%$.
\Cref{CIBandUnderBalanceN200T200Rho025S1Frac10} visualizes the true portion of balanced triangles among all triangles, given by $(1+\mu(t^*)/\rho_n^3)/2$, and the pointwise two-sided $90\%$ CIs obtained by applying the same transformation to the CIs for $\mu(t^*)$ constructed by our method, for one random replication. 
As expected, the proportion of balanced triangles increases monotonically with $t^*$ and remains well within the constructed $90\%$ CIs across all degrees of balance. 

\Cref{CoverageAcrossTUnderBalanceN200T200Frac10Rhon025} further summarizes the coverage proportions of both two-sided and one-sided CIs for all four methods based on $2,000$ replications, and \Cref{BoxplotOf2MethodsAcrossDifferentDegreesOfBalanceUnderBalanceN200T200Frac10Rho025} compares the corresponding two-sided CI lengths across different degrees of balance.
Our method achieves coverage reasonably close to the nominal level while yielding shorter intervals than the one-snapshot baseline.
These results suggest that the proposed inference procedure remains stable across varying degrees of balance.

\begin{figure}[!h]
    \centering
    \includegraphics[width=\linewidth]{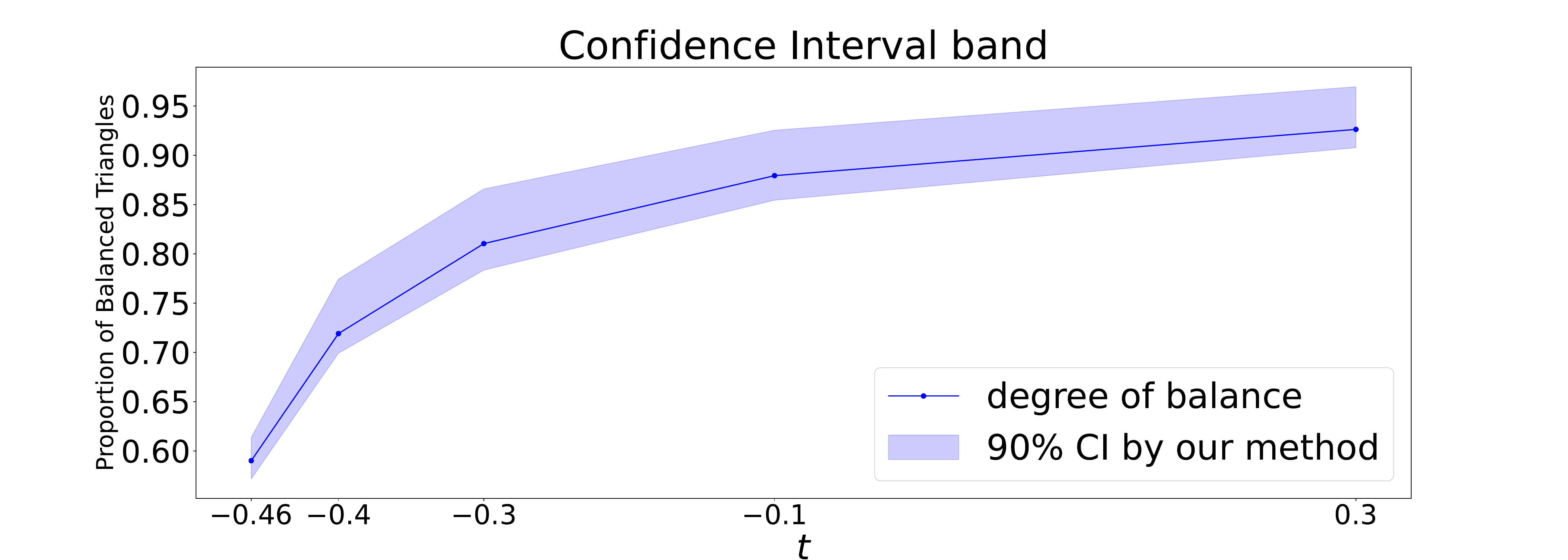}
    \caption{Portion of balanced triangles among all triangles and 90\% two-sided confidence interval constructed by our method under Setting \ref{balance}, where the portion is defined as $(1+\mu(t^*)/\rho_n^3)/2$ and $\mu(t^*)$ is the inference target.}
    \label{CIBandUnderBalanceN200T200Rho025S1Frac10}
\end{figure}

\begin{figure}[!h]
    \centering
    \includegraphics[width=\linewidth]{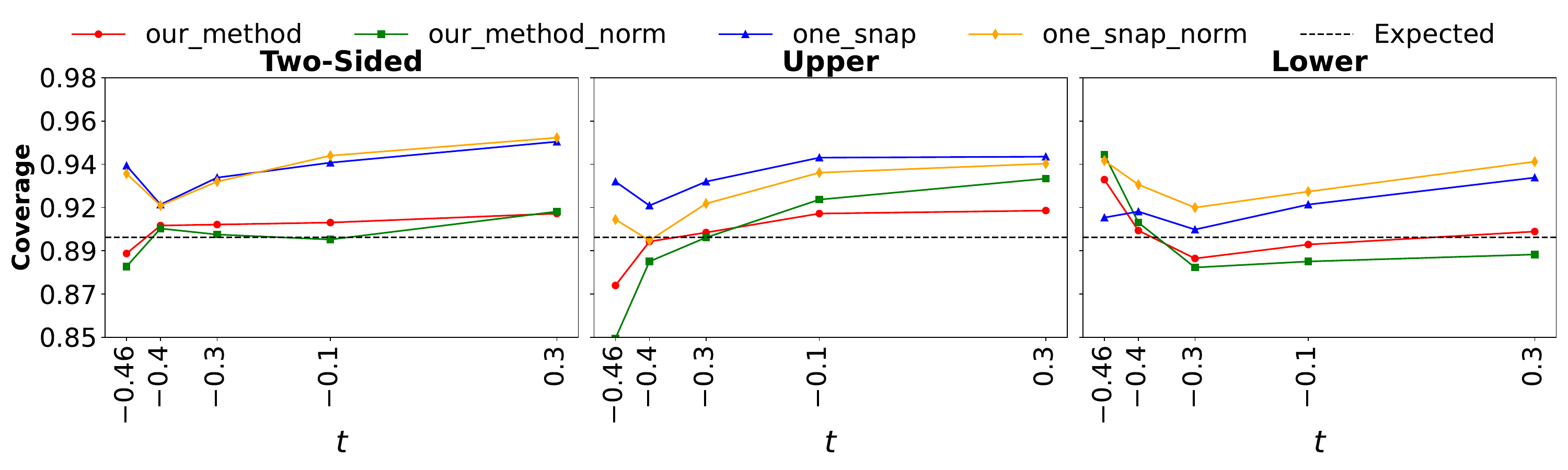}
    \caption{Coverage proportion across varying degrees of balance under Setting \ref{balance}.}
    \label{CoverageAcrossTUnderBalanceN200T200Frac10Rhon025}
\end{figure}

\begin{figure}[!h]
    \centering
    \includegraphics[width=\linewidth]{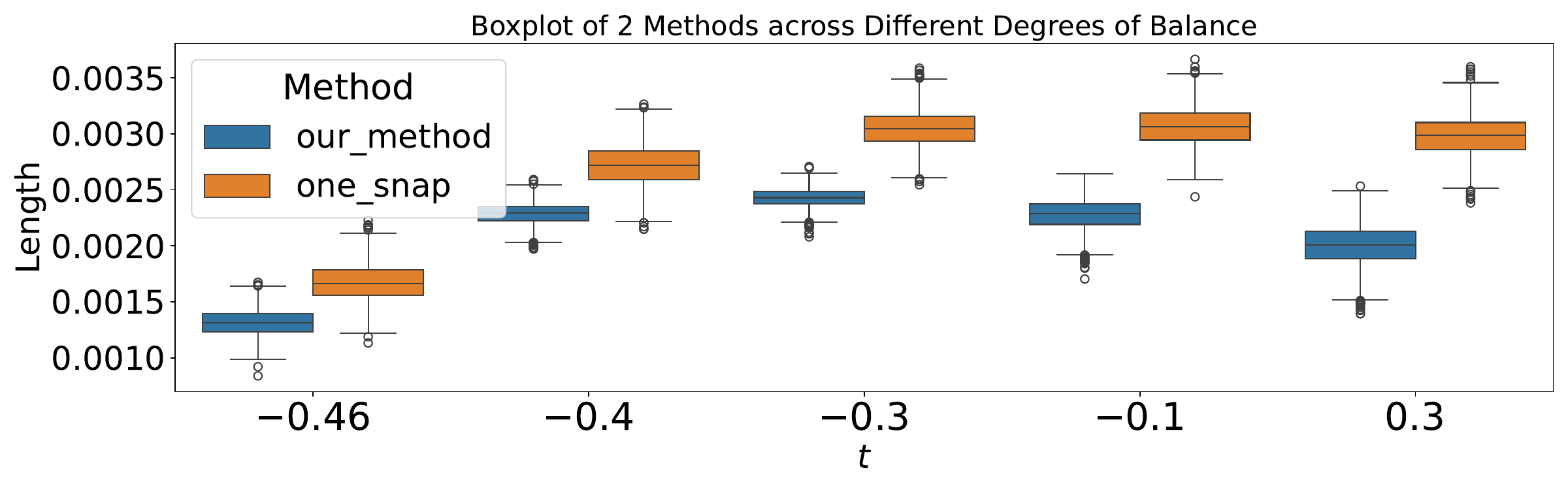}
    \caption{Length of CIs across varying degrees of balance under Setting \ref{balance}.}
    \label{BoxplotOf2MethodsAcrossDifferentDegreesOfBalanceUnderBalanceN200T200Frac10Rho025}
\end{figure}

\section{Application to Dynamic Signed International Relation Network}
\label{Sec:real_data}
In this section, we analyze a real-world dynamic signed network constructed from the Integrated Crisis Early Warning System (ICEWS) dataset~\citep{DVN/28117_2015}. The ICEWS dataset documents political interactions among countries from January 1995 to April 2023 in the form of directed event records. Each record contains a source country, a target country, the date of occurrence, and an \textit{intensity} score that quantifies the strength and sentiment of interaction (e.g., cooperative vs. conflictual) by aggregating weighted counts of relevant news articles within a given month.
To construct an undirected dynamic signed network, we aggregated the intensity scores for each country pair within each month and assigned the sign of their sum as the edge sign. Pairs of countries without any recorded events were treated as having no edge. 
Thus, a positive edge indicates that cooperative events outweighed conflictual events for that country pair in that month, whereas a negative edge indicates that conflictual events dominated.
We restricted our analysis to the period from May 2003 to November 2015 for relatively stable network densities and sufficient node connectivities, and we retained only countries that were involved in interactions in at least $80$ monthly snapshots during the observed period. The final dataset is a dynamic network among $195$ nodes (countries) observed across $150$ monthly snapshots, with network densities ranging approximately from 8\% to 12\%. %

We then apply the proposed method to quantify uncertainty in the degree of structural balance over time and compare it with one-snapshot baseline. \Cref{fig:CIRealDataPeter95} summarizes the 95\% confidence intervals for $\mu(t^*)$ at seven evenly spaced inference time points, and \Cref{fig:CILengthRealDataPeter95} compares the corresponding interval lengths.
Across the seven time points, our method produces intervals that are consistently shorter than, or comparable to, those from the one-snapshot baseline, suggesting that borrowing information from nearby network snapshots improves the precision of inference. 
Compared with the standard normal approximation, our distributional approximation with higher-order correction terms produces confidence intervals that lean toward a stronger degree of balance.

To assess whether the observed signed networks exhibit balance beyond what can be explained by the overall fraction of negative edges, we further compare the estimates with the balance-free null value $\mu_0$ discussed in~\Cref{null_val}. 
Under the balance-free null model~\citep{facchetti2011computing}, the edge signs are generated independently while preserving the observed negative-edge proportion.
As shown in \Cref{fig:CIRealDataPeter95}, at all seven time points, the confidence intervals lie above this null value, and the corresponding two-sided tests based on~\Cref{null_val} yield $p$-values below $0.01$. 
This indicates that the observed networks reflect systematic structural balance beyond what would be expected from edge density and sign proportion~alone.

\begin{figure}[!htbp]
    \centering
    \includegraphics[width=0.83\linewidth]{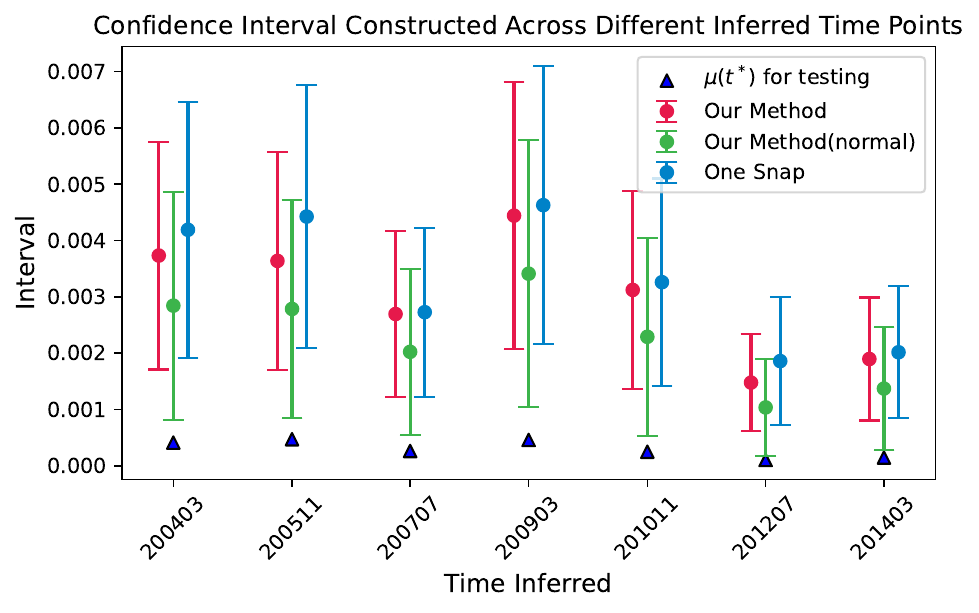}
    \caption{The 95\% confidence intervals for $\mu(t^*)$ at seven inference time points obtained by three methods. The $x$-axis denotes the year and month of each inferred time point (e.g., 200511 corresponds to November 2005). 
    }
    \label{fig:CIRealDataPeter95}
\end{figure}

\begin{figure}[!htbp]
    \centering
    \includegraphics[width=0.83\linewidth]{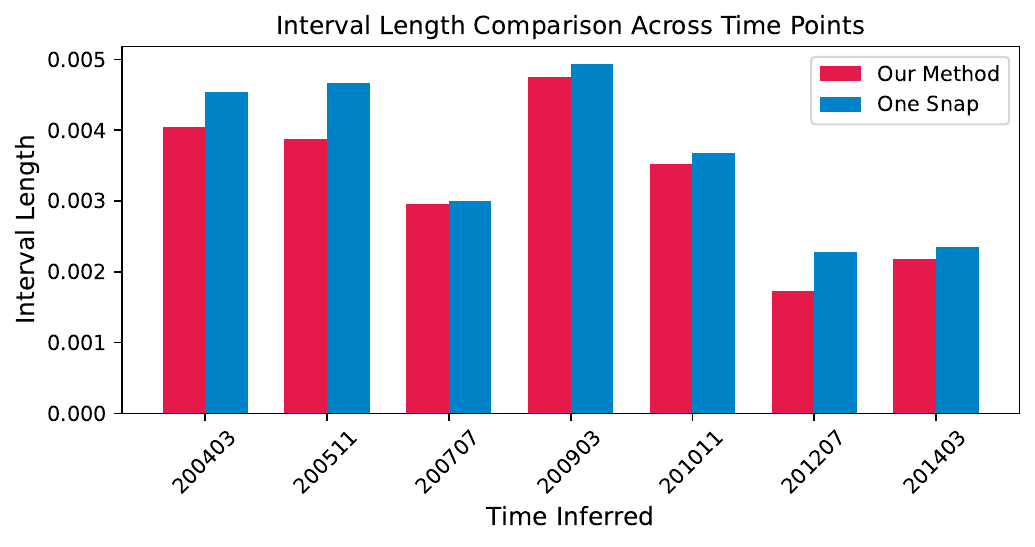}

    \caption{Length comparison of the 95\% confidence intervals for $\mu(t^*)$ at seven inference time points obtained by different methods. The $x$-axis denotes the year and month of each inferred time point (e.g., 200511 corresponds to November 2005). 
    }
    \label{fig:CILengthRealDataPeter95}
\end{figure}

In addition, the temporal pattern in \Cref{fig:CIRealDataPeter95} suggests meaningful variation in global tension structure over the study period. 
The early 2000s show relatively high levels of balance, consistent with a period in which multilateral cooperation and institutional coordination were prominent. 
The degree of balance declines modestly around 2007, possibly reflecting a more mixed or fragmented pattern of international relations during that period. 
The estimate at 2009 is higher than that at 2007, which may reflect renewed coordination during the global financial crisis.
The later time points, especially after 2012, show weaker balance, consistent with increasing fragmentation and a more multipolar geopolitical environment.

\section{Discussion}
In this work, we develop a nonparametric kernel-smoothed estimator for the degree
of balance in dynamic signed networks. 
The estimator first smooths edge observations across nearby snapshots and then constructs the empirical network moment. 
This edge-level smoothing effectively reduces observation noise.
Compared with one-snapshot methods, this kernel-smoothed approach relaxes the sparsity
requirements for valid inference. 
Furthermore, our established approximation error bounds for the Edgeworth expansion of the studentized statistic suggests the trade-off inherent to kernel methods: local averaging reduces stochastic error, while kernel smoothing introduces smoothing bias and discretization error.

Our current framework still assumes that the global scale of network sparsity, characterized by $\rho_n$, is invariant over time. 
When the sparsity level varies substantially over time, however, additional normalization is required to ensure that the degree of balance remains comparable across time points. 
In particular, the scale of the inference target  $\mu(t^*)$ depends on the sparsity level and, under our formulation, varies proportionally to $\rho_n^3(t)$, where $\rho_n(t)$ denotes the time-specific sparsity level. 
To address this issue, one can rescale the inference target by replacing $A_{ij}(t)$ with $A_{ij}(t)\cdot(\rho_n(t))^{-1}$ for $1\leq i, j\leq n$. This rescaling yields a normalized balance measure that is comparable across time.  
In practice, this requires modifying the inference procedure in Section~\ref{procedure} by replacing $\rho_n(t)$ with an empirical estimate of the network density at each time point. 
This adjustment introduces additional randomness into both the point estimator and the corresponding variance estimator. 
For static networks, a closely related issue arises in two-sample network inference, where motif counts from each network are rescaled by their empirical densities to adjust for sparsity heterogeneity before comparing structural patterns~\citep{Shao05082025}. A full theoretical investigation of this rescaling approach in dynamic signed networks is an important direction for future work.

\section*{Acknowledgements} Tang was partially supported by NSF DMS-2412853. 
Zhang was supported by NSF DMS-2311109.

\bibliography{main.bib}
\newpage

\appendix
\begin{center}
{\large\bf SUPPLEMENTAL MATERIAL}
\end{center}
\allowdisplaybreaks

\setcounter{equation}{0}
\setcounter{figure}{0}
\setcounter{table}{0}
\setcounter{page}{1}
\setcounter{section}{0}
\renewcommand{\theequation}{S\arabic{equation}}
\renewcommand{\thefigure}{S\arabic{figure}}

This supplemental material for ``Inference for Balance in Dynamic Signed Networks'' is structured as follows. 
\vspace{-7mm}
\addtocontents{toc}{\protect\setcounter{tocdepth}{1}}
\renewcommand{\contentsname}{}
\tableofcontents

\section{Notation and Preliminaries}

We first introduce the notation and definitions that will be used throughout the analysis. 
We fix a target time point $t$.  All stochastic orders below are uniform in this fixed $t$. 

\paragraph{Notation.} We use the standard asymptotic notation $O(\cdot)$, $o(\cdot)$, $\Omega(\cdot)$, and $\omega(\cdot)$, together with their probabilistic counterparts $O_p(\cdot)$ and $o_p(\cdot)$. For a random variable $Z$ and a deterministic sequence $\{\alpha_n\}$, we write $Z=\tilde O_{p,1}(\alpha_n)$ if $\bbP(|Z|\ge C\alpha_n)=O(n^{-1})$ for some constant $C>0$, and $Z=\tilde O_{p,2}(\alpha_n)$ if $\bbP(|Z|\ge C\alpha_n)=O(n^{-2})$ for some constant $C>0$.
We write $Z_1 \gtrsim Z_2$ if there exists a constant $C>0$ such that $Z_1 \geq C Z_2$ for all sufficiently large $n$. We write $Z_1 \asymp Z_2$ if $Z_1 \gtrsim Z_2$ and $Z_2 \gtrsim Z_1$.

\paragraph{Model Setup and Definitions.}
Write $X=(X_1,\ldots,X_n)$ and use $\bbE[\cdot\mid X]$ for conditional expectation given all latent variables.  Conditional on $X$, the variables $\{A_{ij}(t_\ell):1\le i<j\le n,1\le \ell\le T\}$ are independent across different edge pairs and observation times, and
\[
A_{ij}(s)=
\begin{cases}
0,&\text{with probability } a_{ij}(s)=1-\rho_n F(X_i,X_j,s),\\
-1,&\text{with probability } b_{ij}(s)=\rho_n F(X_i,X_j,s)G(X_i,X_j,s),\\
1,&\text{with probability } c_{ij}(s)=\rho_n F(X_i,X_j,s)\{1-G(X_i,X_j,s)\}.
\end{cases}
\]
We have $0\le G\le1$, $0\le \rho_n F\le1$.  Assumption~\ref{assump: smooth graphon} imposes H\"older smoothness of $F$ and $G$; since products of bounded H\"older functions remain H\"older of the same order, the same order of smoothness is available for $FG$ and for $F(1-2G)$. 
For a node pair $(i,j)$, define
\[
W_{ij}(t)=\bbE[A_{ij}(t)\mid X_i,X_j],\qquad
\eta_{ij}(t)=A_{ij}(t)-W_{ij}(t).
\]
Let \(w_\ell(t)=K_h(t-t_\ell)/\sum_{m=1}^T K_h(t-t_m)\) and define $\tilde A_{ij}(t)=\sum_{\ell=1}^T w_\ell(t) A_{ij}(t_\ell)$,
\[
\tilde W_{ij}(t)=\bbE[\tilde A_{ij}(t)\mid X_i,X_j] =\sum_{\ell=1}^T w_\ell(t) W_{ij}(t_\ell) ,\qquad
\tilde\eta_{ij}(t)=\tilde A_{ij}(t)-\tilde W_{ij}(t).
\]
For simplicity of exposition, we adopt the Gaussian kernel in all our proofs. Proofs for other kernels satisfying properties (K1)-(K4) can be derived analogously.
Together with Assumption~\ref{assump: time obs} and the interior condition on $t$, by \Cref{kappa j}, we have the following standard results for the kernel weights that will be used throughout the proof:
\[
\sum_{\ell=1}^T K_h(t-t_\ell)\asymp T,
\qquad
\sum_{\ell=1}^T K_h^2(t-t_\ell)\asymp T/h, 
\qquad
\max_\ell |w_\ell(t)|=O((Th)^{-1}).
\]

\begin{lemma}[Basic moment bounds]
\label{basic moment}
Under the probability model and Assumption~\ref{assump: time obs}, for every fixed $t$ and $i<j$, almost surely,
\[
W_{ij}(t)=O(\rho_n),\qquad |\eta_{ij}(t)|=O(1),\qquad
\bbE[|\eta_{ij}(t)|^r\mid X_i,X_j]=O(\rho_n),\qquad r=1,2,3.
\]
For the smoothed quantities,
\[
\tilde W_{ij}(t)=O(\rho_n),\qquad |\tilde\eta_{ij}(t)|=O(1),\qquad
\bbE [|\tilde\eta_{ij}(t)|\mid X_i,X_j ]=O(\rho_n),
\]
and
\[
\bbE [\tilde\eta_{ij}^2(t)\mid X_i,X_j ]=O\{\rho_n(Th)^{-1}\}.
\]
\end{lemma}

\begin{proof}[\textbf{Proof of \Cref{basic moment}}]
Because $F$ is bounded and $0\le G\le1$, both nonzero-edge probabilities satisfy
\[
b_{ij}(t)=\rho_nF(X_i,X_j,t)G(X_i,X_j,t)=O(\rho_n),\quad
c_{ij}(t)=\rho_nF(X_i,X_j,t)\{1-G(X_i,X_j,t)\}=O(\rho_n).
\]
Moreover,
\[
W_{ij}(t)=\bbE\{A_{ij}(t)\mid X_i,X_j\}=c_{ij}(t)-b_{ij}(t)=\rho_nF(X_i,X_j,t)\{1-2G(X_i,X_j,t)\}=O(\rho_n).
\]
Since $A_{ij}(t)\in\{-1,0,1\}$ and $|W_{ij}(t)|\le1$, we have $|\eta_{ij}(t)|\le2$.  For $r=1,2,3$,
\[
\begin{aligned}
\bbE[|\eta_{ij}(t)|^r\mid X_i,X_j]
&=a_{ij}(t)|W_{ij}(t)|^r+b_{ij}(t)|-1-W_{ij}(t)|^r+c_{ij}(t)|1-W_{ij}(t)|^r\\
&\le |W_{ij}(t)|^r+2^r\{b_{ij}(t)+c_{ij}(t)\}=O(\rho_n).
\end{aligned}
\]
For the smoothed mean,
\[
|\tilde W_{ij}(t)|\le \sum_{\ell=1}^T |w_\ell(t)|\, |W_{ij}(t_\ell)|=O(\rho_n).
\]
Likewise, $|\tilde\eta_{ij}(t)|\le \sum_\ell |w_\ell(t)| |\eta_{ij}(t_\ell)|=O(1)$ and
\[
\bbE [|\tilde\eta_{ij}(t)|\mid X_i,X_j ]
\le \sum_{\ell=1}^T |w_\ell(t)|\bbE [|\eta_{ij}(t_\ell)|\mid X_i,X_j ]=O(\rho_n).
\]
Finally, by conditional independence across observation times,
\[
\begin{aligned}
\bbE [(\tilde\eta_{ij}(t))^2\mid X_i,X_j ]
&=var\{\tilde A_{ij}(t)\mid X_i,X_j\}\\
&=\sum_{\ell=1}^T w_\ell^2(t)var\{A_{ij}(t_\ell)\mid X_i,X_j\}\\
&\le C\rho_n\frac{\sum_{\ell=1}^T K_h^2(t-t_\ell)}{\left(\sum_{\ell=1}^T K_h(t-t_\ell)\right)^2}
=O(\rho_n(Th)^{-1}).
\end{aligned}
\]
This proves the lemma.
\end{proof}

\section{Variance Decomposition in \Cref{subsec: variance_tech}}
\label{Proof of variance prop}

We first give a sketch of the analysis strategy. To identify the leading term in the variance of $\tilde U_{nh}(t)-\mu(t)$, we first decompose the noiseless U-statistic $U_{nh}(t)$ into its Hoeffding projections, which represent the variation from latent variables. We then decompose the difference $\tilde U_{nh}(t)-U_{nh}(t)$ into terms of different degrees in the observation noise, which represent the variation from observation noise. Finally, we combine these decompositions to identify the leading source of variance and derive the variance estimator $\hat{S}_{nh}^2$ in \eqref{denominator}.

\paragraph{Noiseless U-statistic and projections.}
Define
\[
U_{nh}(t)=\bbE[\tilde U_{nh}(t)\mid X ]={n\choose 3}^{-1}\sum_{i<j<k}\tilde W_{ij}(t)\tilde W_{jk}(t)\tilde W_{ki}(t),
\qquad
\mu_h(t)=\bbE [U_{nh}(t)].
\]
The first two Hoeffding projections are
$$g_{1h}(x, t)=\bbE\left[\tilde W_{12}(t)\tilde W_{23}(t)\tilde W_{31}(t)|X_1=x\right]-\mu_{h}(t),$$ and $$g_{2h}(x_1, x_2, t)=\bbE\left[\tilde W_{12}(t)\tilde W_{23}(t)\tilde W_{31}(t)|X_1=x_1, X_2=x_2\right]-\mu_{h}(t)-g_{1h}(x_1,t)-g_{1h}(x_2,t).$$ 
Let $g_{3h}$ denote the fully degenerate third projection.  Then Hoeffding's decomposition gives
\begin{equation}
\label{eq:hoeffding-decomp-revised}
U_{nh}(t)-\mu_h(t)=Y_1(t)+Y_2(t)+Y_3(t),
\end{equation}
where
\[
Y_1(t)=\frac{3}{n}\sum_{i=1}^n g_{1h}(X_i,t),\qquad
Y_2(t)=\frac{6}{n(n-1)}\sum_{i<j}g_{2h}(X_i,X_j,t),
\]
and
\[
Y_3(t)=\frac{6}{n(n-1)(n-2)}\sum_{i<j<k}g_{3h}(X_i,X_j,X_k,t).
\]
By the basic moment bounds in \Cref{basic moment} and standard concentration for bounded degenerate U-statistics yields the bounds for each projection term. We state the following lemma, the proof of which is given in \Cref{proof variance Hoeffding}.

\begin{lemma}[Latent-variable variation]
\label{variance Hoeffding}
Under the probability model and \Cref{assump: time obs,assump: smooth graphon,non-u,assump: sparsity}, we have
\[
Y_1(t)=\tilde O_{p,2}(\rho_n^3n^{-1/2}\log^{1/2}n),\quad
Y_2(t)=\tilde O_{p,2}(\rho_n^3n^{-1}\log n),\quad
Y_3(t)=\tilde O_{p,2}(\rho_n^3n^{-3/2}\log^{3/2}n),
\]
and the variances and covariances satisfy
\[
    var(Y_1(t))=O(\rho_n^6n^{-1}) ,\quad
    var(Y_2(t))=O(\rho_n^6n^{-2}),\quad
    var(Y_3(t))=O(\rho_n^6n^{-3}),
\]
with all relevant cross-covariances vanishing by orthogonality of Hoeffding projections.
In particular, if \(\xi_{1h}^2(t)\asymp \rho_n^6\), then
\[
    var(Y_1(t))=\frac{9}{n}\xi_{1h}^2(t)\asymp \rho_n^6 n^{-1}.
\]
\end{lemma}

\paragraph{Observation-noise decomposition.}
Writing $\tilde A_{ij}(t)=\tilde W_{ij}(t)+\tilde\eta_{ij}(t)$ and collecting terms according to their degree in $\tilde\eta$, we obtain
\begin{align}
\label{Decomp_QR}
\tilde U_{nh}(t)-U_{nh}(t)
&=Q(t)+R_q(t)+R_c(t),\\
Q(t)
&={n\choose2}^{-1}\sum_{i<j}\left\{\frac{3}{n-2}\sum_{k\ne i,j}\tilde W_{ik}(t)\tilde W_{jk}(t)\right\}\tilde\eta_{ij}(t),\nonumber\\
R_q(t)
&={n\choose3}^{-1}\sum_{i<j}\sum_{k\ne i,j}\tilde W_{ij}(t)\tilde\eta_{ik}(t)\tilde\eta_{jk}(t),\nonumber\\
R_c(t)
&={n\choose3}^{-1}\sum_{i<j<k}\tilde\eta_{ij}(t)\tilde\eta_{jk}(t)\tilde\eta_{ki}(t).\nonumber
\end{align}
Here $Q(t)$ is the linear term in observation noises, $R_q(t)$ is the quadratic term, and $R_c(t)$ is the cubic term. The bounds for these terms are stated in \Cref{variance linear}, the proof of which is given in \Cref{proof variance linear}.

\begin{lemma}[Observation noise decomposition]
\label{variance linear}
Under the probability model and \Cref{assump: time obs,assump: smooth graphon,non-u,assump: sparsity}, we have
\[
Q(t)=\tilde O_{p,2}(n^{-1}(Th)^{-1/2}\rho_n^{5/2}\sqrt{\log n}),
\]
\[
R_q(t)=\tilde O_{p,2}(n^{-3/2}\rho_n^2(Th)^{-1}\sqrt{\log n}),
\]
\[
R_c(t)=\tilde O_{p,2}(n^{-3/2}\rho_n^{3/2}(Th)^{-3/2}\sqrt{\log n})+\tilde O_{p,2}(n^{-2}\rho_n^2(Th)^{-1}\log n).
\]
Moreover, the variances satisfy
\[
var(Q(t))=O(n^{-2}(Th)^{-1}\rho_n^5),\ 
var(R_q(t))=O(n^{-3}(Th)^{-2}\rho_n^4),\ 
var(R_c(t))=O(n^{-3}(Th)^{-3}\rho_n^3),
\]
and the relevant cross-covariances vanish by conditional independence and centering.
\end{lemma}

\paragraph{Variance decomposition of $\tilde U_{nh}(t)-\mu(t)$.}
Combining \eqref{eq:hoeffding-decomp-revised} and \eqref{Decomp_QR}, we obtain
\begin{equation}
\label{eq:master-decomp-revised}
\tilde U_{nh}(t)-\mu(t)
=Q(t)+R_q(t)+R_c(t)+Y_1(t)+Y_2(t)+Y_3(t)+\{\mu_h(t)-\mu(t)\}.
\end{equation}
The last term is deterministic bias and therefore does not contribute to the variance. We establish that the leading source of variance comes from $Y_1(t)$ in \Cref{variance prop}.
\begin{proposition}
\label{variance prop}
Suppose \Cref{assump: time obs,assump: smooth graphon,non-u,assump: sparsity} hold.  Then
\[
    var\{\tilde U_{nh}(t)-\mu(t)\}
    =
    var(Y_1(t))
    +
    var\{\tilde U_{nh}(t)-U_{nh}(t)\}
    +
    O(\rho_n^6 n^{-2}),
\]
with
\[
    var(Y_1(t))=\frac{9}{n}\xi_{1h}^2(t)
    \asymp \rho_n^6 n^{-1},
\]
and
\[
\begin{aligned}
    var(\tilde U_{nh}(t)-U_{nh}(t))
    &=
    O\left(
        n^{-2}(Th)^{-1}\rho_n^5
        +n^{-3}(Th)^{-2}\rho_n^4
        +n^{-3}(Th)^{-3}\rho_n^3
    \right).
\end{aligned}
\]
Consequently,
\[
    var(\tilde U_{nh}(t)-\mu(t))
    =
    \frac{9}{n}\xi_{1h}^2(t)
    \{1+o(1)\}.
\]
\end{proposition}
The above proposition shows that the leading source of variation comes from $Y_1(t)$, which motivates the variance estimator~$\hat{S}_{nh}^2$ in \eqref{denominator}. 

Before presenting the proof of \Cref{variance prop}, we first establish the non-degeneracy of $g_{1h}$.
Assumption~\ref{non-u} states that $\xi_1^2(t):=var\{g_1(X_1,t)\}\ge C\rho_n^6$. Under the smoothness and observation conditions in Assumptions~\ref{assump: time obs} and~\ref{assump: smooth graphon}, the smoothed projection $g_{1h}$ inherits the non-degeneracy of $g_1$ up to a small approximation error, as shown in \Cref{xi}. The proof of \Cref{xi} is given in \Cref{Proof of xi}.

\begin{lemma}[Nondegeneracy from $g_1$ to $g_{1h}$]
\label{xi}
Under the probability model and \Cref{assump: time obs,assump: smooth graphon,non-u},
we have
\[
    \xi_{1h}^2(t):=var\{g_{1h}(X_1,t)\}\gtrsim \rho_n^6,
\]
provided that
\(
    h^\nu+(Th)^{-1}=o(1).
\)
\end{lemma}

\begin{proof}[\textbf{Proof of \Cref{variance prop}}]
Decompose
\[
    \tilde U_{nh}(t)-\mu(t)
    =
    \{\tilde U_{nh}(t)-U_{nh}(t)\}
    +
    \{U_{nh}(t)-\mu_h(t)\}
    +
    \{\mu_h(t)-\mu(t)\}.
\]
The last term is deterministic and therefore does not contribute to the
variance. Also,
\[
    \bbE (\tilde U_{nh}(t)-U_{nh}(t)\mid X_1,\ldots,X_n)=0,
\]
whereas \(U_{nh}(t)-\mu_h(t)\) is measurable with respect to
\((X_1,\ldots,X_n)\). Hence
\[
    cov (\tilde U_{nh}(t)-U_{nh}(t),\,U_{nh}(t)-\mu_h(t))=0,
\]
and therefore
\[
    var(\tilde U_{nh}(t)-\mu(t))
    =
    var(\tilde U_{nh}(t)-U_{nh}(t))
    +
    var(U_{nh}(t)-\mu_h(t)).
\]
By the orthogonality of Hoeffding projections in \eqref{eq:hoeffding-decomp-revised},
\[
    var(U_{nh}(t)-\mu_h(t))
    =
    var(Y_1(t))+var(Y_2(t))+var(Y_3(t)).
\]
By \Cref{variance Hoeffding},
it follows that
\[
    var (U_{nh}(t)-\mu_h(t))
    =
    var(Y_1)+O(\rho_n^6 n^{-2}).
\]
Since $Y_1(t)=3n^{-1}\sum_i g_{1h}(X_i,t)$ is a sum of iid mean-zero terms,
\[
var(Y_1(t))=\frac{9}{n}var(g_{1h}(X_1,t))=\frac{9\xi_{1h}^2(t)}{n}.
\]
Lemma~\ref{xi} gives the lower bound $\xi_{1h}^2(t)\gtrsim\rho_n^6$, so $var(Y_1(t))\gtrsim \rho_n^6n^{-1}$.

Moreover, by \Cref{variance linear}, it follows that
\[
\begin{aligned}
    var(\tilde U_{nh}(t)-U_{nh}(t))
    &=
    O\left(
        n^{-2}(Th)^{-1}\rho_n^5
        +n^{-3}(Th)^{-2}\rho_n^4
        +n^{-3}(Th)^{-3}\rho_n^3
    \right).
\end{aligned}
\]
Finally, it is sufficient to compare the leading variance term $var(Y_1(t))\asymp \rho_n^6 n^{-1}$ with all remaining terms. The ratios of the three observation-noise variances to $\rho_n^6n^{-1}$ are respectively
\[
O((nTh\rho_n)^{-1}),\qquad
O(n^{-2}(Th)^{-2}\rho_n^{-2}),\qquad
O(n^{-2}(Th)^{-3}\rho_n^{-3}).
\]
Under Assumption~\ref{assump: sparsity}, each ratio is $o(1)$, and therefore
\(
    var(\tilde U_{nh}(t)-\mu(t))
    =
    \frac{9}{n}\xi_{1h}^2(t)
    \{1+o(1)\}.
\)
\end{proof}

\subsection{Proof of \Cref{variance Hoeffding}}
\label{proof variance Hoeffding}
\begin{proof}
By \Cref{basic moment}, we have
\(
    |\tilde W_{ij}(t)\tilde W_{jk}(t)\tilde W_{ki}(t)|\le C\rho_n^3.
\)
As a result, the Hoeffding projection kernels $g_{1h},g_{2h}$ and $g_{3h}$, which are finite linear combinations of conditional expectations of $\tilde W_{ij}(t)\tilde W_{jk}(t)\tilde W_{ki}(t)$ and its mean, are also uniformly bounded by $C\rho_n^3$. 

For the first-order projection
\(
    Y_1(t)=\frac{3}{n}\sum_{i=1}^n g_{1h}(X_i,t),
\)
which is an average of i.i.d.\ centered bounded random variables, Bernstein's inequality gives
\[
    Y_1(t)=\tilde O_{p,2}\bigl(\rho_n^3n^{-1/2}\log^{1/2}n\bigr).
\]
For the second- and third-order projections, by applying the multivariate version of Bernstein's inequality (Theorem 1 in \cite{10.1214/EJP.v12-430}), we have
\[
    Y_2(t)=\tilde O_{p,2}\bigl(\rho_n^3n^{-1}\log n\bigr),
    \qquad
    Y_3(t)=\tilde O_{p,2}\bigl(\rho_n^3n^{-3/2}\log^{3/2}n\bigr).
\]
We next compute the variances. Since $Y_1(t)$ is a sum of iid centered terms,
\[
    var(Y_1(t))
    =\frac{9}{n}var(g_{1h}(X_1,t))
    =\frac{9}{n}\xi_{1h}^2(t).
\]
The bound $\|g_{1h}\|_\infty\le C\rho_n^3$ implies
$\xi_{1h}^2(t)=O(\rho_n^6)$, and hence
\(
    var(Y_1(t))=O(\rho_n^6n^{-1}).
\)
Together with the nondegeneracy property in \Cref{xi}, this yields
$var(Y_1(t))\asymp \rho_n^6n^{-1}$.

For the second projection, by its definition, cross-covariances vanish
unless the two index pairs coincide. Consequently,
\[
\begin{aligned}
    var(Y_2(t))
    &=\left\{\frac{6}{n(n-1)}\right\}^2
      \sum_{i<j} \bbE(g_{2h}^2(X_i,X_j,t))  \\
    &=O(\rho_n^6n^{-2}).
\end{aligned}
\]
The same argument for the third-order projection gives
$var(Y_3(t))=O(\rho_n^6n^{-3})$. 
Finally, the Hoeffding projections are mutually
orthogonal, so the covariance between any of $Y_1(t)$, $Y_2(t)$, and $Y_3(t)$ is zero. 
\end{proof}

\subsection{Proof of \Cref{variance linear}}
\label{proof variance linear}

We first state the concentration bounds for the quadratic and cubic observation-noise terms in the following lemma, the proof of which is given in \Cref{quadratic R and cubic R}.

\begin{lemma}[Tail bounds for $R_q(t)$ and $R_c(t)$]
\label{lemma: quadratic R and cubic R}
Assume the probability model and \Cref{assump: time obs,assump: smooth graphon,non-u,assump: sparsity}, and the
variance bounds
\[
    var(R_q(t))=O\{n^{-3}(Th)^{-2}\rho_n^4\},
    \qquad
    var(R_c(t))=O\{n^{-3}(Th)^{-3}\rho_n^3\}.
\]
Then
\[
    R_q(t)
    =\tilde O_{p,2}\bigl(n^{-3/2}\rho_n^2(Th)^{-1}\sqrt{\log n}\bigr),
\]
and
\[
    R_c(t)
    =\tilde O_{p,2}\bigl(n^{-3/2}\rho_n^{3/2}(Th)^{-3/2}\sqrt{\log n}\bigr)
    +\tilde O_{p,2}\bigl(n^{-2}\rho_n^2(Th)^{-1}\log n\bigr).
\]
\end{lemma}

\bigskip
\noindent
We now prove \Cref{variance linear} by applying the Bernstein inequality to $Q(t)$ and \Cref{lemma: quadratic R and cubic R} to $R_q(t)$ and $R_c(t)$.
\begin{proof}[\textbf{Proof of \Cref{variance linear}}]

Write
$\tilde A_{ij}(t)=\tilde W_{ij}(t)+\tilde\eta_{ij}(t)$. 
Collecting terms according to their degree in $\tilde\eta$, we obtain
the decomposition in \eqref{Decomp_QR}.
\paragraph{Linear term.} Define
\[
    \theta_{ij}(t)=\frac{3}{n-2}\sum_{k\ne i,j}\tilde W_{ik}(t)\tilde W_{jk}(t),
\]
then 
\(
Q(t)={n\choose2}^{-1}\sum_{i<j}\theta_{ij}(t)\tilde\eta_{ij}(t).
\)
Conditional on $X$, the summands are independent and centered.
Since
$\bbE\{Q(t)\mid X\}=0$,
\[
    var\{Q(t)\}=\bbE\bigl[var\{Q(t)\mid X\}\bigr].
\]
Using the facts $|\theta_{ij}(t)|\le C\rho_n^2$ and
$var\{\tilde\eta_{ij}(t)\mid X\}=O\{\rho_n(Th)^{-1}\}$ from \Cref{basic moment},
\[
\begin{aligned}
    var\{Q(t)\mid X\}
    &\le {n\choose2}^{-2}\sum_{i<j}
        \theta_{ij}^2(t)var\{\tilde\eta_{ij}(t)\mid X\}   =O\{n^{-2}\rho_n^5(Th)^{-1}\}.
\end{aligned}
\]
Thus
\(
    var\{Q(t)\}=O\{n^{-2}\rho_n^5(Th)^{-1}\}.
\)

Moreover, by plugging in the definition of $\tilde\eta_{ij}(t)$, we obtain
$$Q(t)=\bigg(\binom{n}{2}\sum_{\ell=1}^TK_h(t-t_\ell)\bigg)^{-1}\sum_{i<j}\sum_{\ell=1}^TK_h(t-t_\ell)\theta_{ij}(t)\eta_{ij}(t_\ell).$$
Since, conditional on $X$, the variables
$\{\eta_{ij}(t_\ell):1\le i<j\le n,1\le \ell\le T\}$ are independent
and centered, the conditional variance is bounded by
\[
\begin{aligned}
    \sum_{i<j}\sum_{\ell=1}^T
    K_h^2(t-t_\ell)\theta_{ij}^2(t)
    \bbE\{\eta_{ij}^2(t_\ell)\mid X\}
    &\le C n^2\rho_n^5\sum_{\ell=1}^T K_h^2(t-t_\ell)  =O(n^2\rho_n^5T h^{-1}).
\end{aligned}
\]
Bernstein's inequality therefore gives
\[
    Q(t)
    =\tilde O_{p,2}\bigl(n^{-1}(Th)^{-1/2}\rho_n^{5/2}\sqrt{\log n}\bigr)
     +\tilde O_{p,2}\bigl(n^{-2}(Th)^{-1}\rho_n^2\log n\bigr).
\]
Under \Cref{assump: sparsity}, the second term is dominated by the
first. Hence
\[
    Q(t)=\tilde O_{p,2}\bigl(n^{-1}(Th)^{-1/2}\rho_n^{5/2}\sqrt{\log n}\bigr).
\]

\paragraph{Quadratic and cubic terms.} Conditional on $X$, the products appearing in $R_q(t)$ are centered. Cross
covariances vanish unless the same two noise edges appear in both products; if
two products share only one edge, the remaining centered independent factor
forces the covariance to be zero. Therefore
\[
\begin{aligned}
    var\{R_q(t)\mid X\}
    &\le {n\choose3}^{-2}
        \sum_{i<j}\sum_{k\ne i,j}
        \tilde W_{ij}^2(t)
        \bbE\{\tilde\eta_{ik}^2(t)\mid X\}
        \bbE\{\tilde\eta_{jk}^2(t)\mid X\}  \\
    &=O\{n^{-3}\rho_n^4(Th)^{-2}\},
\end{aligned}
\]
where we used the facts $\bbE[\tilde\eta_{ij}^2(t)]=O(\rho_n(Th)^{-1})$ and $\vert \tilde W_{ij}(t)\vert= O(\rho_n)$ from Lemma \ref{basic moment}.
Since $\bbE\{R_q(t)\mid X\}=0$, the law of total variance gives
\[
    var\{R_q(t)\}=O\{n^{-3}\rho_n^4(Th)^{-2}\}.
\]
Similarly,
\[
\begin{aligned}
    var\{R_c(t)\mid X\}
    &\le {n\choose3}^{-2}\sum_{i<j<k}
        \prod_{e\in\{ij,jk,ki\}}
        \bbE\{\tilde\eta_e^2(t)\mid X\} \\
    &=O\{n^{-3}\rho_n^3(Th)^{-3}\},
\end{aligned}
\]
and hence
\[
    var\{R_c(t)\}=O\{n^{-3}\rho_n^3(Th)^{-3}\}.
\]
Applying \Cref{lemma: quadratic R and cubic R} with these variance bounds yields the
stated high-probability orders for $R_q(t)$ and $R_c(t)$.

\paragraph{Covariance.}
It remains only to note the covariance structure. Conditional on $X$, every
cross product between $Q(t)$ and either $R_q(t)$ or $R_c(t)$ contains at least
one centered noise variable that appears to the first power and is independent
of the remaining factors. Hence
\[
    cov\{Q(t),R_q(t)\}=cov\{Q(t),R_c(t)\}=0.
\]
The same argument applies to the cross product between $R_q(t)$ and $R_c(t)$:
a quadratic term and a cubic term can share at most two noise edges,
leaving at least one centered independent noise factor. Therefore
\[
    cov\{R_q(t),R_c(t)\}=0.
\]
\end{proof}
   
\subsubsection{Proof of \Cref{lemma: quadratic R and cubic R}}
\label{quadratic R and cubic R}

\begin{proof}

Let $S_K(t)=\sum_{\ell=1}^T K_h(t-t_\ell)$. By \Cref{kappa j},
$S_K(t)\asymp T$ and $\sum_{\ell=1}^T K_h^2(t-t_\ell)=O(T/h)$.
Conditional on $X$, the variables
$\{\eta_{ij}(t_\ell):1\le i<j\le n,1\le \ell\le T\}$ are independent,
centered, and bounded. 
By Lemma 7.1 in \cite{schudy2011bernstein}, each $\tilde\eta_{ij}$ is central moment bounded because of its boundedness. 
Then, \Cref{hypergraph lemma} could be applied to the polynomial representations below.

\paragraph{Cubic term.}
Using $\tilde\eta_{ij}(t)=S_K(t)^{-1}\sum_\ell K_h(t-t_\ell)\eta_{ij}(t_\ell)$,
write
\[
\begin{aligned}
R_c(t)
&={n\choose3}^{-1}S_K(t)^{-3}
  \sum_{i<j<k}\sum_{\ell_1,\ell_2,\ell_3=1}^T
  K_h(t-t_{\ell_1})K_h(t-t_{\ell_2})K_h(t-t_{\ell_3})  \\
&\hspace{35mm}\times
  \eta_{ij}(t_{\ell_1})\eta_{jk}(t_{\ell_2})\eta_{ki}(t_{\ell_3}).
\end{aligned}
\]
This is a multilinear polynomial of power $3$. 
The hyperedges of the polynomial representation of $R_c(t)$ are of the form
\[
    h=\{\eta_{ij}(t_{\ell_1})\sim\eta_{jk}(t_{\ell_2})\sim\eta_{ki}(t_{\ell_3}): i<j<k, 1\le \ell \le T\}.
\] 
Each hyperedge contains three variables, and each variable is of the form $\eta_{ij}(t_\ell)$ for some $i<j$ and time index $\ell$. 
The hypergraph vertices are the variables, and the hypergraph edges are the monomials in the polynomial representation of $R_c(t)$, which are of the form $\eta_{ij}(t_{\ell_1})\eta_{jk}(t_{\ell_2})\eta_{ki}(t_{\ell_3})$.
We claim that the hypergraph parameters used in \Cref{hypergraph lemma} satisfy
\[
    \mu_{1c}=O\{n^{-2}\rho_n^2(Th)^{-1}\},\qquad
    \mu_{2c}=O\{n^{-3}\rho_n(Th)^{-2}\},\qquad
    \mu_{3c}=O\{n^{-3}(Th)^{-3}\}.
\]
Combining these quantities with
$var\{R_c(t)\}=O\{n^{-3}\rho_n^3(Th)^{-3}\}$, by \Cref{hypergraph lemma}, we
obtain, with probability at least $1-O(n^{-2})$,
\[
\begin{aligned}
|R_c(t)|
&\le C\Bigl[
    n^{-3/2}\rho_n^{3/2}(Th)^{-3/2}\sqrt{\log n}
    + n^{-2}\rho_n^2(Th)^{-1}\log n \\
&\hspace{23mm}
    + n^{-3}\rho_n(Th)^{-2}\log^2 n
    + n^{-3}(Th)^{-3}\log^3 n
\Bigr].
\end{aligned}
\]
Under \Cref{assump: sparsity}, the last two terms are
absorbed by the first two terms. Hence
\[
    R_c(t)
    =\tilde O_{p,2}\bigl(n^{-3/2}\rho_n^{3/2}(Th)^{-3/2}\sqrt{\log n}\bigr)
    +\tilde O_{p,2}\bigl(n^{-2}\rho_n^2(Th)^{-1}\log n\bigr).
\]

We now verify the claimed bounds for $\mu_{1c}$, $\mu_{2c}$, and $\mu_{3c}$.

\noindent
\underline{For $\mu_{1c}$}, after fixing one variable $S=\{\eta_{ij}(t_{\ell_1})\}$, there are $n-2$ choices of
the third index $k \neq i,j$ and two unrestricted time indices. Using
$\bbE\{|\eta_{ij}(t_\ell)|\mid X\}=O(\rho_n)$ and $S_K(t)\asymp T$ gives the bound
 $$\begin{aligned}
    \mu_{1c}&=O\left(\binom{n}{3}^{-1}T^{-3}\sum_k\sum_{\ell_2, \ell_3=1}^TK_h(t-t_{\ell_1})K_h(t-t_{\ell_2})K_h(t-t_{\ell_3})\bbE[\vert\eta_{jk}(t_{\ell_2})\vert\mid X]\bbE[\vert\eta_{ki}(t_{\ell_3})\vert\mid X]\right)\\
    &=O(n^{-2}\rho_n^2(Th)^{-1}),
    \end{aligned}$$

\noindent
\underline{For $\mu_{2c}$}, fixing two variables $S=\{\eta_{ij}(t_{\ell_1}), \eta_{jk}(t_{\ell_2})\}$ leaves only one
unrestricted time index. Then, we have the bound
$$\begin{aligned}
    \mu_{2c}&=O\left(\binom{n}{3}^{-1}T^{-3}\sum_{\ell_3=1}^TK_h(t-t_{\ell_1})K_h(t-t_{\ell_2})K_h(t-t_{\ell_3})\bbE[\vert\eta_{ki}(t_{\ell_3})\vert\mid X]\right)\\
    &=O(n^{-3}\rho_n(Th)^{-2}).
    \end{aligned}$$

\noindent
\underline{For $\mu_{3c}$}, fixing all three variables $S=\{\eta_{ij}(t_{\ell_1}), \eta_{jk}(t_{\ell_2}), \eta_{ki}(t_{\ell_3})\}$ leaves only
the coefficient, giving $\mu_{3c}=O\{n^{-3}(Th)^{-3}\}$.

\paragraph{Quadratic term.}
Similarly,
\[
\begin{aligned}
R_q(t)
&={n\choose3}^{-1}S_K(t)^{-2}
  \sum_{i<j}\sum_{k\ne i,j}\sum_{\ell_1,\ell_2=1}^T
  \tilde W_{ij}(t)K_h(t-t_{\ell_1})K_h(t-t_{\ell_2}) \\
&\hspace{40mm}\times
  \eta_{ik}(t_{\ell_1})\eta_{jk}(t_{\ell_2}).
\end{aligned}
\]
This is a multilinear polynomial of power $2$. Using
$|\tilde W_{ij}(t)|\le C\rho_n$ and
$\bbE\{|\eta_{ij}(t_\ell)|\mid X\}=O(\rho_n)$, its hypergraph parameters satisfy
\[
    \mu_{1q}=O\{n^{-2}\rho_n^2(Th)^{-1}\},
    \qquad
    \mu_{2q}=O\{n^{-3}\rho_n(Th)^{-2}\}.
\]
The factor $n^{-2}$ in $\mu_{1q}$ comes from the normalization
${n\choose3}^{-1}$ and the $O(n)$ possible third vertices containing a fixed
noise variable. By the variance bound for $R_q(t)$ and
\Cref{hypergraph lemma}, with probability at least $1-O(n^{-2})$,
\[
\begin{aligned}
|R_q(t)|
&\le C\Bigl[
    n^{-3/2}\rho_n^2(Th)^{-1}\sqrt{\log n}
    +n^{-2}\rho_n^2(Th)^{-1}\log n
    +n^{-3}\rho_n(Th)^{-2}\log^2n
\Bigr].
\end{aligned}
\]
The middle term is smaller than the first because $\log n=o(n)$, and the last term is also absorbed by the first under \Cref{assump: sparsity}. Therefore
\[
    R_q(t)=\tilde O_{p,2}\bigl(n^{-3/2}\rho_n^2(Th)^{-1}\sqrt{\log n}\bigr).
\]
Combining the bounds for $R_q(t)$ and $R_c(t)$ proves the lemma.

\end{proof}

\subsection{Proof of~\Cref{xi}} 
\label{Proof of xi}
\begin{proof}
By \Cref{assump: smooth graphon} and the kernel-grid approximation bounds in
\Cref{kappa i1,kappa j},
\[
    \sup_{i,j}
    |\tilde W_{ij}(t)-W_{ij}(t)|
    =
    O(\rho_n(h^\nu+(Th)^{-1})).
\]
Since \(W_{ij}(t)=O(\rho_n)\) and \(\tilde W_{ij}(t)=O(\rho_n)\) by \Cref{basic moment}, the product
difference bound
\[
    |abc-a'b'c'|
    \le
    |a-a'||b||c|
    +|a'||b-b'||c|
    +|a'||b'||c-c'|
\]
gives
\[
\begin{aligned}
&\left|
\tilde W_{12}(t)\tilde W_{23}(t)\tilde W_{31}(t)
-
W_{12}(t)W_{23}(t)W_{31}(t)
\right|  =
O(\rho_n^3(h^\nu+(Th)^{-1})).
\end{aligned}
\]
Taking conditional expectations and subtracting the corresponding means yields
\[
    \|g_{1h}(\cdot,t)-g_1(\cdot,t)\|_\infty
    =
    O(\rho_n^3(h^\nu+(Th)^{-1})).
\]
Therefore, $\xi_{1h}^2(t) \ge \xi_1^2(t) - var\{g_{1h}(X_1,t)-g_1(X_1,t)\} \ge C\rho_n^6 - C' \rho_n^6(h^\nu+(Th)^{-1})^2 $ with some constant $C'$. 
For \(h^\nu+(Th)^{-1}=o(1)\), this gives
\(
    \xi_{1h}^2(t)\gtrsim \rho_n^6.
\)
\end{proof}

\section{Proof of Theorem \ref{main thm}}
\label{Proof of main thm}

To prove the distributional approximation for the studentized statistic, we consider the following decomposition:
\begin{align}
\hat{T}_{nh}(t)=&\frac{\tilde{U}_{nh}(t)-\mu(t)}{\widehat{S}_{nh}(t)}=\frac{\tilde{U}_{nh}(t)-\mu(t)}{\sigma_{nh}(t)}\cdot\frac{\sigma_{nh}(t)}{\widehat{S}_{nh}(t)}\nonumber\\
=&\frac{\tilde{U}_{nh}(t)-\mu_h(t)+\mu_h(t)-\mu(t)}{\sigma_{nh}(t)}\cdot(1+\frac{\widehat{S}_{nh}^2(t)-\widehat{\sigma}_{nh}^2(t)}{\sigma_{nh}^2(t)}+\frac{\widehat{\sigma}_{nh}^2(t)-\sigma_{nh}^2(t)}{\sigma_{nh}^2(t)})^{-1/2}\label{eq:taylor},
\end{align}
where $\sigma_{nh}(t):=3\xi_{1h}/\sqrt n \asymp \rho_n^3 n^{-1/2}$ is a non-random approximation of the standard deviation of $\tilde{U}_{nh}(t)$, and $\widehat{\sigma}_{nh}^2(t)$ is an intermediate variance estimator conditional on $X$ that serves as a bridge between $\widehat{S}_{nh}^2(t)$ and $\sigma_{nh}^2(t)$:
$$\hat \sigma_{nh}^2(t):=\frac{9}{n^2}\sum_{i=1}^n\left\{\binom{n-1}{2}^{-1}\sum_{\substack{j<k\\ j, k\neq i}}\tilde W_{ij}(t)\tilde W_{jk}(t)\tilde W_{ki}(t)-U_{nh}(t)\right\}^2.$$

By \Cref{variance Hoeffding,variance linear}, we already have the expansions:

\begin{align}
    \frac{\tilde U_{nh}(t)-\mu_h(t)}{\sigma_{nh}(t)}  = & \frac{U_{nh}(t)-\mu_h(t)}{\sigma_{nh}(t)} + \frac{\tilde U_{nh}(t)-U_{nh}(t)}{\sigma_{nh}(t)} \nonumber \\
    = & \frac{Y_1(t)+Y_2(t)+Y_3(t)}{\sigma_{nh}(t)}  +  \frac{Q(t)+R_q(t)+R_c(t)}{\sigma_{nh}(t)} \nonumber\\
    = &  \frac{Y_1(t)}{\sigma_{nh}(t)} + \frac{Y_2(t)}{\sigma_{nh}(t)} + \frac{Q(t)}{\sigma_{nh}(t)}+\frac{Y_3(t)+R_q(t)+R_c(t)}{\sigma_{nh}(t)}\nonumber\\
 := &  L_n^{(1)}(t) + L_n^{(2)}(t) + \check\Delta_n(t) + R_n(t),\label{eq:U expansion}
\end{align}
where 
\begin{align*}
L_n^{(1)}(t) & = \frac{1}{\sqrt n\xi_{1h}}\sum_{i=1}^ng_{1h}(X_i, t) = \tilde O_{p, 2}(\log^{1/2}n),\\
L_n^{(2)}(t) & = \frac{2}{\sqrt n(n-1)\xi_{1h}}\sum_{i<j}g_{2h}(X_i, X_j, t) = \tilde O_{p, 2}(n^{-1/2}\log n),\\
\check\Delta_n(t) & = \frac{Q(t)}{\sigma_{nh}(t)} = \tilde O_{p, 2}\{n^{-1/2}(\rho_nTh)^{-1/2}\sqrt{\log n}\},\\
R_n(t) & = \tilde O_{p, 2}\{n^{-1}\log^{3/2}n + n^{-1}(\rho_n Th)^{-1}\sqrt{\log n}+n^{-1}(\rho_n Th)^{-3/2}\sqrt{\log n}\}.
\end{align*}
The following lemma further provides the expansions for the bias and the variance estimators. The proof of \Cref{Thm1} is deferred to \Cref{Proof of multiple points}.
\begin{lemma}\label{Thm1}
Under \Cref{assump: time obs,assump: smooth graphon,non-u,assump: sparsity}, the following expansions hold:
\begin{enumerate}
    \item[(a)] for the bias, we have
    \[
    \frac{\mu_h(t)-\mu(t)}{\sigma_{nh}(t)}= O\left(\sqrt n h^\nu+\frac{\sqrt n}{Th}\right);
    \]
    \item[(b)] for
    \(
    \hat\delta_n(t):=(\hat S_{nh}^2(t)-\hat\sigma_{nh}^2(t))/\sigma_{nh}^2(t),
    \)
    we have
    \[
    \hat\delta_n(t)=\tilde O_{p,1}\{n^{-1}(\rho_n Th)^{-1}\log n + n^{-1}(\rho_n Th)^{-1/2}\sqrt{\log n} + n^{-1}(\rho_n Th)^{-3/2}\sqrt{\log n}\};
    \]
    \item[(c)] for
    \(
    \delta_n(t):=(\hat\sigma_{nh}^2(t)-\sigma_{nh}^2(t))/\sigma_{nh}^2(t),
    \)
    we have
    \[
    \delta_n(t)=\check\delta_n(t)+\tilde O_{p,1}(n^{-1}\log n), %
    \]
where $$\check\delta_n:=\frac{1}{n}\sum_{i=1}^n\frac{g_{1h}^2(X_i, t)-\xi_{1h}^2}{\xi_{1h}^2}+\frac{4}{n(n-1)}\sum_{i\neq j}\frac{g_{1h}(X_i, t)g_{2h}(X_i, X_j, t)}{\xi_{1h}^2}=\tilde O_{p, 1}(n^{-1/2}\log^{1/2}n).$$
    \end{enumerate}
\end{lemma}

\paragraph{Decomposition of the studentized statistic.}
By collecting all the expansions above and applying the Taylor expansion, we have the following expansion for the studentized statistic $\hat T_{nh}(t)$. The proof of \Cref{prop:studentization-expansion} is deferred to \Cref{Proof of Theorem prop:studentization-expansion}.

\begin{lemma}
    \label{prop:studentization-expansion}
    Under \Cref{assump: time obs,assump: smooth graphon,non-u,assump: sparsity}, we have 
    $$\hat T_{nh}(t)=\underbrace{L_n^{(1)}(t)+L_n^{(2)}(t)-\frac{1}{2}L_n^{(1)}(t)\check\delta_n(t)}_{L_{n}(t)}+\check\Delta_n(t)+\tilde O_{p, 1}(\mcalM(n, \rho_n, T, h)).$$ 
\end{lemma}
\noindent Here, the randomness of $L_n(t)$ comes from the latent variables $X$ and the randomness of $\check\Delta_n(t)$ comes from the edge-noise variables. 
Moreover, $L_n(t)$ can be further expanded as a degree-two U-statistic and a remainder term. The proof of \Cref{T_n U} is deferred to \Cref{Proof of T_n U}.

\begin{lemma}[Degree-two representation of $L_n$] \label{T_n U}
Under \Cref{assump: time obs,assump: smooth graphon,non-u,assump: sparsity}, we have 
\[
\begin{aligned}
L_n(t)
&=
\underbrace{-\frac{1}{\sqrt n\,\xi_{1h}^3(t)}
\left\{
    \frac12\bbE[g_{1h}^3(X_1,t)]
    +
    2\bbE[
        g_{1h}(X_1,t)g_{1h}(X_2,t)g_{2h}(X_1,X_2,t)
    ]
\right\} }_{=: \alpha_h(t)}                                                  \\
&\quad+
\frac{1}{\sqrt n\,\xi_{1h}(t)}
\sum_{i=1}^n g_{1h}(X_i,t)
+
\frac{2}{\sqrt n(n-1)}
\sum_{i<j}\tilde g_{2h}(X_i,X_j,t)
+
\tilde O_{p,1}(n^{-1}\log^{3/2}n),
\end{aligned}
\]
where \(\tilde g_{2h}\) is a symmetric degenerate degree-two kernel defined as
\begin{align}
\tilde g_{2h}(x,y,t)
&:=
\frac{g_{2h}(x,y,t)}{\xi_{1h}(t)}                 \nonumber      \\
&\quad
-\frac{n-1}{4n\,\xi_{1h}^3(t)}
\Big[
    g_{1h}(x,t)\{g_{1h}^2(y,t)-\xi_{1h}^2(t)\}
    +
    g_{1h}(y,t)\{g_{1h}^2(x,t)-\xi_{1h}^2(t)\}
\Big]                                             \nonumber      \\
&\quad
-\frac{n-1}{n\,\xi_{1h}^3(t)}
\Big[
    g_{1h}(x,t)\zeta_h(y,t)
    +
    g_{1h}(y,t)\zeta_h(x,t)
\Big], \label{tilde g2h expanded}
\end{align}
with \[
    \zeta_h(x,t)
    :=
    \bbE\{g_{1h}(X_2,t)g_{2h}(X_1,X_2,t)\mid X_1=x\},
\]
and it satisfies
\(
    \bbE\{\tilde g_{2h}(X_1,X_2,t)\mid X_1\}=0\) and 
\(
    \bbE\{\tilde g_{2h}^2(X_1,X_2,t)\}=O(1).
\)
\end{lemma}

\noindent \underline{By combining the decomposition in \Cref{prop:studentization-expansion,T_n U}, we obtain}
$$\hat T_{nh}(t)=\underbrace{\alpha_h(t)+L_n^{(1)}(t)+L_n^{(2)}(t)-\frac{1}{2}L_n^{(3)}(t) -\frac{1}{2}L_n^{(4)}(t)}_{=:J_n(t)}+\check\Delta_n(t)+\tilde O_{p, 1}(\mcalM(n, \rho_n, T, h)),$$ 
where $L_n^{(3)}(t)$ and $L_n^{(4)}(t)$ are defined as
\begin{align*}
    L_n^{(3)}(t):= & \frac{1}{n^{3/2}\xi_{1h}^3(t)}
\sum_{i<j} 
\Big[
    g_{1h}(X_i,t)\{g_{1h}^2(X_j,t)-\xi_{1h}^2(t)\}
    +
    g_{1h}(X_j,t)\{g_{1h}^2(X_i,t)-\xi_{1h}^2(t)\}
\Big]\\
=& \tilde O_{p, 2}(n^{-1/2}\log n),\\
L_n^{(4)}(t):= &\frac{4}{n^{3/2}\xi_{1h}^3(t)}
\sum_{i<j}\Big[
    g_{1h}(X_i,t)\zeta_h(X_j,t)
    +
    g_{1h}(X_j,t)\zeta_h(X_i,t)
\Big] = \tilde O_{p, 2}(n^{-1/2}\log n).
\end{align*}

\paragraph{Distributional approximation.} To bound the approximation error of the population Edgeworth expansion in \Cref{main thm}, we follow the same strategy as in \citet{zhang2022edgeworth}. 
Let $\mcalF_Z$ be the distribution function of a random variable $Z$.
Recall that $G_{nh}(\cdot)$ is the Edgeworth expansion defined in \Cref{main thm}. We have the following three lemmas to bound the approximation error of $G_{nh}$ to $\mcalF_{\hat T_{nh}+\delta_T}$, where $\delta_T\sim\mcalN(0, c_\delta n^{-1}\log n)$ is an independent Gaussian perturbation with a sufficiently large constant $c_\delta$.

\medskip
\noindent
\underline{First}, in \Cref{8.55}, we show that the edge-noise term $\check\Delta_n(t)$ can be approximated by a Gaussian distribution, which yields the bound for $\left\Vert \mcalF_{J_n+\check\Delta_n+\delta_T}(\cdot)-\mcalF_{J_n+\tilde\Delta_n+\delta_T}(\cdot)\right\Vert_\infty$. The proof of this lemma is provided in \Cref{Proof of 8.55}.

\begin{lemma}[Replacing $\check\Delta_n$ by its Gaussian approximation]
    \label{8.55}
    Under \Cref{assump: time obs,assump: smooth graphon,non-u,assump: sparsity}, we have
    \begin{itemize}
        \item[(a)] Conditional on $X$, the random variable $\check\Delta_n(t)$ is asymptotically normal as follows,
        $$\left\Vert \mcalF_{\check\Delta_n|X}(\cdot)-\mcalF_{\mcalN\left(0, (n\rho_n)^{-1}\sigma_X^2\right)}(\cdot)\right\Vert_\infty=\tilde O_{p, 1}\left((\rho_n Th)^{-\frac{1}{2}}n^{-1}\right),$$
        where $\sigma_X^2:=n\rho_n var(\check\Delta_n(t)|X) \asymp (Th)^{-1}$ with probability at least $1-O(n^{-1})$. 
        
        \item[(b)] Denote $\tilde\Delta_n \mid X \sim\mcalN(0, (n\rho_n)^{-1}\sigma_X^2)$. Then, 
        $$\left\Vert \mcalF_{J_n+\check\Delta_n+\delta_T}(\cdot)-\mcalF_{J_n+\tilde\Delta_n+\delta_T}(\cdot)\right\Vert_\infty=O((\rho_n Th)^{-1/2}n^{-1} + n^{-1}).$$ 
    \end{itemize}
\end{lemma}

\medskip
\noindent
\underline{Second}, in \Cref{8.54}, we show that the distribution of $\hat T_{nh}(t)$ can be approximated by the distribution of $J_n(t)+\check\Delta_n(t)$, which yields the bound for $\left\Vert \mcalF_{\hat T_{nh}+\delta_T}(\cdot)-\mcalF_{J_n+\check\Delta_n+\delta_T}(\cdot)\right\Vert_\infty$. The proof of this lemma is provided in \Cref{Proof of 8.54}.

\begin{lemma}[Bound the remainder]
\label{8.54}
Under \Cref{assump: time obs,assump: smooth graphon,non-u,assump: sparsity}, we have
    $$\left\Vert \mcalF_{\hat T_{nh}+\delta_T}(\cdot)-\mcalF_{J_n+\check\Delta_n+\delta_T}(\cdot)\right\Vert_\infty=O(\mcalM(n, \rho_n, T, h)).$$
\end{lemma}

\medskip
\noindent\underline{Third}, in \Cref{8.56}, we show that the distribution of $J_n(t)+\tilde\Delta_n(t)$ can be approximated by the Edgeworth expansion $G_{nh}(u)$, which yields the bound for $\left\Vert \mcalF_{J_n+\tilde\Delta_n+\delta_T}(\cdot)-G_{nh}(\cdot)\right\Vert_\infty$. The proof of this lemma is provided in \Cref{Proof of Lemma 8.56}.
\begin{lemma}[Edgeworth approximation for $J_n(t)+\tilde\Delta_n(t)$]
    \label{8.56}%
    Under \Cref{assump: time obs,assump: smooth graphon,non-u,assump: sparsity}, with sufficiently large constant $c_\delta>0$, we have
    $$\left\Vert \mcalF_{J_{n}+\tilde\Delta_n+\delta_T}(\cdot)-G_{nh}(\cdot)\right\Vert_\infty=O(n^{-1}\log n+(n\rho_n Th)^{-1}\log n).$$
\end{lemma}

\medskip
\noindent\underline{
We are now ready to prove \Cref{main thm} by combining the three lemmas above.}
\begin{proof}[\textbf{Proof of \Cref{main thm}}]

    Recall that the perturbed statistic in the theorem is
\[
\tilde T_{nh}(t)=\hat T_{nh}(t)+\delta_T,
\]
where $\delta_T\sim N(0,c_\delta n^{-1}\log n)$ is independent of the data.  By the triangle inequality,
\[
\begin{aligned}
\left\|\mcalF_{\tilde T_{nh}}-G_{nh}\right\|_\infty
&=\left\|\mcalF_{\hat T_{nh}+\delta_T}-G_{nh}\right\|_\infty\\
&\le \left\|\mcalF_{\hat T_{nh}+\delta_T}-\mcalF_{J_n+\check\Delta_n+\delta_T}\right\|_\infty\\
&\quad+\left\|\mcalF_{J_n+\check\Delta_n+\delta_T}-\mcalF_{J_n+\tilde\Delta_n+\delta_T}\right\|_\infty\\
&\quad+\left\|\mcalF_{J_n+\tilde\Delta_n+\delta_T}-G_{nh}\right\|_\infty.
\end{aligned}
\]
The three terms are controlled by Lemmas~\ref{8.54}, \ref{8.55}, and \ref{8.56}, respectively, and the approximation errors are included in $\mcalM(n,\rho_n,T,h)$.  This proves
\[
\left\|\mcalF_{\tilde T_{nh}}-G_{nh}\right\|_\infty=O\{\mcalM(n,\rho_n,T,h)\}.
\]
For the empirical Edgeworth approximation, use the same arguments as in \citet{zhang2022edgeworth} to show that the difference between the empirical Edgeworth expansion $\hat G_{nh}$ and the population Edgeworth expansion $G_{nh}$ is also of order $\tilde O_{p,1}\{\mcalM(n,\rho_n,T,h)\}$. The triangle inequality gives
\[
\left\|\mcalF_{\tilde T_{nh}}-\hat G_{nh}\right\|_\infty
\le
\left\|\mcalF_{\tilde T_{nh}}-G_{nh}\right\|_\infty
+
\left\|G_{nh}-\hat G_{nh}\right\|_\infty
=\tilde O_{p,1}\{\mcalM(n,\rho_n,T,h)\}.
\]
\end{proof}

\subsection{Proof of \Cref{Thm1}}
\label{Proof of multiple points}

To prove part (a), we first establish the following bounds for the derivatives of $W_{ij}(t)$ and $\mu(t)$, as well as the smoothing bias of $\tilde W_{ij}(t)$. The proof of \Cref{derivative mu} is deferred to \Cref{proof of derivative mu}.
\begin{lemma}[Derivative and smoothing-bias bounds]
\label{derivative mu}
Under \Cref{assump: time obs,assump: smooth graphon}, for every
$1\le r\le \nu$, we have the the bounds for the $r$-th derivative for $\mu(t)$ and $W_{ij}(t)$:
\[
    \sup_{i<j}\left|W_{ij}^{(r)}(t)\right|=O(\rho_n),
    \qquad
    \left|\mu^{(r)}(t)\right|=O(\rho_n^3).
\]
Moreover,
\[
    \sup_{i<j}\left|\tilde W_{ij}(t)-W_{ij}(t)\right|
    =
    O\left\{\rho_n\left(h^\nu+(Th)^{-1}\right)\right\}.
\]
\end{lemma}

\begin{proof}[\textbf{Proof of \Cref{Thm1}(a)}]
By \Cref{derivative mu},
\[
    \sup_{i<j}\left|\tilde W_{ij}(t)-W_{ij}(t)\right|
    =
    O(\rho_n\left(h^\nu+(Th)^{-1}\right)).
\]
Since both $\tilde W_{ij}(t)$ and $W_{ij}(t)$ are uniformly $O(\rho_n)$,
\[
\begin{aligned}
&\left|
\tilde W_{12}(t)\tilde W_{23}(t)\tilde W_{31}(t)
-
W_{12}(t)W_{23}(t)W_{31}(t)
\right|  \\
&\qquad \le
\left|\tilde W_{12}(t)-W_{12}(t)\right|\,|\tilde W_{23}(t)|\,|\tilde W_{31}(t)|
+
|W_{12}(t)|\,\left|\tilde W_{23}(t)-W_{23}(t)\right|\,|\tilde W_{31}(t)|  \\
&\qquad\quad
+
|W_{12}(t)|\,|W_{23}(t)|\,\left|\tilde W_{31}(t)-W_{31}(t)\right|  \\
&\qquad =
O(\rho_n^3 \left(h^\nu+(Th)^{-1}\right)).
\end{aligned}
\]
Taking expectations gives
\(
    \mu_h(t)-\mu(t)=O(\rho_n^3 \left(h^\nu+(Th)^{-1}\right)).
\)
Because
\(
    \sigma_{nh}(t)=\frac{3\xi_{1h}(t)}{\sqrt n}\) 
and
\(
    \xi_{1h}(t)\asymp \rho_n^3,
\)
we obtain
\[
    \frac{\mu_h(t)-\mu(t)}{\sigma_{nh}(t)}
    =
    O\left\{\sqrt n\left(h^\nu+(Th)^{-1}\right)\right\}
    =
    O\left(\sqrt n h^\nu+\frac{\sqrt n}{Th}\right).
\]
This proves part (a).
\end{proof}

To prove part (b), we state the following bound for quadratic and cubic noise remainders when one node is fixed. The proof of \Cref{quadratic Ra} is deferred to \Cref{proof of quadratic Ra}.

\begin{lemma}[Nodewise quadratic and cubic noise remainders]
\label{quadratic Ra}
Fix $i\in\{1,\ldots,n\}$ and define
\[
\begin{aligned}
    R_{a,q}^{(i)}(t)
&=
\binom{n-1}{2}^{-1}
\sum_{\substack{j<k\\ j,k\neq i}}
\tilde W_{ki}(t)\tilde\eta_{ij}(t)\tilde\eta_{jk}(t),
\\
R_{a,c}^{(i)}(t)
&=
\binom{n-1}{2}^{-1}
\sum_{\substack{j<k\\ j,k\neq i}}
\tilde\eta_{ij}(t)\tilde\eta_{jk}(t)\tilde\eta_{ki}(t).
\end{aligned}
\]
Under \Cref{assump: time obs,non-u,assump: sparsity}, uniformly
over $i$,
\[
    R_{a,q}^{(i)}(t)
    =
    \tilde O_{p,2}\left\{
        n^{-1}(Th)^{-1}\rho_n^2\log n
    \right\},
\]
and
\[
    R_{a,c}^{(i)}(t)
    =
    \tilde O_{p,2}\left\{
        n^{-1}(Th)^{-3/2}\rho_n^{3/2}\sqrt{\log n}
    \right\}
    +
    \tilde O_{p,2}\left\{
        n^{-1}(Th)^{-1}\rho_n^2\log n
    \right\}.
\]
Consequently, the nodewise remainder
\[
    R_a^{(i)}(t):=R_{a,q}^{(i)}(t)+R_{a,c}^{(i)}(t)=\tilde O_{p,2}\left\{
        n^{-1}(Th)^{-1}\rho_n^2\log n
    \right\}
    +
    \tilde O_{p,2}\left\{
        n^{-1}(Th)^{-3/2}\rho_n^{3/2}\sqrt{\log n}
    \right\}.
\]
\end{lemma}

\begin{proof}[\textbf{Proof of \Cref{Thm1}(b)}]
Define
\begin{align*}
    \hat a_i(t)
    & :=
    \binom{n-1}{2}^{-1}
    \sum_{\substack{j<k\\ j,k\neq i}}
    \tilde A_{ij}(t)\tilde A_{jk}(t)\tilde A_{ki}(t),\\
    a_i(t)
    & :=
    \binom{n-1}{2}^{-1}
    \sum_{\substack{j<k\\ j,k\neq i}}
    \tilde W_{ij}(t)\tilde W_{jk}(t)\tilde W_{ki}(t)
    =
    \bbE\{\hat a_i(t)\mid X\}.
\end{align*}
Then we have
\(
    \frac1n\sum_{i=1}^n\hat a_i(t)=\tilde U_{nh}(t)\)
and 
\(
    \frac1n\sum_{i=1}^n a_i(t)=U_{nh}(t).
\)
Therefore,
\[
    \frac{n\hat S_{nh}^2(t)}{9}
    =
    \frac1n\sum_{i=1}^n\{\hat a_i(t)-\tilde U_{nh}(t)\}^2,
\qquad
    \frac{n\hat\sigma_{nh}^2(t)}{9}
    =
    \frac1n\sum_{i=1}^n\{a_i(t)-U_{nh}(t)\}^2.
\]
Using the fact 
\(
    \frac1n\sum_{i=1}^n\{\hat a_i(t)-U_{nh}(t)\}
    =
    \tilde U_{nh}(t)-U_{nh}(t),
\)
we obtain
\[
\begin{aligned}
\frac{n}{9}\{\hat S_{nh}^2(t)-\hat\sigma_{nh}^2(t)\}
&=
\frac1n\sum_{i=1}^n\{\hat a_i(t)-a_i(t)\}^2 + \frac2n\sum_{i=1}^n\{a_i(t)-U_{nh}(t)\}\{\hat a_i(t)-a_i(t)\}  \\
&\quad
-\{\tilde U_{nh}(t)-U_{nh}(t)\}^2 .
\end{aligned}
\]
We bound the three terms on the right-hand side. 
First, by \Cref{variance linear},
\[
    \tilde U_{nh}(t)-U_{nh}(t)=Q(t)+R_q(t)+R_c(t),
\]
and hence
\[
\{\tilde U_{nh}(t)-U_{nh}(t)\}^2
=
\tilde O_{p,2}\{n^{-2}(Th)^{-1}\rho_n^5\log n\}
+
\tilde O_{p,2}\{n^{-3}(Th)^{-3}\rho_n^3 \log n \cdot \max\{1, Th\rho_n\}\}.
\]
Second, decompose $\hat a_i(t)-a_i(t)$ into its terms that are linear,
quadratic, and cubic in $\{\tilde\eta_{ij}(t):=\tilde A_{ij}(t)-\tilde W_{ij}(t)\}$. 
By defining $\Theta_{ij}(t)=\sum_{k\neq i, j}\tilde W_{jk}(t)\tilde W_{ki}(t)$, we have
\[
\begin{aligned}
\hat a_i(t)-a_i(t)
&=
\binom{n-1}{2}^{-1}
\sum_{j\neq i}\Theta_{ij}(t)\tilde\eta_{ij}(t) +
\binom{n-1}{2}^{-1}
\sum_{\substack{j<k\\ j,k\neq i}}
\tilde W_{ij}(t)\tilde W_{ki}(t)\tilde\eta_{jk}(t)
+
R_a^{(i)}(t),
\end{aligned}
\]
where $R_a^{(i)}(t)$ collects the quadratic and cubic terms in
$\tilde\eta$. 
By Bernstein's inequality, \Cref{basic moment}, and the fact $|\Theta_{ij}(t)|= O(n\rho_n^2)$, we have
\begin{align*}
\binom{n-1}{2}^{-1}
    \sum_{j\neq i}\Theta_{ij}(t)\tilde\eta_{ij}(t)
    & =
    \tilde O_{p,2}\left\{
        n^{-1/2}(Th)^{-1/2}\rho_n^{5/2}\sqrt{\log n}
    \right\}, \\
    \binom{n-1}{2}^{-1}
    \sum_{\substack{j<k\\ j,k\neq i}}
    \tilde W_{ij}(t)\tilde W_{ki}(t)\tilde\eta_{jk}(t)
    & =
    \tilde O_{p,2}\left\{
        n^{-1}(Th)^{-1/2}\rho_n^{5/2}\sqrt{\log n}
    \right\},
\end{align*}
uniformly in $i$. 
Together with the bound for $R_a^{(i)}(t)$ in \Cref{quadratic Ra}, this yields
\[
\begin{aligned}
\frac1n\sum_{i=1}^n\{\hat a_i(t)-a_i(t)\}^2
&=
\tilde O_{p,1}\{(nTh)^{-1}\rho_n^5\log n\} +
\tilde O_{p,1}\{n^{-2}(Th)^{-3}\rho_n^3\log n \cdot \max\{1, Th\rho_n \log n\}\}.
\end{aligned}
\]
Thrid, by using the leading linear representation above, we have
\[
\begin{aligned}
\hat a_i(t)-a_i(t)
&=
\frac{1}{n-1}
\sum_{j\neq i}
\tilde\Theta_{ij}(t)\tilde\eta_{ij}(t)  +
\tilde O_{p,1}
\left\{
    n^{-1}(Th)^{-1/2}\rho_n^{5/2}\sqrt{\log n}
\right\} \\
&\quad+
\tilde O_{p,1}
\left\{
    n^{-1}(Th)^{-1}\rho_n^2\log n
\right\}
+
\tilde O_{p,1}
\left\{
    n^{-1}(Th)^{-3/2}\rho_n^{3/2}\sqrt{\log n}
\right\},
\end{aligned}
\]
where
\[
    \tilde\Theta_{ij}(t)
    :=
    \frac{2}{n-2}\sum_{k\neq i,j}
    \tilde W_{jk}(t)\tilde W_{ki}(t),
    \qquad\text{with} \quad
    |\tilde\Theta_{ij}(t)|\lesssim \rho_n^2.
\]
Therefore, the cross term can be decomposed as
\[
\begin{aligned}
& \frac{2}{n}\sum_{i=1}^n \{a_i(t)-U_{nh}(t)\}\{\hat a_i(t)-a_i(t)\} \\
 = & 
\frac{2}{n(n-1)}
\sum_{i\neq j}
 \{a_i(t)-U_{nh}(t)\}\tilde\Theta_{ij}(t)\tilde\eta_{ij}(t)   +
\tilde O_{p,1}
\left\{
    n^{-1}(Th)^{-1/2}\rho_n^{11/2}\sqrt{\log n}
\right\} \\
& +
\tilde O_{p,1}
\left\{
    n^{-1}(Th)^{-1}\rho_n^5\log n
\right\}  +
\tilde O_{p,1}
\left\{
    n^{-1}(Th)^{-3/2}\rho_n^{9/2}\sqrt{\log n}
\right\}.
\end{aligned}
\]
It remains to bound the leading linear term. Conditional on \(X\), the
coefficients \(a_i(t)-U_{nh}(t)\) and \(\tilde\Theta_{ij}(t)\) are fixed, and the
variables \(\tilde\eta_{ij}(t)\) are centered and independent across edges.
Using the fact that 
\(
    |a_i(t)-U_{nh}(t)|=O(\rho_n^3).
\)  and
\(|\tilde\Theta_{ij}(t)|\lesssim \rho_n^2\), by Bernstein's inequality, we have
\[
\begin{aligned}
\frac{2}{n(n-1)}
\sum_{i\neq j}
\{a_i(t)-U_{nh}(t)\}\tilde\Theta_{ij}(t)\tilde\eta_{ij}(t)
&=
\tilde O_{p,1}
\left\{
    n^{-1}(Th)^{-1/2}\rho_n^{11/2}\sqrt{\log n}
\right\}  \\
&\quad+
\tilde O_{p,1}
\left\{
    n^{-2}(Th)^{-1}\rho_n^5\log n
\right\}.
\end{aligned}
\]
The second term is dominated by the first under \(\Cref{assump: sparsity}\).
Therefore,
the cross term is bounded by
\[
\begin{aligned}
\frac{2}{n}\sum_{i=1}^n
\{a_i(t)-U_{nh}(t)\}\{\hat a_i(t)-a_i(t)\}
&=
\tilde O_{p,1}
\left\{
    n^{-1}(Th)^{-1/2}\rho_n^{11/2}\sqrt{\log n}
\right\}  \\
&\quad+
\tilde O_{p,1}
\left\{
    n^{-1}(Th)^{-1}\rho_n^5\log n
\right\} \\
&\quad+
\tilde O_{p,1}
\left\{
    n^{-1}(Th)^{-3/2}\rho_n^{9/2}\sqrt{\log n}
\right\}.
\end{aligned}
\]
Combining the three results above, under
\Cref{assump: sparsity}, gives
\[
\begin{aligned}
\frac{n}{9}\{\hat S_{nh}^2(t)-\hat\sigma_{nh}^2(t)\}
&=
\tilde O_{p,1}\left\{
    n^{-1}(Th)^{-1}\rho_n^5\log n
\right\}  \\
&\quad+
\tilde O_{p,1}\left\{
    n^{-1}(Th)^{-1/2}\rho_n^{11/2}\sqrt{\log n}
\right\} \\
&\quad+
\tilde O_{p,1}\left\{
    n^{-1}(Th)^{-3/2}\rho_n^{9/2}\sqrt{\log n}
\right\}.
\end{aligned}
\]
Finally, since
\(
    \frac{n\sigma_{nh}^2(t)}{9}
    =
    \xi_{1h}^2(t)
    \asymp \rho_n^6,
\)
we divide both sides by $\xi_{1h}^2(t)$ and obtain
\[
\hat\delta_n(t)
=
\tilde O_{p,1}\left\{
    n^{-1}(\rho_n Th)^{-1}\log n
\right\}
+
\tilde O_{p,1}\left\{
    n^{-1}(\rho_n Th)^{-1/2}\sqrt{\log n}
\right\}
+
\tilde O_{p,1}\left\{
    n^{-1}(\rho_n Th)^{-3/2}\sqrt{\log n}
\right\}.
\]
This proves part (b).
\end{proof}

We finally analyze the decomposition of $\delta_n(t)=\{\hat\sigma_{nh}^2(t)-\sigma_{nh}^2(t)\}/\sigma_{nh}^2(t)$. The main idea is also to fix one node index and then to apply Hoeffding expansion to the leave-one-out U-statistic.

\begin{proof}[\textbf{Proof of \Cref{Thm1}(c)}]
Since $n^{-1}\sum_i a_i(t)=U_{nh}(t)$, we have
$$
    \begin{aligned}
        \frac{n\hat\sigma_{nh}^2(t)}{9}&=\frac{1}{n}\sum_{i=1}^n(a_i(t)-U_{nh}(t))^2\\
        &=\frac{1}{n}\sum_{i=1}^n(a_i(t)-\mu_h(t))^2-(\mu_h(t)-U_{nh}(t))^2\\
        &=\frac{1}{n}\sum_{i=1}^n(a_i(t)-\mu_h(t))^2-\tilde O_{p, 1}\left(\rho_n^{6}n^{-1}\log n\right),
    \end{aligned}
    $$
where the last step used the fact  \(
    U_{nh}(t)-\mu_h(t)
    =
    \tilde O_{p,1}\left(\rho_n^3n^{-1/2}\sqrt{\log n}\right)
\)
in \Cref{variance Hoeffding}.

Define
\(
    \bar a_i(t):=\mu_h(t)+g_{1h}(X_i,t)=\bbE\left[\tilde W_{12}(t)\tilde W_{23}(t)\tilde W_{31}(t)|X_1=X_i\right].
\)
Then
\[
\begin{aligned}
\frac{n\hat\sigma_{nh}^2(t)}{9}
&=
\frac1n\sum_{i=1}^n g_{1h}^2(X_i,t)
+
\frac2n\sum_{i=1}^n\{a_i(t)-\bar a_i(t)\}g_{1h}(X_i,t) \\
&\quad+
\frac1n\sum_{i=1}^n\{a_i(t)-\bar a_i(t)\}^2
+
\tilde O_{p,1}\left(\rho_n^6n^{-1}\log n\right).
\end{aligned}
\]
Using the same arguments (8.48)-(8.50) in \cite{zhang2022edgeworth}, by Hoeffding expansion, we have  uniformly in $i$ that
$$\begin{aligned}
        a_i(t)-\tilde a_i(t)&=\binom{n-1}{2}^{-1}\sum_{\substack{j<k\\i\neq j, k}}\left\{\tilde W_{ij}(t)\tilde W_{jk}(t)\tilde W_{ki}(t)-\bbE[\tilde W_{ij}(t)\tilde W_{jk}(t)\tilde W_{ki}(t)|X_i]\right\}\\
        &=\frac{2}{n-1}\sum_{j\neq i}g_{2h}(X_i,X_j,t)
    +
    \frac{2}{(n-1)(n-2)}
    \sum_{\substack{j<k\\ j,k\neq i}}
    g_{3h}(X_i,X_j,X_k,t)\\
        &=\tilde O_{p, 2}\left(\rho_n^3 n^{-\frac{1}{2}}\sqrt{\log n}\right)+\tilde O_{p, 2}\left(\rho_n^3 n^{-1}\log n\right).
    \end{aligned}$$
Hence,
\[
    \frac1n\sum_{i=1}^n\{a_i(t)-\bar a_i(t)\}^2
    =
    \tilde O_{p,1}\left(\rho_n^6n^{-1}\log n\right).
\]
Moreover, for the cross term, we have
\[
\begin{aligned}
\frac2n\sum_{i=1}^n\{a_i(t)-\bar a_i(t)\}g_{1h}(X_i,t)
& =\frac{2}{n}\sum_{i=1}^n\left(\frac{2}{n-1}\sum_{j\neq i}g_{2h}(X_i, X_j, t)+\tilde O_{p, 2}\left(\rho_n^3 n^{-1}\log n\right)\right)g_{1h}(X_i, t)\\
&=
\frac{4}{n(n-1)}
\sum_{i,j: i\neq j}
g_{1h}(X_i,t)g_{2h}(X_i,X_j,t) +
\tilde O_{p,1}\left(\rho_n^6n^{-1}\log n\right),
\end{aligned}
\]
where we used the fact that $g_{1h}(X_i,t)=O(\rho_n^3)$.

Combining the two results above, we have
\[
\begin{aligned}
\frac{n\hat\sigma_{nh}^2(t)}{9}
=
\frac1n\sum_{i=1}^n g_{1h}^2(X_i,t)
+
\frac{4}{n(n-1)}
\sum_{i,j: i\neq j}
g_{1h}(X_i,t)g_{2h}(X_i,X_j,t)  +
\tilde O_{p,1}\left(\rho_n^6n^{-1}\log n\right).
\end{aligned}
\]
Since
\(
    \frac{n\sigma_{nh}^2(t)}{9}=\xi_{1h}^2(t)\asymp \rho_n^6,
\)
we obtain
\[
\begin{aligned}
\delta_n(t)
&=
\frac{\hat\sigma_{nh}^2(t)-\sigma_{nh}^2(t)}{\sigma_{nh}^2(t)} \\
&=
\frac1n\sum_{i=1}^n
\frac{g_{1h}^2(X_i,t)-\xi_{1h}^2(t)}{\xi_{1h}^2(t)}
+
\frac{4}{n(n-1)}
\sum_{i\neq j}
\frac{g_{1h}(X_i,t)g_{2h}(X_i,X_j,t)}{\xi_{1h}^2(t)}
+
\tilde O_{p,1}(n^{-1}\log n).
\end{aligned}
\]
The first two terms are exactly $\check\delta_n(t)$. Since
$g_{1h}^2(X_i,t)/\xi_{1h}^2(t)$ is uniformly bounded and centered after
subtracting one, Bernstein's inequality gives
\[
    \frac1n\sum_{i=1}^n
    \frac{g_{1h}^2(X_i,t)-\xi_{1h}^2(t)}{\xi_{1h}^2(t)}
    =
    \tilde O_{p,1}(n^{-1/2}\sqrt{\log n}).
\]
The second term in $\check\delta_n(t)$ is a degenerate second-order
U-statistic with bounded kernel and, by Theorem 1 in \cite{10.1214/EJP.v12-430}, it is no larger. Hence
\[
    \check\delta_n(t)
    =
    \tilde O_{p,1}(n^{-1/2}\sqrt{\log n}),
\]
and
\[
    \delta_n(t)
    =
    \check\delta_n(t)+\tilde O_{p,1}(n^{-1}\log n).
\]
This proves part (c).
\end{proof}

\subsubsection{Proof of \Cref{derivative mu}}
\label{proof of derivative mu}
\begin{proof}
Recall that
\(
    W_{ij}(t)=\rho_n F(X_i,X_j,t)\{1-2G(X_i,X_j,t)\}.
\)
By \Cref{assump: smooth graphon}, the functions $F$ and $G$ have uniformly
bounded derivatives up to order $\nu$ on the interior time interval. Since
products of bounded smooth functions have the same order of smoothness, every
derivative $W_{ij}^{(r)}(t)$, $1\le r\le \nu$, is bounded by $C\rho_n$ uniformly
over $i,j$.

Similarly, for
\(
    \mu(t)=\bbE\{W_{12}(t)W_{23}(t)W_{31}(t)\},
\)
differentiating this product and using the bound
$\sup_{i<j}|W_{ij}^{(r)}(t)|=O(\rho_n)$ gives
$|\mu^{(r)}(t)|=O(\rho_n^3)$ for $1\le r\le \nu$.

Finally, we give the standard bound for the smoothing approximation error. We prove it here for self-consistency.
Let
\(
    D_T(t):=\frac1T\sum_{\ell=1}^T K_h(t-t_\ell)\) and 
\(
    s_\ell:=\frac{t_\ell-t}{h}.
\)
By \Cref{kappa j}, \(D_T(t)\asymp 1\).  Since the kernel weights are
normalized,
\[
\tilde W_{ij}(t)-W_{ij}(t)
=
D_T(t)^{-1}
\frac1T\sum_{\ell=1}^T
K_h(t-t_\ell)\{W_{ij}(t_\ell)-W_{ij}(t)\}.
\]
By the Taylor expansion around \(t\), we have
\[
W_{ij}(t_\ell)-W_{ij}(t)
=
\sum_{r=1}^{\nu-1}
\frac{h^r s_\ell^r}{r!}W_{ij}^{(r)}(t)
+
\frac{h^\nu s_\ell^\nu}{\nu!}W_{ij}^{(\nu)}(t)
+
o(\rho_n h^\nu |s_\ell|^\nu).
\]
For the lower-order terms with \(1\le r\le\nu-1\), by \Cref{kappa i1},
we have
\[
\frac1T\sum_{\ell=1}^T K_h(t-t_\ell)s_\ell^r
=
O\left\{(Th)^{-1}+h^{-(r-1)}e^{-c'h^{-\alpha}}\right\},
\]
as the kernel moments satisfy \(\kappa_{r1}=0\). Therefore, together with \(W_{ij}^{(r)}(t)=O(\rho_n)\) uniformly in \(i,j\),  the
lower-order Taylor terms contribute
\[
O\{\rho_n(Th)^{-1}\}+o(\rho_n h^\nu).
\]
For the \(\nu\)-th-order term,
\[
\frac1T\sum_{\ell=1}^T K_h(t-t_\ell)s_\ell^\nu
=
\kappa_{\nu1}
+
O\left\{(Th)^{-1}+h^{-(\nu-1)}e^{-c'h^{-\alpha}}\right\}
=
O(1),
\]
and hence this term contributes \(O(\rho_n h^\nu)\).  The Taylor
remainder is also \(O(\rho_n h^\nu)\).  Since \(D_T(t)\asymp1\), we conclude
that
\[
    \sup_{i<j}
    |\tilde W_{ij}(t)-W_{ij}(t)|
    =
    O\left\{\rho_n\left(h^\nu+(Th)^{-1}\right)\right\}.
\]
\end{proof}

\subsubsection{Proof of \Cref{quadratic Ra}}
\label{proof of quadratic Ra}
\begin{proof}
We give the proof for a fixed $i$; the uniform version follows by applying the
same bound with constants large enough to allow a union bound over
$i=1,\ldots,n$. We use the same technique as in the proof of \Cref{lemma: quadratic R and cubic R} in \Cref{quadratic R and cubic R}.

First, consider $R_{a,q}^{(i)}(t)$. Conditional on $X$, the summands are
centered except when the same two noise variables reappear. As in the proof of
\Cref{lemma: quadratic R and cubic R},
\[
    var\{R_{a,q}^{(i)}(t)\mid X\}
    =
    O\left\{n^{-2}(Th)^{-2}\rho_n^4\right\}.
\]
The hypergraph parameters required for \Cref{hypergraph lemma} satisfy
\[
    \mu_{1q}
    =
    O\left\{n^{-1}(Th)^{-1}\rho_n^2\right\},
    \qquad
    \mu_{2q}
    =
    O\left\{n^{-2}(Th)^{-2}\rho_n\right\}.
\]
Therefore, with probability at least $1-O(n^{-2})$,
\[
\begin{aligned}
|R_{a,q}^{(i)}(t)|
&\le
C\Bigl[
    n^{-1}(Th)^{-1}\rho_n^2\sqrt{\log n}
    + n^{-1}(Th)^{-1}\rho_n^2\log n   + n^{-2}(Th)^{-2}\rho_n\log^2 n
\Bigr].
\end{aligned}
\]
The last term is dominated by
$n^{-1}(Th)^{-1}\rho_n^2\log n$ under \Cref{assump: sparsity}, so
\[
    R_{a,q}^{(i)}(t)
    =
    \tilde O_{p,2}\left\{
        n^{-1}(Th)^{-1}\rho_n^2\log n
    \right\}.
\]

Next, consider $R_{a,c}^{(i)}(t)$. Again using the same hypergraph argument,
\[
    var\{R_{a,c}^{(i)}(t)\mid X\}
    =
    O\left\{n^{-2}(Th)^{-3}\rho_n^3\right\},
\]
and
\[
    \mu_{1c}
    =
    O\left\{n^{-1}(Th)^{-1}\rho_n^2\right\},
    \quad
    \mu_{2c}
    =
    O\left\{n^{-2}(Th)^{-2}\rho_n\right\},
    \quad
    \mu_{3c}
    =
    O\left\{n^{-2}(Th)^{-3}\right\}.
\]
Hence, with probability at least $1-O(n^{-2})$,
\[
\begin{aligned}
|R_{a,c}^{(i)}(t)|
&\le
C\Bigl[
    n^{-1}(Th)^{-3/2}\rho_n^{3/2}\sqrt{\log n}
    +n^{-1}(Th)^{-1}\rho_n^2\log n \\
&\hspace{25mm}
    +n^{-2}(Th)^{-2}\rho_n\log^2 n
    +n^{-2}(Th)^{-3}\log^3 n
\Bigr].
\end{aligned}
\]
The third term is dominated by the second term under
\Cref{assump: sparsity}. The last term is dominated by the first term because
\[
    n^{-2}(Th)^{-3}\log^3 n
    \le
    n^{-1}(Th)^{-3/2}\rho_n^{3/2}\sqrt{\log n}
\]
is equivalent to
\[
    \rho_n
    \gtrsim
    n^{-2/3}(Th)^{-1}\log^{5/3}n.
\]
This condition is implied by \Cref{assump: sparsity}. Therefore,
\[
    R_{a,c}^{(i)}(t)
    =
    \tilde O_{p,2}\left\{
        n^{-1}(Th)^{-3/2}\rho_n^{3/2}\sqrt{\log n}
    \right\}
    +
    \tilde O_{p,2}\left\{
        n^{-1}(Th)^{-1}\rho_n^2\log n
    \right\}.
\]
Combining the bounds for $R_{a,q}^{(i)}(t)$ and $R_{a,c}^{(i)}(t)$ proves the
lemma.
\end{proof}

\subsection{Proof of \Cref{prop:studentization-expansion}}
\label{Proof of Theorem prop:studentization-expansion}
\begin{proof}
For brevity, suppress the dependence on $t$.  Define the bias term
\(
    b_n:=\frac{\mu_h-\mu}{\sigma_{nh}}.
\)
By \Cref{Thm1}(a),
\[
    b_n=O\left(\sqrt n h^\nu+\frac{\sqrt n}{Th}\right).
\]
Next, define
\(
    z_n:=\hat\delta_n+\delta_n.
\)
By \Cref{Thm1}(b)--(c),
\[
    z_n
    =
    \check\delta_n
    +
    \tilde O_{p,1}(\widetilde{\mcalM} (n, \rho_n, T, h)),
\]
where
\[
\begin{aligned}
\widetilde{\mcalM} (n, \rho_n, T, h)
&:=
n^{-1}\log n
+n^{-1}(\rho_n Th)^{-1}\log n
+n^{-1}(\rho_n Th)^{-1/2}\sqrt{\log n}  \\
&\quad
+n^{-1}(\rho_n Th)^{-3/2}\sqrt{\log n}.
\end{aligned}
\]
Since $\check\delta_n=\tilde O_{p,1}(n^{-1/2}\sqrt{\log n})$, Taylor's
expansion gives
\[
    (1+z_n)^{-1/2}
    =
    1-\frac12\check\delta_n+\tilde O_{p,1}(\widetilde{\mcalM} (n, \rho_n, T, h)).
\]
Therefore, together with the decomposition in \eqref{eq:U expansion},
\[
\begin{aligned}
\hat T_{nh}
&=
\left(
    L_n^{(1)}+L_n^{(2)}+\check\Delta_n+R_{n}+b_n
\right)
\left(
    1-\frac12\check\delta_n+\tilde O_{p,1}(\widetilde{\mcalM} (n, \rho_n, T, h))
\right)  \\
&=
L_n^{(1)}+L_n^{(2)}-\frac12L_n^{(1)}\check\delta_n
+\check\Delta_n
+
\tilde O_{p,1}\{\mcalM(n,\rho_n,T,h)\}.
\end{aligned}
\]
Here, the terms involving $R_{n}$, $b_n$,
$\Delta_n\check\delta_n$, $\check\Delta_n\check\delta_n$, and the Taylor
remainder are all absorbed into $\mcalM(n,\rho_n,T,h)$ by its definition and \Cref{assump: sparsity}. Thus,
\[
    \hat T_{nh}(t)
    =
    L_n(t)+\check\Delta_n(t)
    +
    \tilde O_{p,1}\{\mcalM(n,\rho_n,T,h)\},
\]
where
\[
    L_n(t)
    =
    L_n^{(1)}(t)+L_n^{(2)}(t)-\frac12L_n^{(1)}(t)\check\delta_n(t).
\]
This proves the lemma.
\end{proof}

\subsection{Proof of \Cref{T_n U}}
\label{Proof of T_n U}

\begin{proof}
To simplify notation in this proof only, define the shorthands
\[
    g_i:=g_{1h}(X_i,t),
    \qquad
    g_{ij}:=g_{2h}(X_i,X_j,t),
    \qquad
    \xi:=\xi_{1h}(t).
\]
Recall that
\[
    L_n^{(1)}(t)=\frac{1}{\sqrt n\,\xi}\sum_{i=1}^n g_i=\tilde O_{p, 2}(\log^{1/2} n),
    \quad
    L_n^{(2)}(t)
    =
    \frac{2}{\sqrt n(n-1)\xi}
    \sum_{i<j}g_{ij}=\tilde O_{p, 2}(n^{-1/2}\log n),
\]
and
\[
\check\delta_n(t)
=
\frac1n\sum_{i=1}^n\frac{g_i^2-\xi^2}{\xi^2}
+
\frac{4}{n(n-1)}
\sum_{i\ne j}\frac{g_i g_{ij}}{\xi^2}.
\]
We treat the two summands in \(\check\delta_n(t)\) separately.

For the first product term, we have
\[
\begin{aligned}
L_n^{(1)}
\left\{
    \frac1n\sum_{j=1}^n\frac{g_j^2-\xi^2}{\xi^2}
\right\}
&=
\frac{1}{n^{3/2}\xi^3}
\sum_{i=1}^n\sum_{j=1}^n
g_i(g_j^2-\xi^2)                                      \\
&=
\frac{1}{n^{3/2}\xi^3}
\sum_{i=1}^n g_i(g_i^2-\xi^2)          +
\frac{1}{n^{3/2}\xi^3}
\sum_{i<j}
\left[
    g_i(g_j^2-\xi^2)+g_j(g_i^2-\xi^2)
\right].
\end{aligned}
\]
Since
\(
    \bbE\{g_1(g_1^2-\xi^2)\}=\bbE[g_1^3],
\)
Bernstein's inequality gives
\[
    \frac{1}{n^{3/2}\xi^3}
    \sum_{i=1}^n g_i(g_i^2-\xi^2)
    =
    \frac{\bbE[g_{1h}^3(X_1,t)]}{\sqrt n\,\xi_{1h}^3(t)}
    +
    \tilde O_{p,1}(n^{-1}\log^{1/2}n).
\]
Thus, the first part contributes a deterministic shift, while the off-diagonal part is a symmetric degree-two term.

For the second summand in \(\check\delta_n(t)\), Hoeffding decomposition gives
\[
\frac{4}{n(n-1)}
\sum_{i\ne j}\frac{g_i g_{ij}}{\xi^2}
=
\frac{4}{n\xi^2}\sum_{i=1}^n\zeta_h(X_i,t)
+
\tilde O_{p,1}(n^{-1}\log n),
\]
where the first Hoeffding projection is defined as
\[
    \zeta_h(x,t)
    :=
    \bbE\{g_{1h}(X_2,t)g_{2h}(X_1,X_2,t)\mid X_1=x\},
\]
satisfying 
\(
    \bbE\{\zeta_h(X_1,t)\}=0\) and
\(
    \|\zeta_h(\cdot,t)\|_\infty=O(\rho_n^6).
\)
Since
\(
    L_n^{(1)}=\tilde O_{p,1}(\log^{1/2}n),
\)
the product of \(L_n^{(1)}\) with this degenerate remainder is
\(
    \tilde O_{p,1}(n^{-1}\log^{3/2}n).
\)
The product of \(L_n^{(1)}\) with the first Hoeffding projection is
\[
\begin{aligned}
L_n^{(1)}
\left\{
    \frac{4}{n\xi^2}
    \sum_{j=1}^n\zeta_h(X_j,t)
\right\}
&=
\frac{4}{n^{3/2}\xi^3}
\sum_{i=1}^n\sum_{j=1}^n
g_i\zeta_h(X_j,t)                                      \\
&=
\frac{4}{n^{3/2}\xi^3}
\sum_{i=1}^n g_i\zeta_h(X_i,t)                         +
\frac{4}{n^{3/2}\xi^3}
\sum_{i<j}
\left[
    g_i\zeta_h(X_j,t)+g_j\zeta_h(X_i,t)
\right].
\end{aligned}
\]
Note that
\[
\begin{aligned}
\bbE\{g_1\zeta_h(X_1,t)\}
&=
\bbE\{
    g_{1h}(X_1,t)g_{1h}(X_2,t)
    g_{2h}(X_1,X_2,t)
\}.
\end{aligned}
\]
Therefore, Berstein's inequality gives
\[
\begin{aligned}
\frac{4}{n^{3/2}\xi^3}
\sum_{i=1}^n g_i\zeta_h(X_i,t)
&=
\frac{4}{\sqrt n\,\xi_{1h}^3(t)}
\bbE[
    g_{1h}(X_1,t)g_{1h}(X_2,t)
    g_{2h}(X_1,X_2,t)
]                                                     +
\tilde O_{p,1}(n^{-1}\log^{1/2}n).
\end{aligned}
\]

Combining these expansions with
\(
    L_n
    =
    L_n^{(1)}
    +
    \Delta_n
    -
    \frac12L_n^{(1)}\check\delta_n,
\)
we obtain
\[
\begin{aligned}
L_n(t)
&=
-\frac{1}{\sqrt n\,\xi_{1h}^3(t)}
\left\{
    \frac12\bbE[g_{1h}^3(X_1,t)]
    +
    2\bbE[
        g_{1h}(X_1,t)g_{1h}(X_2,t)g_{2h}(X_1,X_2,t)
    ]
\right\}                                               \\
&\quad+
\frac{1}{\sqrt n\,\xi_{1h}(t)}
\sum_{i=1}^n g_{1h}(X_i,t)                              \\
&\quad+
\frac{2}{\sqrt n(n-1)}
\sum_{i<j}\tilde g_{2h}(X_i,X_j,t)
+
\tilde O_{p,1}(n^{-1}\log^{3/2}n),
\end{aligned}
\]
where the degree-two kernel is defined as
\begin{align*}
\tilde g_{2h}(x,y,t)
&:=
\frac{g_{2h}(x,y,t)}{\xi_{1h}(t)}                 \nonumber      \\
&\quad
-\frac{n-1}{4n\,\xi_{1h}^3(t)}
\Big[
    g_{1h}(x,t)\{g_{1h}^2(y,t)-\xi_{1h}^2(t)\}
    +
    g_{1h}(y,t)\{g_{1h}^2(x,t)-\xi_{1h}^2(t)\}
\Big]                                             \nonumber      \\
&\quad
-\frac{n-1}{n\,\xi_{1h}^3(t)}
\Big[
    g_{1h}(x,t)\zeta_h(y,t)
    +
    g_{1h}(y,t)\zeta_h(x,t)
\Big]. 
\end{align*}

It remains only to check the stated properties of \(\tilde g_{2h}\). The first
term is degenerate because \(g_{2h}\) is the second Hoeffding projection. For
the second term, conditional on \(X_1=x\), we have
\(
    \bbE\{g_{1h}^2(X_2,t)-\xi_{1h}^2(t)\}=0\) and
    \(
    \bbE\{g_{1h}(X_2,t)\}=0.
\)
For the third term, we have
\(
    \bbE\{\zeta_h(X_2,t)\}=0\) and
\(
    \bbE\{g_{1h}(X_2,t)\}=0.
\)
It follows that
\(
    \bbE\{\tilde g_{2h}(X_1,X_2,t)\mid X_1\}=0.
\)
Finally, using
\(
    \|g_{1h}(\cdot,t)\|_\infty+\|g_{2h}(\cdot,\cdot,t)\|_\infty
    \le C\rho_n^3\),
\(
    \xi_{1h}(t)\asymp \rho_n^3\), and
\(
    \|\zeta_h(\cdot,t)\|_\infty=O(\rho_n^6),
\)
we obtain
\[
    \|\tilde g_{2h}(\cdot,\cdot,t)\|_\infty\le C,
    \qquad
    \bbE\{\tilde g_{2h}^2(X_1,X_2,t)\}=O(1).
\]
This proves the lemma.

\end{proof}

\subsection{Proof of \Cref{8.55}}
\label{Proof of 8.55}

\begin{proof}[\textbf{Proof of Part (a)}]
Recall that
\[
    Q(t)=\binom n2^{-1}\sum_{i<j}\theta_{ij}(t)\tilde\eta_{ij}(t),
    \qquad
     \theta_{ij}(t):=\frac{3}{n-2}\sum_{k\neq i,j}
        \tilde W_{ik}(t)\tilde W_{jk}(t).
\]
Let
\(
    w_\ell(t):=K_h(t-t_\ell)/\sum_{\ell=1}^T K_h(t-t_\ell)
\),
then
\[
    \check\Delta_n(t)
    =
    \frac{Q(t)}{\sigma_{nh}(t)}
    = \sigma_{nh}^{-1}(t)
    \binom n2^{-1} 
    \sum_{i<j}\sum_{\ell=1}^T
    w_\ell(t)
      \theta_{ij}(t)
    \eta_{ij}(t_\ell).
\]
Conditional on $X$, the variables
$\{\eta_{ij}(t_\ell):1\le i<j\le n,1\le \ell\le T\}$ are independent and
centered.  
Therefore,
\[
\operatorname{var}\{\check\Delta_n(t)\mid X\}
=\sigma_{nh}^{-2}(t)
 {n\choose 2}^{-2}
\sum_{i<j} \sum_{\ell=1}^T
    w^2_\ell(t)   \theta_{ij}^2 
\operatorname{var}\{A_{ij}(t_\ell)\mid X_i, X_j\}.
\]

We first show that
\[
    \sigma_X^2
    :=
    n\rho_n\operatorname{var}\{\check\Delta_n(t)\mid X\}
    \asymp
    (Th)^{-1}
\]
with probability at least \(1-O(n^{-1})\).
Under the model,
\[
\begin{aligned}
\operatorname{var}\{A_{ij}(s)\mid X_i=x,X_j=y\}
&=
\rho_nF(x,y,s)
\left[
1-\rho_nF(x,y,s)\{1-2G(x,y,s)\}^2
\right].
\end{aligned}
\]
Since \(F\) is uniformly bounded, \(0\le G\le1\), and \(\rho_n=o(1)\),
\[
    1-\rho_nF(x,y,s)\{1-2G(x,y,s)\}^2
    \ge
    1-\rho_n\|F\|_\infty
    \ge C_0
\]
with some constant \(C_0>0\) for large \(n\). Therefore,
\(
   \operatorname{var}\{A_{ij}(s)\mid X_i=x,X_j=y\}
    \gtrsim
    \rho_nF(x,y,s).
\)
Next, let
\(
    f_{2h}(x_1,x_2,t)
    :=
    \bbE \left[\tilde W_{13}(t)\tilde W_{23}(t)\mid X_1=x_1, X_2=x_2\right].
\)
Then, by Bernstein's inequality and a union bound, we have
\[
     \theta_{ij}(t)
    =
    3f_{2h}(X_i,X_j,t)+r_{ij}(t),
\qquad \text{with} \quad \max_{i<j}|r_{ij}(t)|
    =
    O\left(\rho_n^2\sqrt{\frac{\log n}{n}}\right),
\]
with probability at least \(1-O(n^{-1})\).
Define the non-smoothed version of \(f_{2h}\):
\[
f_{2}(x_1,x_2,t):=
    \bbE \left[ W_{13}(t) W_{23}(t)\mid X_1=x_1, X_2=x_2\right].
\]
Then, by definition, we have
\[
    g_1(x,t)
    =
    \bbE[W_{12}(t)f_{2}(X_1,X_2,t)\mid X_1=x]-\mu(t).
\]
By Cauchy-Schwarz inequality, we have
\[
\begin{aligned}
& \left[
\bbE[W_{12}(t)f_{2}(X_1,X_2,t)\mid X_1=x]
\right]^2\\
\le & 
\bbE\left[
    \frac{W_{12}^2(t)}{F(X_1,X_2,t)}
    \mathbf 1\{F(X_1,X_2,t)>0 \} \mid X_1=x
\right]
\bbE\left[
    F(X_1,X_2,t)f_{2}^2(X_1,X_2,t) \mid X_1=x
\right] \\
\le & 
C\rho_n^2
\bbE\left[
    F(X_1,X_2,t)f_{2}^2(X_1,X_2,t) \mid X_1=x
\right],
\end{aligned}
\]
where the third inequality uses the fact that $W_{12}(t)=\rho_nF(X_1,X_2,t)\{1-2G(X_1,X_2,t)\}$ and the boundedness of $F$ in \Cref{assump: smooth graphon}.
It follows that
\[
    \xi_1^2(t)
    \le
    C\rho_n^2
    \bbE[F(X_1,X_2,t)f_{2}^2(X_1,X_2,t)].
\]
Since \(\xi_1^2(t)\ge C\rho_n^6\) by \Cref{non-u}, we obtain
\[
    \bbE[F(X_1,X_2,t)f_{2}^2(X_1,X_2,t)]
    \gtrsim
    \rho_n^4.
\]
By the smoothness of \(F\) and \Cref{kappa j}, this also implies
\[
    \bbE\left[
        f_{2h}^2(X_1,X_2,t)
            \frac{
    \sum_{\ell=1}^T K_h^2(t-t_\ell)F(X_i,X_j,t_\ell)
}{
    \sum_{\ell=1}^T K_h^2(t-t_\ell)
}
    \right]
    \gtrsim
    \rho_n^4.
\]
A bounded U-statistic concentration inequality then yields
\[
\begin{aligned}
&{n\choose2}^{-1}
\sum_{i<j}
f_{2h}^2(X_i,X_j,t)
\frac{
    \sum_{\ell=1}^T K_h^2(t-t_\ell)F(X_i,X_j,t_\ell)
}{
    \sum_{\ell=1}^T K_h^2(t-t_\ell)
}
\gtrsim
\rho_n^4
\end{aligned}
\]
with probability at least \(1-O(n^{-1})\). Combining this with
\(\theta_{ij}=3f_{2h}(X_i,X_j)+r_{ij}\) and the bound on \(r_{ij}\), we get
\[
\begin{aligned}
&{n\choose2}^{-1}
\sum_{i<j}
\theta_{ij}^2(t)
\frac{
    \sum_{\ell=1}^T K_h^2(t-t_\ell)F(X_i,X_j,t_\ell)
}{
    \sum_{\ell=1}^T K_h^2(t-t_\ell)
}
\gtrsim
\rho_n^4
\end{aligned}
\]
with probability at least \(1-O(n^{-1})\).

Using this lower bound, \(\sum_\ell K_h(t-t_\ell)\asymp T\),
\(\sum_\ell K_h^2(t-t_\ell)\asymp T/h\), and
\(\sigma_{nh}^2(t)=9\xi_{1h}^2(t)/n\asymp\rho_n^6/n\), we obtain
\[
\begin{aligned}
\sigma_X^2
&=
n\rho_n \sigma_{nh}^{-2}(t)
{n\choose2}^{-2}
\left\{\sum_{m=1}^TK_h(t-t_m)\right\}^{-2}
\sum_{i<j}\sum_{\ell=1}^T
K_h^2(t-t_\ell)
\theta_{ij}^2(t)
\operatorname{var}\{A_{ij}(t_\ell)\mid X\}  \\
&\gtrsim
n\rho_n
\left(\frac{n}{\rho_n^6}\right)
n^{-4}
\rho_n
\frac{T/h}{T^2}
\left(n^2\rho_n^4\right) \\
&\gtrsim
(Th)^{-1}.
\end{aligned}
\]
The matching upper bound follows from
\(|\theta_{ij}(t)|\le C\rho_n^2\) and
\(\operatorname{var}\{A_{ij}(t_\ell)\mid X\}\le C\rho_n\). Hence,
\(
    \sigma_X^2\asymp (Th)^{-1}
\)
with probability at least \(1-O(n^{-1})\).

We now apply the Berry--Esseen inequality conditionally on \(X\). On the event above,
\[
    \operatorname{var}\{\check\Delta_n(t)\mid X\}
    =
    (n\rho_n)^{-1}\sigma_X^2
    \asymp
    (n\rho_nTh)^{-1}.
\]
Define
\(
    Z_{ij\ell}:=\sigma_{nh}^{-1}(t) \binom n2^{-1}w_\ell(t)\theta_{ij}(t)\eta_{ij}(t_\ell)\), then
        \(
    \check\Delta_n(t)=\sum_{i<j}\sum_{\ell=1}^T Z_{ij\ell}.
\)
By
\Cref{basic moment,kappa j},
we have 
\(
    |Z_{ij\ell}|
    \lesssim
    n^{-3/2}(\rho_n Th)^{-1},
\) uniformly in $i,j,\ell$. 
Therefore, 
\[\frac{|Z_{ij\ell}|}{\sqrt{\operatorname{var}\{\check\Delta_n(t)\mid X\}}} \lesssim n^{-1}(\rho_n Th)^{-1/2} =o(1),\]
where the last step follows from \Cref{assump: sparsity}. 
By Theorem 2.1 of \cite{chen2001non}, $\left\Vert \mcalF_{\frac{\check\Delta_n(t)}{\sqrt{var(\check\Delta_n(t)|X)}}}-\Phi\right\Vert_\infty$ is upper bounded by the third moment of each term in the summands of $\check\Delta_n(t)/\sqrt{var(\check\Delta_n(t)|X)}$:
\[
\begin{aligned}
\left\|
\mcalF_{\check\Delta_n\mid X}
-
\mcalF_{\mcalN(0,(n\rho_n)^{-1}\sigma_X^2)}
\right\|_\infty  & \lesssim
\sum_{i<j}\sum_{\ell=1}^T
    \frac{\bbE\{|Z_{ij\ell}|^3\mid X\}}{
\operatorname{var}^{3/2}\{\check\Delta_n(t)\mid X\}
}\\
& \lesssim \frac{n^2\{\sum_{\ell=1}^T |K_h(t-t_\ell)|^3\}
        n^{-6}T^{-3}n^{3/2}\rho_n^{-3}\rho_n}
        {n^{-3/2}(\rho_n Th)^{-3/2}}\\
& \lesssim n^{-1}(\rho_n Th)^{-1/2},
\end{aligned}
\]
because $\sum_{\ell=1}^T |K_h(t-t_\ell)|^3=O(Th^{-2})$ by \Cref{kappa j}. 
This proves the bound on the above event and hence in the stated $\tilde O_{p,1}$ sense.

\end{proof}

\begin{proof}[\textbf{Proof of Part (b)}]
Conditional on $X$ and $\delta_T$, the random variable $J_n(t)+\delta_T$ is
fixed, whereas $\check\Delta_n(t)$ and $\tilde\Delta_n(t)$ have conditional
CDFs whose Kolmogorov distance is bounded by part (a).  Therefore, uniformly
in $u$,
\[
\begin{aligned}
&\left|
\mcalF_{J_n+\check\Delta_n+\delta_T}(u)
-\mcalF_{J_n+\tilde\Delta_n+\delta_T}(u)
\right| \\
&\qquad=
\left|
\bbE\left[
\mcalF_{\check\Delta_n\mid X}\{u-J_n-\delta_T\}
-
\mcalF_{\tilde\Delta_n\mid X}\{u-J_n-\delta_T\}
\right]
\right|                                      \\
&\qquad\le C n^{-1}(\rho_nTh)^{-1/2}+O(n^{-1}),
\end{aligned}
\]
where the $O(n^{-1})$ term accounts for the exceptional event in the
$\tilde O_{p,1}$ bound. This proves part (b). 
\end{proof}

\subsection{Proof of \Cref{8.54}}
\label{Proof of 8.54}
\begin{proof}
    Let
\(
    Y:=J_n(t)+\check\Delta_n(t)+\delta_T\) and
\(
    Z:=\hat T_{nh}(t)-J_n(t)-\check\Delta_n(t).
\)
Then $\hat T_{nh}(t)+\delta_T=Y+Z$.  By
\Cref{prop:studentization-expansion}, we have
\(
    Z=\tilde O_{p,1}\{\mcalM(n,\rho_n,T,h)\}.
\)
It remains to verify the smoothness condition in \Cref{random perturbation} for $Y$.

Let $\tilde Y:=J_n(t)+\tilde\Delta_n(t)+\delta_T$.  For any $u\in\bbR$ and
$a>0$,
\[
\begin{aligned}
\mcalF_Y(u+a)-\mcalF_Y(u)
&\le
2\|\mcalF_Y-\mcalF_{\tilde Y}\|_\infty
+2\|\mcalF_{\tilde Y}-G_{nh}\|_\infty
+G_{nh}(u+a)-G_{nh}(u).
\end{aligned}
\]
By \Cref{8.55,8.56}, the first two terms are bounded by
\[
    O\{n^{-1}(\rho_nTh)^{-1/2}+n^{-1}\}
    +O\{n^{-1}\log n+(n\rho_nTh)^{-1}\log n\},
\]
respectively.
The Edgeworth approximation $G_{nh}$ has a uniformly bounded density, because
it is a standard normal density multiplied by a fixed-degree polynomial plus $\Phi$.  Hence,
\[
    G_{nh}(u+a)-G_{nh}(u)\le Ca,
\]
for some constant $C$.  
Thus, the smoothness condition in \Cref{random perturbation} holds with
\[
    \zeta_n=
    n^{-1}\log n+n^{-1}(\rho_nTh)^{-1/2}+(n\rho_nTh)^{-1}\log n,
\]
which is bounded by $\mcalM(n,\rho_n,T,h)$.
Applying \Cref{random perturbation} gives
\[
    \left\|
    \mcalF_{\hat T_{nh}+\delta_T}
    -\mcalF_{J_n+\check\Delta_n+\delta_T}
    \right\|_\infty
    =O\{\mcalM(n,\rho_n,T,h)\}.
\]

\end{proof}

\subsection{Proof of Lemma \ref{8.56}}
\label{Proof of Lemma 8.56}
We first state four auxiliary bounds (\Cref{a,b,c,d}) for the proof of \Cref{8.56}. Let $Ch.f.(G;s):=\int_{-\infty}^\infty e^{isx}dG(x)$ denote the characteristic function for any function $G$.

\begin{lemma}
    \label{a}Under \Cref{assump: time obs,assump: smooth graphon,non-u,assump: sparsity},
    for any fixed $\epsilon'=\frac{\epsilon}{2}>0$, we  have $$\int_{n^{\epsilon'}}^n\left\vert\frac{Ch.f.(G_{nh};s)}{s}\right\vert ds=O(n^{-1}).$$
\end{lemma}

\begin{lemma}\label{b}
    Under \Cref{assump: time obs,assump: smooth graphon,non-u,assump: sparsity}, for any fixed $C_1>0$ and sufficiently large $c_\delta$, we have $$\int_{C_1\sqrt n}^n\left\vert\frac{\bbE[\exp(is(J_n+\tilde \Delta_n+\delta_T))]}{s}\right\vert ds=O(n^{-1}\log n)$$
\end{lemma}

\begin{lemma}
    \label{c}
    Assume \Cref{assump: time obs,assump: smooth graphon,non-u,assump: sparsity}.
    Let $C_1>0$ be a sufficiently small constant and $\epsilon'=\frac{\epsilon}{2}>0$ be a fixed constant. %
we have $$\int_{n^{\epsilon'}}^{C_1\sqrt n}\left\vert\frac{\bbE[\exp(is(J_n+\tilde \Delta_n+\delta_T))]}{s}\right\vert ds=O((n\rho_nTh)^{-1}\log n+n^{-1}\log n).$$
\end{lemma}

\begin{lemma}
\label{d}
    Under \Cref{assump: time obs,assump: smooth graphon,non-u,assump: sparsity}, for fixed $\epsilon'=\frac{\epsilon}{2}>0$ chosen such that $\epsilon'\leq1/7$, %
    we have $$\int_0^{n^{\epsilon'}}\left\vert\frac{\bbE[e^{is(J_n+\tilde\Delta_n+\delta_T)}]-Ch.f.(G_{nh};s)}{s}\right\vert ds=O(n^{-1}\log n+(n\rho_nTh)^{-1}).$$ 
\end{lemma}

The proofs of \Cref{a,b,c,d} are deferred to \Cref{Proof of a}, \Cref{Proof of b}, \Cref{Proof of c}, \Cref{Proof of d}, respectively.

\begin{proof}[\textbf{Proof of \Cref{8.56}}]
   We apply \Cref{smoothing lemma} with $\gamma=n$, and split the smoothing integral into
\[
    [0,n^{\epsilon'}],\qquad
    [n^{\epsilon'},C_1\sqrt n],\qquad
    [C_1\sqrt n,n].
\]
The first interval is controlled by \Cref{d}; the second and third intervals
are controlled by \Cref{c,b}; and the tail of the Edgeworth characteristic
function is controlled by \Cref{a}.  Since $\sup_u |G_{nh}'(u)|=O(1)$, the
smoothing remainder is $O(n^{-1})$.  Combining the four bounds gives
\[
    \left\|
        \mcalF_{J_n+\tilde\Delta_n+\delta_T}-G_{nh}
    \right\|_\infty
    =O\{n^{-1}\log n+(n\rho_nTh)^{-1}\log n\}.
\]
\end{proof}

\subsubsection{Proof of \Cref{a}}
\label{Proof of a}
\begin{proof}
The characteristic function associated with the Edgeworth expansion has the
form
\[
    Ch.f.(G_{nh};s)=e^{-s^2/2}P(s),
\]
where $P$ is a fixed-degree polynomial.  For every fixed integer $d\ge -1$,
there exists $C_d<\infty$ such that $|s|^d e^{-s^2/2}\le C_d e^{-s^2/3}$ for
$|s|>1$.  Therefore,
\[
    \int_{n^{\epsilon'}}^n
    \left|\frac{Ch.f.(G_{nh};s)}{s}\right|ds
    \le C\int_{n^{\epsilon'}}^\infty e^{-s^2/3}ds
    =O(n^{-1}).
\]
\end{proof}

\subsubsection{Proof of \Cref{b}}
\label{Proof of b}
\begin{proof}
Conditional on $X$, $\tilde\Delta_n\sim N\{0,(n\rho_n)^{-1}\sigma_X^2\}$.
Since $\delta_T\sim N(0,c_\delta n^{-1}\log n)$ is independent of all other
variables,
\[
\begin{aligned}
\bbE e^{is(J_n+\tilde\Delta_n+\delta_T)}
&=
\exp\left(-\frac12 c_\delta n^{-1}\log n\,s^2\right)
\bbE[e^{isJ_n}\bbE[e^{is\tilde\Delta_n}|X]]\\
&=
\exp\left(-\frac12 c_\delta n^{-1}\log n\,s^2\right)
\bbE\left[e^{isJ_n}
\exp\left\{-\frac12(n\rho_n)^{-1}\sigma_X^2s^2\right\}\right].
\end{aligned}
\]
The absolute value of the expectation is at most one.  Hence, for
$|t|\ge C_1\sqrt n$,
\[
    \left|\bbE e^{is(J_n+\tilde\Delta_n+\delta_T)}\right|
    \le \exp\left(-\frac12 c_\delta C_1^2\log n\right).
\]
Choosing $c_\delta$ large enough makes the above $O(n^{-1})$.  Thus, we have
\[
    \int_{C_1\sqrt n}^n
    \left|
        \frac{\bbE e^{is(J_n+\tilde\Delta_n+\delta_T)}}{s}
    \right|ds
    \le O(n^{-1})\int_{C_1\sqrt n}^n s^{-1}ds
    =O(n^{-1}\log n).
\]

\end{proof}

\subsubsection{Proof of \Cref{c}}
\label{Proof of c}

We first collect some facts about the conditional variance
\(\sigma_X^2=n\rho_n\,var(\check\Delta_n\mid X)\) that will be used for proving \Cref{c}.

\begin{lemma}[Hoeffding decomposition of $\sigma_X^2$ and lower-tail bound]
    \label{variance U}
Under \Cref{assump: time obs,assump: smooth graphon,non-u,assump: sparsity}, we have
\[
    \sigma_X^2
    =
    \bbE[\sigma_X^2]
    +
    \frac1n\sum_{i=1}^n g_{\sigma;1}(X_i)
    +
    R_{\sigma,n},
    \qquad \text{with} \quad 
    R_{\sigma,n}
    =
    \tilde O_{p,1}\{n^{-1}(Th)^{-1}\log n\}.
\]
where $g_{\sigma;1}$ is defined in \eqref{eq:def-gsigma1} with \(
    \|g_{\sigma;1}\|_\infty\le C(Th)^{-1}
\), \(|R_{\sigma,n}|\le C(Th)^{-1}\) for some constant $C>0$, and
\(
    \bbE[\sigma_X^2]\asymp (Th)^{-1}
\).
Moreover, there exists a positive constant $c_1>0$ such that
\[
    \bbP\left(
        \bbE[\sigma_X^2]
        +
        \frac1n\sum_{i=1}^n g_{\sigma;1}(X_i)
        \le 0
    \right)
    \le
    \exp(-c_1 n).
\]
\end{lemma}

\begin{proof}[\textbf{Proof of \Cref{variance U}}]
    Recall from the proof of \Cref{8.55} that, conditional on \(X\),
\[
\begin{aligned}
\sigma_X^2
&=
n\rho_n\sigma_{nh}^{-2}
{n\choose2}^{-2}
S_K(t)^{-2}
\sum_{i<j}\sum_{\ell=1}^T
K_h^2(t-t_\ell)
\theta_{ij}^2(t)
var\{A_{ij}(t_\ell)\mid X\},
\end{aligned}
\]
where
\[
    S_K(t)=\sum_{\ell=1}^T K_h(t-t_\ell),
    \qquad
    \theta_{ij}(t)=\frac{3}{n-2}\sum_{k\ne i,j}
    \tilde W_{ik}(t)\tilde W_{jk}(t).
\]
For ease of proof, we define the shorthands
\[
c_{\sigma}(n,\rho_n, T,h):= n\rho_n\sigma_{nh}^{-2}
{n\choose2}^{-2}
S_K(t)^{-2}, 
\qquad
v_{ij}(t_\ell):= var\{A_{ij}(t_\ell)\mid X_i,X_j\},
\]
and
\[
    f_{ij}(t)
    :=
    \bbE[\tilde W_{ik}(t)\tilde W_{jk}(t) \mid X_i, X_j],
    \qquad 
    \psi_{ij;k}(t)
    :=
    \tilde W_{ik}(t)\tilde W_{jk}(t)-f_{ij}(t).
\]
Then, we have for each \(i<j\),
\[
    \theta_{ij}(t)
    =
    \frac{3}{n-2}\sum_{k\ne i,j}
    \tilde W_{ik}(t)\tilde W_{jk}(t)
    =
    3f_{ij}(t)
    +
    \frac{3}{n-2}\sum_{k\ne i,j}
    \psi_{ij;k}(t).
\]
Substituting the decomposition of \(\theta_{ij}(t)\) into the exact formula
for \(\sigma_X^2\) gives
\begin{align*}
\sigma_X^2
&=
c_{\sigma}(n,\rho_n, T,h)
\sum_{i<j}\sum_{\ell=1}^T
K_h^2(t-t_\ell)
\left\{
    3f_{ij}(t)
    +
    \frac{3}{n-2}\sum_{k\ne i,j}\psi_{ij;k}(t)
\right\}^2
v_{ij}(t_\ell)                                                  \\
&=
\underbrace{
c_{\sigma}(n,\rho_n, T,h)
\sum_{i<j}\sum_{\ell=1}^T
K_h^2(t-t_\ell)
9f_{ij}^2(t)v_{ij}(t_\ell)
}_{=:U_{\sigma,2}}
\\
&\quad+
\underbrace{
c_{\sigma}(n,\rho_n, T,h)
\sum_{i<j}\sum_{\ell=1}^T
K_h^2(t-t_\ell)
\frac{18}{n-2}f_{ij}(t)
\left(\sum_{k\ne i,j}\psi_{ij;k}(t)\right)
v_{ij}(t_\ell)
}_{=:U_{\sigma,3}}
\\
&\quad+
\underbrace{
c_{\sigma}(n,\rho_n, T,h)
\sum_{i<j}\sum_{\ell=1}^T
K_h^2(t-t_\ell)
\frac{9}{(n-2)^2}
\left(\sum_{k\ne i,j}\psi_{ij;k}(t)\right)^2
v_{ij}(t_\ell)
}_{=:R_{\sigma,0}} .
\end{align*}
The first term 
\(
    U_{\sigma,2}
    =
    {n\choose2}^{-1}
    \sum_{i<j}H_{\sigma,2}(X_i,X_j)
\)
is a degree-two U-statistic,
where
\[
    H_{\sigma,2}(X_i,X_j)
    :=
    {n\choose2}c_{\sigma}(n,\rho_n, T,h)
    \sum_{\ell=1}^T
    K_h^2(t-t_\ell)
    9f_{ij}^2(t)v_{ij}(t_\ell).
\]
The second term
\(
    U_{\sigma,3}
    =
    {n\choose3}^{-1}
    \sum_{i<j<k}H_{\sigma,3}(X_i,X_j,X_k)
\)
 is a degree-three U-statistic after
symmetrization,
where
\[
\begin{aligned}
H_{\sigma,3}(X_i,X_j,X_k)
&:=
{n\choose3}c_{\sigma}(n,\rho_n, T,h)\frac{18}{n-2}
\sum_{\ell=1}^T K_h^2(t-t_\ell)                         \\
&\quad\times
\Big[
    f_{ij}(t)\psi_{ij;k}(t)v_{ij}(t_\ell) +
    f_{ik}(t)\psi_{ik;j}(t)v_{ik}(t_\ell) +
    f_{jk}(t)\psi_{jk;i}(t)v_{jk}(t_\ell)
\Big].
\end{aligned}
\]
Now we  define \(g_{\sigma;1}\) as the first Hoeffding projection of
\(U_{\sigma,2}+U_{\sigma,3}\):
\begin{equation}
\label{eq:def-gsigma1}
\begin{aligned}
g_{\sigma;1}(x)
&:=
2\Big[
    \bbE\{H_{\sigma,2}(X_1,X_2)\mid X_1=x\}
    -
    \bbE\{H_{\sigma,2}(X_1,X_2)\}
\Big]                                                  \\
&\quad+
3\Big[
    \bbE\{H_{\sigma,3}(X_1,X_2,X_3)\mid X_1=x\}
    -
    \bbE\{H_{\sigma,3}(X_1,X_2,X_3)\}
\Big].
\end{aligned}
\end{equation}
Then \(\bbE\{g_{\sigma;1}(X_1)\}=0\), and using the facts 
\(
    S_K(t)\asymp T\) and \(
    \sum_{\ell=1}^T K_h^2(t-t_\ell)=O(T/h)
\) by \Cref{kappa j},
together with the bounds uniformly in $i,j,\ell$
\[
    v_{ij}(t_\ell)=O(\rho_n),\qquad
    f_{ij}(t)=O(\rho_n^2),\qquad
    \psi_{ij;k}(t)=O(\rho_n^2), \qquad
    \sigma_{nh}^2(t)\asymp \rho_n^6/n,
\]
we have
\[
    \|g_{\sigma;1}\|_\infty\le C(Th)^{-1}.
\]
The Hoeffding decomposition of the two U-statistics gives
\[
    U_{\sigma,2}+U_{\sigma,3}
    =
    \bbE(U_{\sigma,2}+U_{\sigma,3})
    +
    \frac1n\sum_{i=1}^n g_{\sigma;1}(X_i)
    +
    R_{\sigma,1},
\]
where the second- and third-order remainders satisfy
\[
    R_{\sigma,1}
    =
    \tilde O_{p,1}\{n^{-1}(Th)^{-1}\log n\}.
\]
It remains to bound \(R_{\sigma,0}\). Conditional on \(X_i,X_j\),
\(\psi_{ij;k}(t)\) is centered in \(X_k\) and uniformly
\(O(\rho_n^2)\). Bernstein's inequality, followed by a union bound over
\((i,j)\), gives
\[
    \max_{i<j}
    \left|
        \frac1{n-2}\sum_{k\ne i,j}\psi_{ij;k}
    \right|
    =
    O\left(\rho_n^2\sqrt{\frac{\log n}{n}}\right)
\]
with probability at least \(1-O(n^{-1})\). Therefore, it follows that
\[
\begin{aligned}
R_{\sigma,0}
&\le
C\,c_{\sigma}(n,\rho_n, T,h)
\sum_{i<j}\sum_{\ell=1}^T
K_h^2(t-t_\ell)
v_{ij}(t_\ell)
\left\{
    \rho_n^2\sqrt{\frac{\log n}{n}}
\right\}^2                                                \\
&=
\tilde O_{p,1}\{n^{-1}(Th)^{-1}\log n\}.
\end{aligned}
\]
Since \(R_{\sigma,0}\) is also uniformly bounded by \(C(Th)^{-1}\), its
expectation  $\bbE R_{\sigma,0}$ has the order $O\{n^{-1}(Th)^{-1}\log n\}$. Absorbing
\(R_{\sigma,0}-\bbE R_{\sigma,0}\) into the remainder and replacing
\(\bbE(U_{\sigma,2}+U_{\sigma,3})\) by
\(
    \bbE[\sigma_X^2]
\)
yields
\[
    \sigma_X^2
    =
     \bbE[\sigma_X^2]
    +
    \frac1n\sum_{i=1}^n g_{\sigma;1}(X_i)
    +
    \tilde O_{p,1}\{n^{-1}(Th)^{-1}\log n\}.
\]

Finally, the lower-bound argument in \Cref{8.55} gives
\(\sigma_X^2\gtrsim (Th)^{-1}\) with probability at least \(1-O(n^{-1})\),
while the deterministic upper bound gives \(\sigma_X^2\le C(Th)^{-1}\).
Hence,
\(
    \bbE[\sigma_X^2]\asymp (Th)^{-1}.
\)
It follows that
\(
    (Th)\bar\sigma_X^2\ge c>0\)
    and
    \(
    \|(Th)g_{\sigma;1}\|_\infty\le C.
\)
Bernstein's inequality gives

$$\begin{aligned}
    \bbP\left(\frac{1}{n}\sum_{i=1}^ng_{\sigma;1}(X_i)+\bbE[\sigma_X^2]\leq0\right)&=\bbP\left(\frac{1}{n}\sum_{i=1}^n(Th)g_{\sigma;1}(X_i)\leq-(Th)\bbE[\sigma_X^2]\right)\\
    &\leq\exp\left(-\frac{\left((Th)\bbE[\sigma_X^2]\right)^2}{\|(Th)g_{\sigma;1}\|_\infty^2}n\right)\\
    &\leq\exp(-c_1n),
    \end{aligned}$$
    which completes the proof.
\end{proof}

\begin{proof}[\textbf{Proof of \Cref{c}}]

    Fix the target network time point \(t\), and write \(s\) for the Fourier variable. We consider the range
\(
    n^{\epsilon'}\le |s|\le C_1\sqrt n,
\)
where \(C_1>0\) is chosen sufficiently small.

Since \(\delta_T\) is independent and $\vert\exp(is\delta_T)\vert=1$, we have 
\[\vert\bbE[\exp(is(J_n+\tilde\Delta_n+\delta_T))]\vert\leq \vert\bbE[\exp(is(J_{n}+\tilde\Delta_n))]\vert\bbE[\vert\exp(is\delta_T)\vert]\vert\leq\vert\bbE[\exp(is(J_{n}+\tilde\Delta_n))]\vert.\]
Conditional on $X$, by \Cref{8.55}, we have $\tilde\Delta_n \mid X \sim\mcalN(0, (n\rho_n)^{-1}\sigma_X^2)$ and $J_n$ is \(X\)-measurable. It follows that 
\[
\left|
\bbE e^{is(J_n+\tilde\Delta_n)}
\right|
= 
\left|
\bbE [e^{isJ_n} \bbE[e^{is\tilde\Delta_n}\mid X]]
\right|
=
\left|
\bbE\left[
    e^{isJ_n}
    \exp\left\{-\frac{s^2}{2n\rho_n}\sigma_X^2\right\}
\right]
\right|.
\]
By \Cref{variance U}, we have
\[
    \sigma_X^2
    =
    \bbE[\sigma_X^2]
    +
    \frac1n\sum_{i=1}^n g_{\sigma;1}(X_i)
    +
    R_{\sigma,n},
    \qquad
    R_{\sigma,n}
    =
    \tilde O_{p,1}\{n^{-1}(Th)^{-1}\log n\}.
\]
There exists a constant \(C>0\) such that the event  Let
\(
    B_n
    :=
    \left\{
        |R_{\sigma,n}|
        \le
        Cn^{-1}(Th)^{-1}\log n
    \right\},
\)
holds with probability at least \(1-O(n^{-1})\), i.e., the complement event satisfies $\bbP(B_n^c)\lesssim n^{-1}.$
Also define the event
\[
    W_n
    =
    \left\{
        \bbE[\sigma_X^2]
        +
        \frac1n\sum_{i=1}^n g_{\sigma;1}(X_i)
        \le 0
    \right\},
\]
for which \(\bbP(W_n)\le \exp(-c_1n)\) for some positive constant $c_1>0$ by \Cref{variance U}.

Next, we first bound the difference
\[
\begin{aligned}
&\left|
\bbE\left[
e^{isJ_n}
\exp\left\{
-\frac{s^2}{2n\rho_n}\sigma_X^2
\right\}
\right]
-
\bbE\left[
e^{isJ_n}
\exp\left(-\frac{s^2}{2n\rho_n}\left\{\bbE[\sigma_X^2]+\frac{1}{n}\sum_{i=1}^ng_{\sigma;1}(X_i)\right\}\right)
\right]
\right|      ,                                         
\end{aligned}
\]
and then bound the leading term.
To bound the difference, we discuss three cases:\\
(1) \underline{On \(B_n\cap W_n^c\)}, by substituting the decomposition of $\sigma_X^2$, we have 
\begin{align*}
    & \left|
e^{isJ_n}
\exp\left\{
-\frac{s^2}{2n\rho_n}\sigma_X^2
\right\}
-
e^{isJ_n}
\exp\left(-\frac{s^2}{2n\rho_n}\left\{\bbE[\sigma_X^2]+\frac{1}{n}\sum_{i=1}^ng_{\sigma;1}(X_i)\right\}\right)
\right| \\
    = &  \Bigg\vert e^{isJ_n}\exp\left(-\frac{s^2}{2n\rho_n}\left\{\bbE[\sigma_X^2]+\frac{1}{n}\sum_{i=1}^ng_{\sigma;1}(X_i)\right\}\right)\left( \exp\left(-\frac{s^2}{2n\rho_n}R_{\sigma,n}\right) - 1 \right) 
\Bigg \vert\\
\le & \left|
    \exp\left\{-\frac{s^2}{2n\rho_n}R_{\sigma,n}\right\}-1
\right| \\
\lesssim & s^2n^{-2}(\rho_n Th)^{-1}\log n,
\end{align*} 
uniformly for \(n^{\epsilon'}\le |s|\le C_1\sqrt n\). The first inequality follows from the fact that $\vert\exp(iz)\vert=1\ \forall z\in\bbC$, and the second exponential factor is bounded by one. The last inequality follows from the fact that \(\exp(z)=1+O(z)\) as \(z\to0\) and the definition of \(B_n\).\\
(2) \underline{On \(B_n^c \cap W_n^c\)}, we use the basic bounds
\begin{align*}
    & \left|
e^{isJ_n}
\exp\left\{
-\frac{s^2}{2n\rho_n}\sigma_X^2
\right\}
-
e^{isJ_n}
\exp\left(-\frac{s^2}{2n\rho_n}\left\{\bbE[\sigma_X^2]+\frac{1}{n}\sum_{i=1}^ng_{\sigma;1}(X_i)\right\}\right)
\right| \\
\le & \left|
    e^{isJ_n}
    \exp\left\{-\frac{s^2}{2n\rho_n}\sigma_X^2\right\}
\right| + \Bigg\vert e^{isJ_n}\exp\left(-\frac{s^2}{2n\rho_n}\left\{\bbE[\sigma_X^2]+\frac{1}{n}\sum_{i=1}^ng_{\sigma;1}(X_i)\right\}\right)
\Bigg \vert\\
\le & 2.
\end{align*}
(3) \underline{On \(W_n\)}, we use the basic bound
\(
\left|
    e^{isJ_n}
    \exp\left\{-\frac{s^2}{2n\rho_n}\sigma_X^2\right\}
\right|
\le 1
\) 
and apply the univeral bounds $\bbE[\sigma_X^2] \asymp (Th)^{-1}$ and $\|g_{\sigma,1}\|_\infty \lesssim (Th)^{-1}$ by \Cref{variance U} to obtain that
\[
\left|
    \exp\left\{
        -\frac{s^2}{2n\rho_n}
        \left(
            \bbE[\sigma_X^2]
            +
            \frac1n\sum_{i=1}^n g_{\sigma;1}(X_i)
        \right)
    \right\}
\right|
\le
\exp\left\{
    C\frac{s^2}{n\rho_nTh}
\right\}.
\]
Combining the above bounds under the three cases gives
\[
\begin{aligned}
&\left|
\bbE\left[
e^{isJ_n}
\exp\left\{
-\frac{s^2}{2n\rho_n}\sigma_X^2
\right\}
\right]
-
\bbE\left[
e^{isJ_n}
\exp\left\{
-\frac{s^2}{2n\rho_n}
\left(
    \bbE[\sigma_X^2]
    +
    \frac1n\sum_{i=1}^n g_{\sigma;1}(X_i)
\right)
\right\}
\right]
\right|                                                     \\
&\qquad\lesssim s^2n^{-2}(\rho_n Th)^{-1}\log n \cdot \bbP(B_n \cap W_n^c)+ 
    \bbP(B_n^c \cap W_n^c)
    + \left(\exp\left\{
    C\frac{s^2}{n\rho_nTh}
\right\}+1\right) \bbP(W_n)\\
& \qquad \lesssim s^2n^{-2}(\rho_n Th)^{-1}\log n + n^{-1} + \exp\left\{
    C\frac{s^2}{n\rho_nTh}-c_1n \right\}+\exp(-c_1n).
\end{aligned}
\]
Since \(s^2\le C_1^2n\) and \((\rho_nTh)^{-1}=o(n)\) under
\Cref{assump: sparsity}, we have
\[
    \exp\left\{
        C\frac{s^2}{n\rho_nTh}-c_1n
    \right\}
    \le
    \exp(-c_1n/2)
\]
for all sufficiently large \(n\).  Therefore,
\begin{align}
&\int_{n^{\epsilon'}}^{C_1\sqrt n}
\frac{1}{s}
\left|
\bbE\left[
e^{isJ_n}
\exp\left\{
-\frac{s^2}{2n\rho_n}\sigma_X^2
\right\}
\right]
-
\bbE\left[
e^{isJ_n}
\exp\left\{
-\frac{s^2}{2n\rho_n}
\left(
    \bbE[\sigma_X^2]
    +
    \frac1n\sum_{i=1}^n g_{\sigma;1}(X_i)
\right)
\right\}
\right]
\right|ds                                            \nonumber         \\
&\qquad =
O\left\{
    (n\rho_nTh)^{-1}\log n
    +
    n^{-1}\log n
\right\}. \label{eq:diff-bound}
\end{align}

It remains to bound
\[
\Psi_n(s):=
\bbE\left[
e^{isJ_n}
\exp\left\{
-\frac{s^2}{2n\rho_n}
\left(
    \bbE[\sigma_X^2]
    +
    \frac1n\sum_{i=1}^n g_{\sigma;1}(X_i)
\right)
\right\}
\right].
\]
We process by following the proof strategy of \cite{bickel1986edgeworth} to control the leading term.
Recall that 
\[J_n(t) = \alpha_n(t) +
\frac{1}{\sqrt n\,\xi_{1h}(t)}
\sum_{i=1}^n g_{1h}(X_i,t)
+
\frac{2}{\sqrt n(n-1)}
\sum_{i<j}\tilde g_{2h}(X_i,X_j,t),\] 
where $\alpha_n$ is a derministic term and $\tilde g_{2h}$ is a symmetric degenerate degree-two kernel defined in \eqref{tilde g2h expanded} in \Cref{T_n U}. 
For the rest of this proof, we omit the dependence on $t$ for notational simplicity.
Since $\alpha_n$ is a deterministic term and $|e^{is \alpha_n}|=1$, we have
\begin{align*}
     \left| \Psi_n(s) \right|
 = &  \Bigg|
\bbE\Bigg[
\exp\Bigg\{\frac{is}{\sqrt n\,\xi_{1h}}
\sum_{i=1}^n g_{1h}(X_i)
-\frac{s^2}{2n\rho_n}
\left(
    \bbE[\sigma_X^2]
    +
    \frac1n\sum_{i=1}^n g_{\sigma;1}(X_i)
\right)\\
& \qquad \qquad 
+
\frac{2is}{\sqrt n(n-1)}
\sum_{i<j}\tilde g_{2h}(X_i,X_j)
\Bigg\}
\Bigg]
\Bigg|.
\end{align*}
Define the shorthand notation \(
    \bar g(x):=\bbE[\sigma_X^2]+g_{\sigma;1}(x).
\) Then the above expression can be rewritten as
\[
\left|
\bbE\left[
\exp\left\{
    \sum_{i=1}^n
    \left(
        \frac{is}{\sqrt n\,\xi_{1h}}g_{1h}(X_i)
        -
        \frac{s^2}{2n^2\rho_n}\bar g(X_i)
    \right)
    +
    \frac{2is}{\sqrt n(n-1)}
    \sum_{1\le i<j\le n}\tilde g_{2h}(X_i,X_j)
\right\}
\right]
\right|.
\]
Set
\[
    \gamma_n(s)
    :=
    \bbE\left[
    \exp\left\{
        \frac{is}{\sqrt n\,\xi_{1h}}g_{1h}(X_1)
        -
        \frac{s^2}{2n^2\rho_n}\bar g(X_1)
    \right\}
    \right]
.
\]
Since \(g_{1h}(x)=O(\rho_n^3)\), \(\bar g(x)=O((Th)^{-1})\), and \(\xi_{1h}\asymp \rho_n^3\), we have \(\left| \frac{is}{\sqrt n\,\xi_{1h}}g_{1h}(X_1)
        -
        \frac{s^2}{2n^2\rho_n}\bar g(X_1)\right|\) uniformly bounded for all \(|s|\le C_1\sqrt n\). Therefore, by Taylor expansion of the exponential function, 
we have 
\begin{align*}
    \gamma_n(s)
    &=
\bbE\Bigg[
    1
    +
    \frac{is}{\sqrt n\,\xi_{1h}}g_{1h}(X_1)
    -
    \frac{s^2}{2n^2\rho_n}\bar g(X_1)
    \\
    & \qquad  + 
    \frac12 \left(
        \frac{is}{\sqrt n\,\xi_{1h}}g_{1h}(X_1)
    -
    \frac{s^2}{2n^2\rho_n}\bar g(X_1)
    \right)^2+
    O\left(
        \left| \frac{is}{\sqrt n\,\xi_{1h}}g_{1h}(X_1)\right|^3
        +
        \left| \frac{s^2}{2n^2\rho_n}\bar g(X_1) \right|^3
    \right) 
\Bigg] \\
&=
1
-
\frac{s^2}{2n^2\rho_n}\bbE[\bar g(X_1)]
- \frac{s^2}{2n \xi_{1h}^2}\bbE[g_{1h}^2(X_1)]
+
O\left(
    \frac{s^3}{n^{3/2}}
    +
    \frac{s^2}{n^2\rho_nTh}
\right) \\
&=
1
-\frac{s^2}{2n}
+
O\left(\frac{s^3}{n^{3/2}}
+
\frac{s^2}{n^2\rho_nTh}
\right),
\end{align*}
where the last two equalities hold due to the facts that \(\bbE[g_{1h}(X_1)]=0\), \(\bbE[g_{1h}^2(X_1)]=\xi_{1h}^2\), and the universal bounds for \(g_{1h}\) and \(\bar g\).
Using the fact that
\((\rho_nTh)^{-1}=o(n)\) under \Cref{assump: sparsity}, we have
\begin{align}
    |\gamma_n(s)| = 1 - \frac{s^2}{2n} + o\left(\frac{s^2}{n}\right)
    \le
    \exp\left\{-\frac{s^2}{3n}\right\}, \label{gamma_n bound}
\end{align}
over the range $s\leq C_1\sqrt n$ with a sufficiently small constant \(C_1>0\).

Following \citet[(2.17)--(2.20)]{bickel1986edgeworth}, we further define
\[
    \square_n(m)
    :=
    \sum_{i=1}^m\sum_{j=i+1}^n
    \tilde g_{2h}(X_i,X_j;t),
\]
with $m$ being a truncation index to be specified later, and it follows that 
$$\bbE[\vert\square_n(m)\vert^r]=O\left((mn)^{r/2}\right), \quad \text{ for } r\geq2.$$
By applying Taylor expansion of the factor
\(
    \exp\left\{
        \frac{2is}{\sqrt n(n-1)}\square_n(m)
    \right\}
\)
up to a fixed order \(r \ge 2\), we obtain
\begin{align}
& \left| \Psi_n(s) \right|     \nonumber    \\
       \le & 
    \sum_{v=0}^r
    \frac{1}{v!}
    \left(
        \frac{2s}{\sqrt n(n-1)}
    \right)^v
    \Bigg|
    \bbE\Bigg[
        \exp\Bigg\{
            \sum_{i=1}^n
            \left(
                \frac{is}{\sqrt n\,\xi_{1h}}g_{1h}(X_i)
                -
                \frac{s^2}{2n^2\rho_n}\bar g(X_i)
            \right) \nonumber
            \\            
& \qquad \qquad \qquad \qquad \qquad \qquad \qquad  +
            \frac{2is}{\sqrt n(n-1)}
            \sum_{i=m+1}^{n-1}\sum_{j=i+1}^n
            \tilde g_{2h}(X_i,X_j)
        \Bigg\}
        \square_n^v(m)
    \Bigg]
    \Bigg|                                         \nonumber            \\
&
    +
    C
    \left(
        \frac{s}{\sqrt n(n-1)}
    \right)^{r+1}\left(1+ e^{(\rho_n Th)^{-1}-c_1 n}\right)
    \bbE[|\square_n(m)|^{r+1}]. \label{S4}
\end{align}
Here, the remainder term is obtained by applying the fact that \(\vert e^{iz}\vert=1\) for all \(z\in\bbC\) and the event \(W_n:= \{\frac{1}{n}\sum_{i=1}^n \bar g(X_i) \ge 0\}\) holds with probability at least \(1-\exp(-c_1n)\) for some constant \(c_1>0\) by \Cref{variance U} and otherwise \(\bar g(x)\) is universally bounded by \((Th)^{-1}\).
Again, since
\((\rho_nTh)^{-1}=o(n)\) under \Cref{assump: sparsity} and $\bbE[\vert\square_n(m)\vert^{r+1}]=O\left((mn)^{(r+1)/2}\right)$,
the remainder term is in the order of $
        \left(\frac{s\sqrt m}{n}\right)^{r+1}$. 
Choosing \(m = \frac{6rn\log n}{s^2}\) gives the remainder term of order \[
    O\left(
        n^{-\frac{r+1}{2}}(\log n)^{\frac{r+1}{2}}
    \right).
\]

We further bound the terms in the finite Taylor sum by expanding $\square^v_n(m)$ as follows:
\begin{equation}\label{S6}\begin{aligned}
    &\Bigg|
    \bbE\Bigg[
        \exp\Bigg\{
            \sum_{i=1}^n
            \left(
                \frac{is}{\sqrt n\,\xi_{1h}}g_{1h}(X_i)
                -
                \frac{s^2}{2n^2\rho_n}\bar g(X_i)
            \right)+
            \frac{2is}{\sqrt n(n-1)}
            \sum_{i=m+1}^{n-1}\sum_{j=i+1}^n
            \tilde g_{2h}(X_i,X_j)
        \Bigg\}
        \square_n^v(m)
    \Bigg]
    \Bigg| \\
    &=\Bigg\vert\bbE\Bigg[\exp\Bigg\{\sum_{i=1}^n\left(
                \frac{is}{\sqrt n\,\xi_{1h}}g_{1h}(X_i)
                -
                \frac{s^2}{2n^2\rho_n}\bar g(X_i)
            \right)+\frac{2is}{\sqrt n(n-1)}
            \sum_{i=m+1}^{n-1}\sum_{j=i+1}^n
            \tilde g_{2h}(X_i,X_j)\Bigg\}\\
    &\qquad \qquad \cdot \left(\sum_{i=1}^m\sum_{j=i+1}^n\tilde g_{2h}(X_i, X_j)\right)^v\Bigg]\Bigg\vert\\
    &=\sum_{\substack{1\leq j_1\leq m\\j_1\leq k_1\leq n}}\cdots\sum_{\substack{1\leq j_v\leq m\\j_v\leq k_v\leq n}}\Bigg\vert\bbE\Bigg[\exp\Bigg\{\sum_{i=1}^n\left(
                \frac{is}{\sqrt n\,\xi_{1h}}g_{1h}(X_i)
                -
                \frac{s^2}{2n^2\rho_n}\bar g(X_i)
            \right)\\
    &\qquad \qquad \qquad  \qquad \qquad \qquad \qquad +\frac{2is}{\sqrt n(n-1)}
            \sum_{i=m+1}^{n-1}\sum_{j=i+1}^n
            \tilde g_{2h}(X_i,X_j)\Bigg\}\prod_{\ell=1}^v\tilde g_{2h}(X_{j_\ell}, X_{k_\ell})\Bigg]\Bigg\vert.
\end{aligned}
\end{equation}
Expanding \(\square_n^v(m)\) gives at most \((mn)^v\) products of kernels.  For each
such product, at most \(2v\) indices among \(1,\ldots,m\) appear in the
kernels.  Let \(S:=[1:m]\backslash\cup_{\ell=1}^v\{j_\ell, k_\ell\}\) be the remaining indices in \(\{1,\ldots,m\}\). Then \(|S|\ge m-2v\), and the variables indexed by \(S\) factor out by independence.
Therefore, each sum term in (\ref{S6}) can be factored into two terms: 
\begin{align*}
& \left\vert\bbE\left[\exp\left\{\sum_{i\in S}\left(\frac{is}{\sqrt n\,\xi_{1h}}g_{1h}(X_i)
                -
                \frac{s^2}{2n^2\rho_n}\bar g(X_i)\right)\right\}\right]\right\vert\\
       \times & \Bigg\vert\bbE\Bigg[\exp\Bigg\{\sum_{i\notin S}\left(\frac{is}{\sqrt n\,\xi_{1h}}g_{1h}(X_i)
                -
                \frac{s^2}{2n^2\rho_n}\bar g(X_i)\right)+\frac{2is}{\sqrt n(n-1)}\sum_{i=m+1}^{n-1}\sum_{j=i+1}^n\tilde g_{2h}(X_i, X_j)\Bigg\}\\
        &\qquad \qquad  \cdot \prod_{\ell=1}^v\tilde g_{2h}(X_{j_\ell}, X_{k_\ell})\Bigg]\Bigg\vert,
\end{align*}
where  the first term is the contribution of the variables indexed by \(S\) and the second term is the contribution of the remaining variables. We will bound these two terms separately.

For the first term, we have 
\begin{equation}
    \label{S8}
    \left\vert\bbE\left[\exp\left\{\sum_{i\in S}\left(\frac{is}{\sqrt n\,\xi_{1h}}g_{1h}(X_i)
                -
                \frac{s^2}{2n^2\rho_n}\bar g(X_i)\right)\right\}\right]\right\vert \le |\gamma_n(s)|^{|S|} \leq \vert \gamma_n(s) \vert^{m-2v}.
\end{equation} This inequality holds because $|\gamma_n(s)|\leq \exp(-s^2/3n) \le 1$ as shown in \eqref{gamma_n bound} and $\vert S\vert\geq m-2v$.

For the second term, using the property that $\vert\exp(iz)\vert=1$ for all $z\in\bbC$ and $\vert \tilde g_{2h}(X_i, X_j)\vert=O(1)$, we have
\begin{align*}
   &  \Bigg\vert\bbE\Bigg[\exp\Bigg\{\sum_{i\notin S}\left(\frac{is}{\sqrt n\,\xi_{1h}}g_{1h}(X_i)
                -
                \frac{s^2}{2n^2\rho_n}\bar g(X_i)\right)+\frac{2is}{\sqrt n(n-1)}\sum_{i=m+1}^{n-1}\sum_{j=i+1}^n\tilde g_{2h}(X_i, X_j)\Bigg\}\\
        &\qquad \qquad  \cdot \prod_{\ell=1}^v\tilde g_{2h}(X_{j_\ell}, X_{k_\ell})\Bigg]\Bigg\vert
        \le   \bbE \left[
    \exp\left\{
        -\frac{s^2}{2n^2\rho_n}
        \sum_{i\notin S}\bar g(X_i)
    \right\} \right].
\end{align*}
Note that the number of indices in the summation of the above expectation is at least $n-m\asymp n(1-\frac{\log n}{s^2})\ge n(1-n^{-2\epsilon'}\log n)=O(n)$.
Therefore, we can apply the same concentration bound for $\bar g(X_i)$ as shown in \Cref{variance U} to control the above expectation.
In particular, the event $W_{S^c}:=\{\sum_{i\notin S} \bar g(X_i) \ge 0\}$ holds with probability at least $1-\exp(-c_1n)$ for some constant $c_1>0$. \\
(1) \underline{On the event $W_{S^c}$}, we have $\sum_{i\notin S}\bar g(X_i)\geq 0$,
and it follows that 
\begin{align*}
    \exp\left\{
        -\frac{s^2}{2n^2\rho_n}
        \sum_{i\notin S}\bar g(X_i) \right\} \leq 1.
\end{align*}
(2) \underline{On the event $W_{S^c}^c$}, we again apply the univeral bound $\|\bar g(x)\|_\infty \lesssim (Th)^{-1}$ by \Cref{variance U} to obtain that
\[
\left|
    \exp\left\{
        -\frac{s^2}{2n\rho_n} \frac{1}{n}
        \sum_{i\notin S}\bar g(X_i)
    \right\}
\right|
\lesssim 
\exp\left\{
    C\frac{s^2}{n\rho_nTh}
\right\}.
\]
Combining the above two cases, we have
\begin{align}
   \bbE \left[
    \exp\left\{
        -\frac{s^2}{2n^2\rho_n}
        \sum_{i\notin S}\bar g(X_i)
    \right\} \right]
    &\leq \exp\left\{C\frac{s^2}{n\rho_nTh}\right\}\exp(-c_1n)+1=O(1), \label{S9}
\end{align} 
where we use \((\rho_nTh)^{-1}=o(n)\) under \Cref{assump: sparsity}.

Plugging results of (\ref{S8}, \ref{S9}) into (\ref{S6}), we obtain that
\begin{align*}
&\Bigg|
    \bbE\Bigg[
        \exp\Bigg\{
            \sum_{i=1}^n
            \left(
                \frac{is}{\sqrt n\,\xi_{1h}}g_{1h}(X_i)
                -
                \frac{s^2}{2n^2\rho_n}\bar g(X_i)
            \right)+
            \frac{2is}{\sqrt n(n-1)}
            \sum_{i=m+1}^{n-1}\sum_{j=i+1}^n
            \tilde g_{2h}(X_i,X_j)
        \Bigg\}
        \square_n^v(m)
    \Bigg]
    \Bigg| \\
&\quad \leq(mn)^v|\gamma_n(s)|^{m-2v}.
\end{align*}
By further plugging the above inequality and the remainder bound $O\left(
        n^{-\frac{r+1}{2}}(\log n)^{\frac{r+1}{2}}
    \right)$  into (\ref{S4}), we have
\begin{align*}
    \left| \Psi_n(s) \right| &\lesssim \sum_{v=0}^r\left(\frac{s}{\sqrt n(n-1)}\right)^v(mn)^v|\gamma_n(s)|^{m-2v}+n^{-\frac{r+1}{2}}(\log n)^{\frac{r+1}{2}}.
\end{align*}
Since \(m=\lfloor6rn\log n/s^2\rfloor\) and
\(|\gamma_n(s)|\le \exp\{-s^2/(3n)\}\) in \eqref{gamma_n bound}, we have
\[
    |\gamma_n(s)|^{m-2v}
    \le
    Cn^{-2r},
    \qquad \text{ for }0\le v\le r.
\]
Thus, for $s\ge n^{\epsilon'}$, we have
\[
\begin{aligned}
\left| \Psi_n(s) \right| 
&\lesssim
    n^{-2r}
    \sum_{v=0}^r
    \left(
        \frac{\sqrt n\log n}{s}
    \right)^v
    +
    n^{-(r+1)/2}(\log n)^{(r+1)/2}                              \\
&\lesssim
    n^{-\frac{3}{2}r-r\epsilon'}\log^rn+ n^{-(r+1)/2}(\log n)^{(r+1)/2} \\
    & \lesssim n^{-(r+1)/2}(\log n)^{(r+1)/2}.
\end{aligned}
\]
Taking \(r=2\), we obtain, uniformly for
\(n^{\epsilon'}\le s\le C_1\sqrt n\), that
\(
\left| \Psi_n(s) \right| 
\lesssim 
n^{-3/2}(\log n)^{3/2}.
\)
It follows that
\begin{align*}
    \int_{n^{\epsilon'}}^{C_1\sqrt n} \left|\frac{\Psi_n(s)}{s}\right| ds
    &\lesssim
    n^{-3/2}(\log n)^{3/2} \int_{n^{\epsilon'}}^{C_1\sqrt n} \frac{1}{s} ds
    \lesssim
    n^{-3/2}(\log n)^{5/2}.
\end{align*}

Combining this leading-term bound with \eqref{eq:diff-bound} completes the proof.

\end{proof}

\subsubsection{Proof of \Cref{d}}
\label{Proof of d}
\begin{proof}
    Again, fix the target network time point \(t\), and write \(s\) for the Fourier variable. We consider the range
\(
    0\le |s|\le n^{\epsilon'},
\)
with  \(\epsilon' = \epsilon/2 \le \frac17 \).

Conditional on \(X\), we have
\(
    \tilde\Delta_n(t)\mid X
    \sim
    \mcalN\{0,(n\rho_n)^{-1}\sigma_X^2\},
\)
and \(\delta_T\sim \mcalN(0,c_\delta n^{-1}\log n)\) is independent of all
other variables. Therefore, it follows that
\[
\begin{aligned}
\bbE e^{is(J_n+\tilde\Delta_n+\delta_T)}   =
\bbE\left[
    e^{isJ_n}
    \exp\left\{
        -\frac{s^2}{2}
        \left(
            \frac{\sigma_X^2}{n\rho_n}
            +
            c_\delta n^{-1}\log n
        \right)
    \right\}
\right].
\end{aligned}
\]
Since $|s|\le n^{\epsilon'}$ and $\sigma_X^2 \lesssim (Th)^{-1}$, the exponent satisfies
\(s^2 ( \frac{\sigma_X^2}{n\rho_n}
            +
            c_\delta n^{-1}\log n) \lesssim n^{2\epsilon'}(n\rho_n Th)^{-1} + n^{2\epsilon'-1}\log n.\)
As \(\epsilon'=\epsilon/2\le 1/7\), the first term is $o(1)$ by \Cref{assump: sparsity}, while the second term is also $o(1)$.
Then, by applying the Taylor expansion $e^z=1+z+O(z^2)$ for $z=o(1)$, we obtain
\begin{align}\label{T1}
\bbE e^{is(J_n+\tilde\Delta_n+\delta_T)} 
&=
\bbE[e^{isJ_n}]
-
\frac{s^2}{2n\rho_n}
\bbE[e^{isJ_n}\sigma_X^2]                       
-
\frac12 c_\delta n^{-1}\log n\,s^2
\bbE[e^{isJ_n}]
+
R_G(s),
\end{align}
where
\[
\begin{aligned}
\int_0^{n^{\epsilon'}}
\left|
    \frac{R_G(s)}{s}
\right|ds
&\le
C\int_0^{n^{\epsilon'}}
s^3
\bbE\left[
    \left\{
        \frac{\sigma_X^2}{n\rho_n}
        +
        c_\delta n^{-1}\log n
    \right\}^2
\right]ds                                      \\
&=
O\left\{
    n^{4\epsilon'}(n\rho_nTh)^{-2}
    +
    n^{4\epsilon'-2}\log^2 n
\right\}                                      \\
&=
O\{(n\rho_nTh)^{-1}+n^{-1}\log n\}.
\end{aligned}
\]
The last step uses \(\epsilon'=\epsilon /2\le 1/7\) and
\Cref{assump: sparsity}.

Next, we analyze the leading term $\bbE[e^{isJ_n}]$ and the remaining  term $\bbE[e^{isJ_n}\sigma_X^2]$ in \eqref{T1} separately.
\paragraph{(1) The leading term \(\bbE[e^{isJ_n}]\) in \eqref{T1}.}
Recall that
\[
    J_n
    =
    \underbrace{L_n^{(1)}}_{\tilde O_{p,2}(\log^{1/2}n)}
    +
    \underbrace{ 
    \alpha_h +
    L_n^{(2)} -
    \frac12 L_n^{(3)} -
    \frac12 L_n^{(4)}}_{\tilde O_{p,2}(n^{-1/2}\log n)} .
\]
We keep the leading term $L_n^{(1)}$ and use Taylor expansion for the remaining terms:
\begin{align}
\bbE\Big[e^{is J_n}\Big]=&\bbE\Bigg[e^{isL_n^{(1)}}\Bigg\{1+is\Big(\alpha_h+L_n^{(2)}-\dfrac{L_n^{(3)}+L_n^{(4)}}{2}\Big)-\frac{1}{2}s^2\Big(\alpha_h+L_n^{(2)}-\dfrac{L_n^{(3)}+L_n^{(4)}}{2}\Big)^2 \nonumber\\
&\quad \quad \quad \quad \quad +O\left(s^3\left \vert\alpha_h+L_n^{(2)}-\dfrac{L_n^{(3)}+L_n^{(4)}}{2}\right \vert^3\right)\Bigg\}\Bigg]. \label{eq: Ttilde leading term}
\end{align}

    We first analyze the intercept term and the linear (w.r.t\ $s$) term in (\ref{eq: Ttilde leading term}). 
Note that $\bbE[e^{\mathrm{i}sL_n^{(1)}}]=\bbE\Big[e^{\mathrm{i}s\cdot \sum_{i=1}^n g_{1h}(X_i,t)/(\sqrt{n}\xi_{1h}(t))}\Big]=\prod_{i=1}^n\bbE\Big[e^{\mathrm{i}s\cdot g_{1h}(X_i)/(\sqrt{n}\xi_{1h}(t))}\Big]$. Define $\varphi_n(s):=\bbE\left[\exp\left(is\frac{g_{1h}(X_1, t)}{\sqrt n\xi_{1h}}\right)\right]$, then we have $\bbE[e^{\mathrm{i}sL_n^{(1)}}]= \left(\varphi_n(s)\right)^n$.
    By Section VI, Lemma 4 of \cite{petrov1972independent}, we have for $k=0,1,2$, 
    $$\varphi_n^{n-k}(s)=e^{-s^2/2}\left(1-\frac{is^3 }{6\sqrt n\xi_{1h}^3}\bbE[g_{1h}^3(X_1, t)]\right)+O(n^{-1}\log nP_k^{\ge 1}(s)e^{-s^2/12}),$$
    where $\text{P}^{\ge 1}_k(s)$ denotes a polynomial of $s$ with a fixed degree that involves only components with powers greater than $1$. 
    It follows directly that the intercept term satisfies
    \begin{align}
        \bbE\left[\exp\left(isL_n^{(1)}\right)\right]=\varphi_n^n(s)=e^{-s^2/2}\left(1-\frac{is^3 }{6\sqrt n\xi_{1h}^3}\bbE[g_{1h}^3(X_1, t)]\right)+O(n^{-1}\log nP_0(s)e^{-s^2/12}). \label{intercept L_n^{(1)}}
    \end{align}

    The linear term satisfies:
\begin{align}
    \bbE\Big[e^{\mathrm{i}sL_n^{(1)}}\cdot\mathrm{i}s\alpha_h\Big] & =  - \dfrac{\mathrm{i}s}{\sqrt{n}\,\xi_{1h}^3(t)}e^{-s^2/2}\Big(\frac12\bbE[g_{1h}^3(X_1,t)]
    +
    2\bbE[
        g_{1h}(X_1,t)g_{1h}(X_2,t)g_{2h}(X_1,X_2,t)
    ]\Big) \nonumber \\
&\quad  +O\Big(n^{-1}\cdot \text{P}_k^{\ge 1}(t) e^{-t^2 / 12}\Big), \label{linear term alpha}\\
\bbE[e^{isL_n^{(1)}}\cdot isL_n^{(2)}]
            &=\bbE\left[e^{isL_n^{(1)}}is\frac{2}{\sqrt n(n-1)}\sum_{i<j}\frac{g_{2h}(X_i, X_j, t)}{\xi_{1h}(t)}\right] \nonumber\\
        &=\frac{2is}{\sqrt n(n-1)}\binom{n}{2}\varphi_n^{n-2}(s)\bbE\left[\exp\left(is\frac{g_{1h}(X_1, t)+g_{1h}(X_2, t)}{\sqrt n\xi_{1h}(t)}\right)\frac{g_{2h}(X_1, X_2, t)}{\xi_{1h}(t)}\right] \nonumber\\
            &=\frac{is\sqrt n}{\xi_{1h}(t)}\varphi_n^{n-2}(s)\bbE\Bigg[g_{2h}(X_1, X_2, t)+is\frac{(g_{1h}(X_1, t)+g_{1h}(X_2, t))g_{2h}(X_1, X_2, t)}{\sqrt n\xi_{1h}(t)}\nonumber\\
            &\qquad \qquad \qquad \qquad -\frac{s^2(g_{1h}(X_1, t)+g_{1h}(X_2, t))^2g_{2h}(X_1, X_2, t)}{2n\xi_{1h}^2(t)}\nonumber\\
            &\qquad \qquad \qquad \qquad+O(n^{-3/2}s^3\rho_n^3)\Bigg]\nonumber \\
            &=-\frac{is^3}{\sqrt n\xi_{1h}^3(t)}e^{-s^2/2} \bbE[g_{1h}(X_1, t)g_{1h}(X_2, t)g_{2h}(X_1, X_2, t)]+O(n^{-1}e^{-s^2/12}P^{\ge 1}(s)),  \label{linear term L_n^{(2)}}
\end{align}
 where we use the facts that $\{X_i\}_{i=1}^n$ are independent and identically distributed and that $\bbE[g_{2h}(X_1, X_2, t)]=0$ and $\bbE[g_{1h}^k(X_1, t)g_{2h}(X_1, X_2, t)]=0.$

    Similarly, we can derive
\begin{align}
     & \bbE[e^{isL_n^{(1)}}\cdot isL_n^{(3)}] \nonumber \\
       & \quad =   \bbE\left[e^{isL_n^{(1)}}\frac{is}{n^{3/2}\xi_{1h}^3(t)}
\sum_{1 \le i\neq j \le n} 
\Big[
    g_{1h}(X_i,t)\{g_{1h}^2(X_j,t)-\xi_{1h}^2(t)\}
\Big]\right] \nonumber \\
            & \quad = \frac{is(n-1)}{\sqrt n \xi_{1h}^3(t)}\varphi_n^{n-2}(s) \bbE\left[\exp\left(is\frac{g_{1h}(X_1, t)+g_{1h}(X_2, t)}{\sqrt n\xi_{1h}(t)}\right)
\Big[
    g_{1h}(X_1,t)\{g_{1h}^2(X_2,t)-\xi_{1h}^2(t)\}
\Big]\right] \nonumber \\
&\quad =\frac{is(n-1)}{\sqrt n \xi_{1h}^3(t)}\varphi_n^{n-2}(s))\bbE\Bigg[g_{1h}(X_1,t)\{g_{1h}^2(X_2,t)-\xi_{1h}^2(t)\} \nonumber\\
   &\qquad \qquad \qquad \qquad \qquad \qquad  +is\frac{(g_{1h}(X_1, t)+g_{1h}(X_2, t)) g_{1h}(X_1,t)\{g_{1h}^2(X_2,t)-\xi_{1h}^2(t)\}}{\sqrt n\xi_{1h}(t)}\nonumber\\
            &\qquad \qquad \qquad \qquad \qquad \qquad -\frac{s^2(g_{1h}(X_1, t)+g_{1h}(X_2, t))^2 g_{1h}(X_1,t)\{g_{1h}^2(X_2,t)-\xi_{1h}^2(t)\}}{2n\xi_{1h}^2(t)}\nonumber\\
            &\qquad \qquad \qquad \qquad \qquad \qquad  +O(n^{-3/2}s^3\rho_n^9)\Bigg]\nonumber \\
            & \quad =-\frac{is^3(n-1)}{\sqrt n \xi_{1h}^3(t)}\varphi_n^{n-2}(s)\bbE\left[\frac{g_{1h}^2(X_1, t)g_{1h}(X_2, t)(g_{1h}^2(X_2, t)-\xi_{1h}^2(t))}{n\xi_{1h}^2}+O(n^{-3/2}s^3\rho_n^9) \right] \nonumber\\
            &\quad =-\frac{is^3\varphi_n^{n-2}(s)}{\sqrt n\xi_{1h}^3(t)}\bbE[g_{1h}^3(X_1, t)]+O(n^{-1}e^{-s^2/12}P^{\ge 1}(s)),\label{linear term L_n^{(3)}}
\end{align}
    where the last two equalities we use the property that $\bbE[g_{1h}^2(X_i, t)]=\xi_{1h}^2$, $\bbE[g_{1h}(X_i, t)]=0$, and $\xi_{1h}(t)\asymp \rho_n^3 $.
For the forth linear term, we can similarly derive that
\begin{align}
     & \bbE[e^{isL_n^{(1)}}\cdot isL_n^{(4)}] \nonumber \\
       & \quad =   \bbE\left[e^{isL_n^{(1)}}\frac{4is}{n^{3/2}\xi_{1h}^3(t)} 
\sum_{1 \le i\neq j \le n} 
\Big[
    g_{1h}(X_i,t)\zeta_h(X_j,t)
\Big]\right] \nonumber \\
            & \quad = \frac{4is(n-1)}{\sqrt n \xi_{1h}^3(t)}\varphi_n^{n-2}(s) \bbE\left[\exp\left(is\frac{g_{1h}(X_1, t)+g_{1h}(X_2, t)}{\sqrt n\xi_{1h}(t)}\right)
\Big[
    g_{1h}(X_1,t)\zeta_h(X_2,t)
\Big]\right] \nonumber \\
&\quad =\frac{4is(n-1)}{\sqrt n \xi_{1h}^3(t)}\varphi_n^{n-2}(s))\bbE\Bigg[g_{1h}(X_1,t)\zeta_h(X_2,t) \nonumber\\
   &\qquad \qquad \qquad \qquad \qquad \qquad  +is\frac{(g_{1h}(X_1, t)+g_{1h}(X_2, t)) g_{1h}(X_1,t)\zeta_h(X_2,t)}{\sqrt n\xi_{1h}(t)}\nonumber\\
            &\qquad \qquad \qquad \qquad \qquad \qquad -\frac{s^2(g_{1h}(X_1, t)+g_{1h}(X_2, t))^2 g_{1h}(X_1,t)\zeta_h(X_2,t)}{2n\xi_{1h}^2(t)}\nonumber\\
            &\qquad \qquad \qquad \qquad \qquad \qquad  +O(n^{-3/2}s^3\rho_n^9)\Bigg]\nonumber \\
            & \quad =-\frac{4is^3(n-1)}{\sqrt n \xi_{1h}^3(t)}\varphi_n^{n-2}(s)\bbE\left[\frac{g_{1h}^2(X_1, t)g_{1h}(X_2, t)\zeta_h(X_2,t)}{n\xi_{1h}^2}+O(n^{-3/2}s^3\rho_n^9) \right] \nonumber\\
            &\quad =-\frac{4is^3\varphi_n^{n-2}(s)}{\sqrt n\xi_{1h}^3(t)}\bbE[g_{1h}(X_1, t)g_{1h}(X_2, t)g_{2h}(X_1, X_2, t)]+O(n^{-1}e^{-s^2/12}P^{\ge 1}(s)),\label{linear term L_n^{(4)}}
\end{align}
    where the last two equalities we use the property that $\bbE[g_{1h}(X_i, t)]=0$, $\bbE[\zeta_{h}(X_i, t)]=0$, $\bbE[g_{1h}^2(X_i, t)]=\xi_{1h}^2$, $\bbE[g_{1h}(X_i, t)\zeta_h(X_i,t)]=\bbE[g_{1h}(X_i, t)g_{1h}(X_j, t)g_{2h}(X_i, X_j, t)]$, and $\xi_{1h}(t)\asymp \rho_n^3 $.

We combine the results in \Cref{intercept L_n^{(1)},linear term alpha,linear term L_n^{(2)},linear term L_n^{(3)},linear term L_n^{(4)}} and obtain that, for the intercept and linear terms in \eqref{eq: Ttilde leading term}, 
\begin{equation}
        \label{T9}\begin{aligned}
            &\bbE\left[e^{isL_n^{(1)}}\Bigg\{1+is\Big(\alpha_h+L_n^{(2)}-\dfrac{L_n^{(3)}+L_n^{(4)}}{2}\Big) \Bigg\} \right]\\
            &\quad =e^{-s^2/2}\bigg\{1-\frac{is}{\sqrt n\xi_{1h}^3} \left(\frac{\bbE[g_{1h}^3(X_1, t)]}{2}+2\bbE[g_{1h}(X_1, t)g_{1h}(X_2, t)g_{2h}(X_1, X_2, t)]\right)\\
            &\qquad \qquad \qquad+\frac{is^3}{\sqrt n\xi_{1h}^3}\left(\frac{\bbE[g_{1h}^3(X_1,t)]}{3}+\bbE[g_{1h}(X_1, t)g_{1h}(X_2, t)g_{2h}(X_1, X_2, t)]\right)\bigg\}\\
            &\qquad+O(n^{-1}\log n P^{\ge 1}(s)e^{-s^2/12}).
        \end{aligned}
    \end{equation}
Following the same procedure, it can be shown that the quadratic term $\bbE[e^{isL_n^{(1)}}\cdot s^2(\alpha_h+L_n^{(2)}-\dfrac{L_n^{(3)}+L_n^{(4)}}{2})^2]$ in \eqref{eq: Ttilde leading term} can be bounded by the order $O(n^{-1}e^{-s^2/12}P^{\ge 1}(s))$. 

Finally, the remainder term in (\ref{eq: Ttilde leading term}) can be bounded by 
\begin{align}
	\left|\bbE\left[e^{\mathrm{i}sL_n^{(1)}} O\left(\left\vert\alpha_n+L_n^{(2)}-\dfrac{L_n^{(3)}+L_n^{(4)}}{2}\right\vert^3s^3\right)\right] \right| & \lesssim s^3\bbE\left[\left\vert\alpha_n+L_n^{(2)}-\dfrac{L_n^{(3)}+L_n^{(4)}}{2}\right\vert^3\right] \nonumber \\
    & = O(s^3 n^{-3/2}\log^3n). \label{remainder term}
\end{align}
To see the last bound, note that  $\alpha_n\asymp n^{-1/2}$, and for any fixed $k>0$, by applying the multivariate version of Bernstein's inequality (Theorem 1 in \cite{10.1214/EJP.v12-430}), $L_n^{(2)}$, $L_n^{(3)}$, and $L_n^{(4)}$ are bounded by $C(k)n^{-1/2} \log n$ with probability $1-O(n^{-k})$.
On this event, 
\(\left|\alpha_n+L_n^{(2)}-\frac{L_n^{(3)}+L_n^{(4)}}{2}\right|^3
=O(n^{-3/2}\log^3 n).\)
On the complement event, the same quantity is at most $O(n^{3/2})$, since the terms are universally bounded by $\sqrt n$. Choosing $k=3$ makes the contribution from this complement event $O(n^{3/2}n^{-3})=O(n^{-3/2})$, which is absorbed into $O(n^{-3/2}\log^3 n)$. Thus, the expectation is $O(n^{-3/2}\log^3 n)$.

By plugging the bounds in \eqref{T9} and \eqref{remainder term} into \eqref{eq: Ttilde leading term}, we have
\begin{align}
\bbE\Big[e^{is J_n}\Big] = Ch.f.(G_{nh};s) + O(n^{-1}\log n P^{\ge 1}(s)e^{-s^2/12}) + O(s^3 n^{-3/2}\log^3 n),\label{Jn expansion}
\end{align}
where the characteristic function approximation $Ch.f.(G_{nh};s)$ is defined as 
    $$\begin{aligned}
        Ch.f.(G_{nh};s)&:=e^{-s^2/2}\bigg\{1-\frac{is}{\sqrt n\xi_{1h}^3} \left(\frac{\bbE[g_{1h}^3(X_1, t)]}{2}+2\bbE[g_{1h}(X_1, t)g_{1h}(X_2, t)g_{2h}(X_1, X_2, t)]\right)\\
            & \qquad \qqquad+\frac{is^3}{\sqrt n\xi_{1h}^3}\left(\frac{\bbE[g_{1h}^3(X_1,t)]}{3}+\bbE[g_{1h}(X_1, t)g_{1h}(X_2, t)g_{2h}(X_1, X_2, t)]\right)\bigg\}.
    \end{aligned}$$
Correspondingly, we invoke the following two integrals $$\int_{-\infty}^u\frac{1}{2\pi}\int_\bbR e^{-\frac{s^2}{2}}e^{-isx}is^3dsdx=\int_{-\infty}^u\frac{x(3-x^2)}{\sqrt{2\pi}}e^{-\frac{x^2}{2}}dx=(u^2-1)\varphi(u);$$ and $$\int_{-\infty}^u\frac{1}{2\pi}\int_\bbR e^{-\frac{s^2}{2}}e^{-isx}isdsdx=\int_{-\infty}^u\frac{x}{\sqrt{2\pi}}e^{-\frac{x^2}{2}}dx=-\varphi(u),$$ to derive the distribution approximation via the Fourier inversion identity $$\begin{aligned}
        G_{nh}(x)=&\Phi(x)+\frac{\varphi(x)}{\sqrt n\xi_{1h}^3}\cdot\bigg\{(-1)\cdot(-1)\left(\frac{\bbE[g_{1h}^3(X_1, t)]}{2}+2\bbE[g_{1h}(X_1, t)g_{1h}(X_2, t)g_{2h}(X_1, X_2, t)]\right)\\
        &+(x^2-1)\left(\frac{\bbE[g_{1h}^3(X_1, t)]}{3}+\bbE[g_{1h}(X_1, t)g_{1h}(X_2, t)g_{2h}(X_1, X_2, t)]\right)\bigg\}\\
        =&\Phi(x)+\frac{\varphi(x)}{\sqrt n\xi_{1h}^3}\cdot\bigg\{\frac{2x^2+1}{6}\bbE[g_{1h}^3(X_1, t)]+(x^2+1)\bbE[g_{1h}(X_1, t)g_{1h}(X_2, t)g_{2h}(X_1, X_2, t)]\bigg\}.
    \end{aligned}$$
It follows that, by \eqref{Jn expansion}, 
\begin{align}
	& \int_0^{n^{\epsilon'}}\left|\frac{\bbE\big[e^{\mathrm{i}sJ_n}\big]-Ch.f.(\mathcal{G}_{nh};s)}{s}\right| ds \nonumber\\
	= & O\left( n^{-1}\int_0^{n^{\epsilon'}}\mid \text{P}^{\ge 0}(s) e^{-s^2 / 12}\mid ds + n^{-3/2}\log^3n \int_0^{n^{\epsilon'}} s^2 ds\right) \nonumber \\
	= & O(n^{-1} +n^{3\epsilon'-3/2}\log^3n)= O(n^{-1}), \label{eq: Jn bound}
\end{align}
where the last equality holds since $\epsilon' < 1/7$.

\paragraph{(2) The remaining term $\bbE[e^{isJ_{n}}\sigma_X^2]$ in \eqref{T1}.}
By \Cref{variance U}, we have the following expansion:
\begin{align*}
     \bbE[e^{isJ_{n}}\sigma_X^2]   &  = \bbE\left[e^{isJ_{n}}\left(\bbE[\sigma_X^2]+\frac{1}{n}\sum_{i=1}^ng_{\sigma;1}(X_i)+R_{\sigma, n}\right)\right] \\
     & \leq \bbE[e^{isJ_{n}}]\bbE[\sigma_X^2]+\bbE\left[e^{isJ_{n}}g_{\sigma;1}(X_1)\right]+\bbE\left[|R_{\sigma, n}|\right], \\
\end{align*}
where $R_{\sigma, n}=\tilde O_{p, 1}(n^{-1}(Th)^{-1}\log n)$ is the remainder term in the expansion of $\sigma_X^2$ with universal bound $|R_{\sigma, n}|\le C (Th)^{-1}$. 
It follows that $\bbE\left[|R_{\sigma, n}|\right]=O(n^{-1}(Th)^{-1}\log n)$.

By \eqref{Jn expansion}, we have $\bbE[e^{isJ_{n}}]=Ch.f.(G_{nh};s)+O(n^{-1}\log n P^{\ge 1}(s)e^{-s^2/12}+s^3 n^{-3/2}\log^3 n)$, where the leading term is of the form $e^{-s^2/2}P^{\ge 0}(s)$.
Together with $\bbE[\sigma_X^2]\asymp (Th)^{-1}$, we have \(|\bbE[e^{isJ_{n}}]|\bbE[\sigma_X^2] = O(e^{-s^2/2}P^{\ge 0}(s)(Th)^{-1})\).
Moreover, by applying the same expansion technique in \eqref{eq: Ttilde leading term} together with $\|g_{\sigma;1}\|_{\infty} \lesssim (Th)^{-1}$, we can obtain the bound
\(
        \left\vert\bbE[e^{isJ_{n}}g_{\sigma;1}(X_1)]\right\vert=O(e^{-s^2/12}P^{\ge 0}(s)(Th)^{-1}).\)
Therefore, we have
\begin{align}
        \int_0^{n^{\epsilon'}}\left\vert\frac{s^2}{2n\rho_n}
\bbE[e^{isJ_n}\sigma_X^2] \frac{1}{s}\right\vert ds & \lesssim \int_0^{n^{\epsilon'}}\left[e^{-s^2/12}P^{\ge 1}(s)(n\rho_nTh)^{-1} + s(n^{-2}(\rho_nTh)^{-1}\log n)\right] ds \nonumber \\
& =O((n\rho_nTh)^{-1} + n^{2\epsilon'-2}(\rho_nTh)^{-1}\log n) \nonumber \\
& = O((n\rho_nTh)^{-1}) . \label{eq: Jn sigma_X^2 bound}
    \end{align}
Here, we use the fact that $\int_0^\infty e^{-s^2/12}P^{\ge 0}(s) ds$ is bounded by a constant, and $\epsilon'<\frac{1}{2}$.

    Finally, we also have the following bound  
    $$\int_0^{n^{\epsilon'}}\bbE\left\vert e^{isJ_{n}}\frac{s^2}{2}c_\delta n^{-1}\log n\right\vert\cdot s^{-1} ds=n^{-1}\log n\cdot O\left(\int_0^{n^{\epsilon'}}\vert s\bbE[e^{isJ_{n}}]\vert ds\right)=O(n^{-1}\log n),$$ since the leading term  in $\bbE[\exp(isJ_{n})]$ is of the form $e^{-t^2/2}P^{\ge 0}(s)$.

    By combining the above result with (\ref{eq: Jn bound}) and (\ref{eq: Jn sigma_X^2 bound}), we obtain:  $$\int_0^{n^{\epsilon'}}\left\vert\frac{\bbE[e^{is(J_{n}+\tilde\Delta_n+\delta_T)}]-Ch.f.(G_{nh};s)}{s}\right\vert ds=O(n^{-1}\log n+(n\rho_n)^{-1}(Th)^{-1}).$$
\end{proof}

\section{Technical Lemmas}
\begin{lemma}(Koksma-Hlawka Inequality, Theorem 2.11 in \cite{niederreiter1992random})\label{lem:KoksmaHlawka}
    For any $x_1, \cdots, x_N\in[a, b)$, let $f$ have bounded variation $TV(f)$ on $[a, b]$. Then we have $$\left\vert\frac{1}{N}\sum_{i=1}^Nf(x_i)-\int_a^b f(u)du\right\vert\leq TV(f)\cdot D_N^*(x_1, \cdots, x_n),$$ where $D_N^*(x_1, \cdots, x_n)=\sup_{t\in[a, b]\cap[0, \infty]}\left\vert\frac{1}{N}\sum_{i=1}^N\mathbf{1}_{[0, t)}(x_i)-t\right\vert$.
\end{lemma}

\begin{lemma}
\label{kappa j} Let $\kappa_{0j}=\int K^j(u)du$ and define for $j \geq 1$
 $$\epsilon_{0j}:=\kappa_{0j}-\frac{1}{h}\sum_{\ell=1}^TK^j\left(\frac{t_\ell-t}{h}\right)\Delta t_\ell, \quad \text{where } \Delta t_\ell:= t_\ell - t_{\ell-1}.$$ 
 Under Assumptions~\ref{assump: time obs} and~(K1)-(K4), for every fixed $j \geq 1$, there exists a constant $c'$ such that uniformly for $t \in [\delta, \mcalL-\delta]$,
 \[\epsilon_{0j} = O\left(\frac{1}{Th}+\exp(-c'h^{-\alpha})\right).\]
In particular, for $j=1$, 
\[\frac{1}{T}\sum_{\ell=1}^TK_h(t_\ell-t) = \kappa_{0j} - O\left(\frac{1}{Th}+\exp(-c'h^{-\alpha})\right) = O(1), \quad \text{where } K_h(u)=\frac{1}{h}K(\frac{u}{h}).\]
\end{lemma}

\begin{proof}[\textbf{Proof of \Cref{kappa j}}]

We first apply the triangle inequality:
    $$\vert\epsilon_{0j}\vert\leq\left\vert\frac{1}{h}\int_0^{\mcalL}K^j\left(\frac{u-t}{h}\right)du-\frac{1}{h}\sum_\ell K^j\left(\frac{t_\ell-t}{h}\right)\Delta t_\ell\right\vert+\left\vert\int_{-\frac{t}{h}}^{\frac{\mcalL-t}{h}}K^j(u)du-\int_{-\infty}^{\infty}K^j(u)du\right\vert.$$
In the first term, the integrand $g(u):=K^j\left(\frac{u-t}{h}\right)$ has bounded variation on $[0, \mcalL]$ as $TV(g; [0, \mcalL]) \leq  TV(K^j, [-t/h, (\mcalL-t)/h])\leq  TV(K^j, \mathbb{R}) \leq j M^{j-1} TV(K, \mathbb{R})<\infty$, where the second last inequality holds as $K(\cdot)$ is bounded by $M$ by (K2) and the last inequality holds by (K4) with $j=0$. Thus, by
\Cref{lem:KoksmaHlawka} and Assumption~\ref{assump: time obs},
$$\left\vert\frac{1}{h}\int_0^{\mcalL}K^j\left(\frac{u-t}{h}\right)du-\frac{1}{h}\sum_\ell K^j\left(\frac{t_\ell-t}{h}\right)\Delta t_\ell\right\vert\leq \frac{1}{h} \cdot  TV(g;[0, \mcalL]) \cdot\frac{C}{T} =O(\frac{1}{hT}).$$

For the second term, since $t\in [\delta, \mcalL-\delta]$, both lower and upper limits satisfy $t/h \geq \delta/ h$ and $(\mcalL-t)/h \geq \delta /h$. By (K3) and $|K^j(u)| \leq M^{j-1}|K(u)|$, we have 
\begin{align*}
\left|\int_{-\frac{t}{h}}^{\frac{\mcalL -t}{h}} K^j(u) d u - \int_{-\infty}^{\infty}K^j(u)du\right| & \leq \int_{\frac{\mcalL-t}{h}}^{\infty}\left|K^j(u)\right| d u+\int_{-\infty}^{-\frac{t}{h}}\left|K^j(u)\right| d u \\
& \leq M^{j-1}\left(\int_{\frac{\mcalL-t}{h}}^{\infty}|K(u)| d u+\int_{\frac{t}{h}}^{\infty}|K(u)| d u\right) \\
& \leq 2M^{j-1} C^{\prime} \exp \left\{-c(\delta / h)^\alpha\right\} \\
& =O\left(\exp \left\{-c^{\prime} h^{-\alpha}\right\}\right),
\end{align*}
where the second last inequality holds by the tail bound in (K4) and the fact that $K(\cdot)$ is even in (K1), and $c'$ is a constant. 

Combining the bounds for these two terms completes the proof. The special case for $j=1$ holds as $\kappa_{01}=1$ by (K1)
\end{proof}

\begin{lemma}
\label{kappa i1} Let $\kappa_{i 1}:=  \int_{-\infty}^{\infty} u^i K(u) d u$ and define for $1 \le i \le \nu$
$$\epsilon_{i1}:=\kappa_{i1}-\frac{1}{h}\sum_{\ell=1}^T\left(\frac{t_\ell-t}{h}\right)^iK\left(\frac{t_\ell-t}{h}\right) \Delta t_\ell, \quad \text{where } \Delta t_\ell:= t_\ell - t_{\ell-1}.$$
Under Assumptions~\ref{assump: time obs} and~(K1)-(K4), for every fixed $1 \le i \le \nu$, there exists a constant $c'$ such that uniformly for $t \in [\delta, \mcalL-\delta]$,
$$\epsilon_{i1}=O\left(\frac{1}{Th}+h^{-(i-1)}\exp(-{c'}h^{-\alpha})\right).$$ 
Therefore, 
\begin{itemize}
	\item for $i=1, \cdots, \nu-1$, due to vanishing moments $\kappa_{i1}=0$, it follows that 
	\[\frac{1}{h}\sum_{\ell=1}^T\left(\frac{t_\ell-t}{h}\right)^iK\left(\frac{t_\ell-t}{h}\right) \Delta t_\ell = O\left(\frac{1}{Th}+h^{-(i-1)}\exp(-{c'}h^{-\alpha})\right);\]
	\item for $i=\nu$, \[\frac{1}{h}\sum_{\ell=1}^T\left(\frac{t_\ell-t}{h}\right)^\nu K\left(\frac{t_\ell-t}{h}\right) \Delta t_\ell = \kappa_{\nu1}+O\left(\frac{1}{Th}+h^{-(\nu-1)}\exp(-{c'}h^{-\alpha})\right).\]
\end{itemize}

\end{lemma}

\begin{proof}[\textbf{Proof of \Cref{kappa i1}}] By the triangle inequality,
    $$\begin{aligned}
\vert\epsilon_{i1}\vert&\leq\left\vert\frac{1}{h}\int_0^{\mcalL}\left(\frac{u-t}{h}\right)^iK\left(\frac{u-t}{h}\right)du-\frac{1}{h}\sum_\ell\left(\frac{t_\ell-t}{h}\right)^i K\left(\frac{t_\ell-t}{h}\right)\Delta t_\ell\right\vert\\
        &\quad+\left\vert\frac{1}{h}\int_0^{\mcalL}\left(\frac{u-t}{h}\right)^iK\left(\frac{u-t}{h}\right)du-\int_{-\infty}^{\infty}u^iK(u)du\right\vert. 
    \end{aligned}$$
    For the first term, we observe that $$\frac{1}{h}\int_0^{\mcalL}\left(\frac{u-t}{h}\right)^iK\left(\frac{u-t}{h}\right)du=\int_{-\frac{t}{h}}^{\frac{\mcalL-t}{h}}u^iK(u)du.$$ The integrand $f(u)=u^i K(u)$ has bounded total variation {by (K4) and $\Delta t_\ell /h \le C/(Th)$ by Assumption~\ref{assump: time obs}}. Therefore, by \Cref{lem:KoksmaHlawka}, we obtain the error bound for the first term is $O\left(\frac{1}{Th}\right)$. 
    For the second term, {using $\min\{t,\mcalL-t\} \geq \delta$ and (K3), we have 
    \begin{align*}
\left\vert\int_{-\frac{t}{h}}^{\frac{\mcalL-t}{h}}u^iK(u)du-\int_{-\infty}^{\infty}u^iK(u)du\right\vert & \leq \int_{(\mathcal{L}-t) / h}^{\infty}|u|^i|K(u)| d u+\int_{-\infty}^{-t / h}|u|^i|K(u)| d u \\
& \leq 2 \int_{\delta / h}^{\infty} u^i|K(u)| d u \\
& \leq 2 C\left(\frac{\delta}{h}\right)^{i-1} \exp \left\{-c\left(\frac{\delta}{h}\right)^\alpha\right\} \\
& =O\left(h^{-(i-1)} e^{-c^{\prime}  h^{-\alpha}}\right).
\end{align*}}
Combining these two bounds completes the proof. {The special cases hold due to vanishing moments by (K1).}
\end{proof}

We next present several technical lemmas that are useful for bounding the higher-order residual terms in the difference between the empirical network moment and its noiseless U-statistics, as well as for controlling the CDF approximation error. To this end, we first introduce some basic definitions.

Let $H=(\mcalV(H), \mcalH(H))$ be a hypergraph consisting a vertex set $\mcalV(H)=[n]$ and a set of hyperedges $\mcalH(H)$. Each hyperedge $h$ consists of a set $\mcalV(h)\subset\mcalV(H)$ of $\vert\mcalV(h)\vert$ vertices. We also associate a real-valued weight $w_h$ with each hyperedge $h\in\mcalH(H)$. Given such a weighted hypergraph, we define the corresponding multilinear polynomial {of power $q$:} $$f(x)=\sum_{h\in\mcalH(H){: \vert\mcalV(h)\vert\leq q}}w_h\prod_{v\in\mcalV(h)}x_v.$$ For the statement of lemmas, we  define for each $r\in\{0, 1, \cdots, q\}$, $$\mu_r:=\max_{S {\subseteq}[n]:\vert S\vert=r}\left(\sum_{h\in\mcalH: S \subset \mcalV(h)}\vert w_h\vert\prod_{v\in\mcalV(h)\backslash S}\bbE[\vert Y_v\vert]\right),$$ where $\{Y_v\}_{v\in[n]}$ is a collection of real-valued random variables.

\begin{definition}{(Central Moment Bounded Random Variable)}
\label{central_def}
    A random variable $Z$ is called \textbf{central moment bounded} with real parameter $L>0$, if for any integer $i\geq1$, $$\bbE[\vert Z-\bbE[Z]\vert^i]\leq i\cdot L\cdot\bbE[\vert Z-\bbE[Z]\vert^{i-1}].$$
\end{definition}

\begin{lemma}{(Theorem 1.3 in \cite{schudy2011bernstein})}
\label{hypergraph lemma}
   Let $Y=(Y_1, \cdots, Y_n)$ be $n$ independent central moment bounded random variables  with the same parameter $L$ and $f(y)$ be a multilinear polynomial  of power $q$. Define $f(Y)=f(Y_1, \cdots, Y_n)$, then $$\bbP(\vert f(Y)-\bbE[f(Y)]\vert\geq\lambda)\leq e^2\max\left\{\exp\left(-\frac{\lambda^2}{var(f(Y))C^q}\right), \max_{r\in[q]}\left\{\exp\left(-\left(\frac{\lambda}{\mu_r L^r C^q}\right)^{1/r}\right)\right\}\right\},$$ where $C$ is a universal constant.
\end{lemma}

\begin{lemma}
    \label{random perturbation} (Lemma 8.2 in \cite{zhang2022edgeworth})
    Suppose we have random variables $X, Y, Z$ satisfying $$X=Y+Z$$ such that the CDF of $Y$ is smooth, and there exists a universal constant $0<M<\infty$ and a positive sequence $\zeta_n$ such that $\mcalF_Y(u+a)-\mcalF_Y(u)\leq Ma+O(\zeta_n)$ for any $u\in\bbR$ and $a>0$. Also assume that $\bbP(\vert Z\vert\geq\tilde \zeta_n)\leq n^{-1}$, that is $Z=\tilde O_{p, 1}(\tilde \zeta_n)$ for a positive sequence $\tilde \zeta_n$. We have $$\Vert \mcalF_X(u)-\mcalF_Y(u)\Vert_\infty=O(\zeta_n+\tilde \zeta_n+n^{-1}).$$
\end{lemma}

\begin{lemma}(Esseen's smoothing Lemma~\citep[Section XVI.3]{feller1971introduction})
    \label{smoothing lemma}For any distribution function $F$ and a general function $G$ that has universally bounded derivative and satisfies $G(-\infty)=0$ and $G(\infty)=1$, {for an arbitrary constant $\gamma>0$,} there exist some universal constants $C_1, C_2>0$ such that $$\left\Vert F(\cdot)-G(\cdot)\right\Vert_\infty\leq C_1\int_{-\gamma}^\gamma\left\vert\frac{Ch.f.(F;t)-Ch.f.(G;t)}{t}\right\vert dt+\frac{C_2\sum_u\vert G'(u)\vert}{\gamma},$$  where $Ch.f.(G;t)$ is the characteristic function of $G$, defined as $Ch.f.(G;t):=\int_{-\infty}^\infty e^{itx}dG(x).$
\end{lemma}

\section{More Simulation Results}

\subsection{Analysis of CDF Approximation}
\label{CDF approx}
We compare the distance between the CDF of $\hat T_{nh}+\delta_T$, where $\delta_T\sim\mcalN(0, c_\delta n^{-1}\log n)$ with $c_\delta=0.01$ is the Gaussian smoothing random variable introduced in \Cref{subsec: studentization}, and its Edgeworth expansion approximation in \Cref{main thm}. 
To obtain the CDF of $\hat T_{nh}+\delta_T$, we run simulations across $N=2000$ independent seeds and compute the corresponding empirical CDF.
This comparison is conducted under both Setting 1 and Setting 2 with the default parameter setting. 
 \begin{figure}[t]
    
    \begin{subfigure}{0.48\textwidth}
        \centering
        \includegraphics[width=1\linewidth]{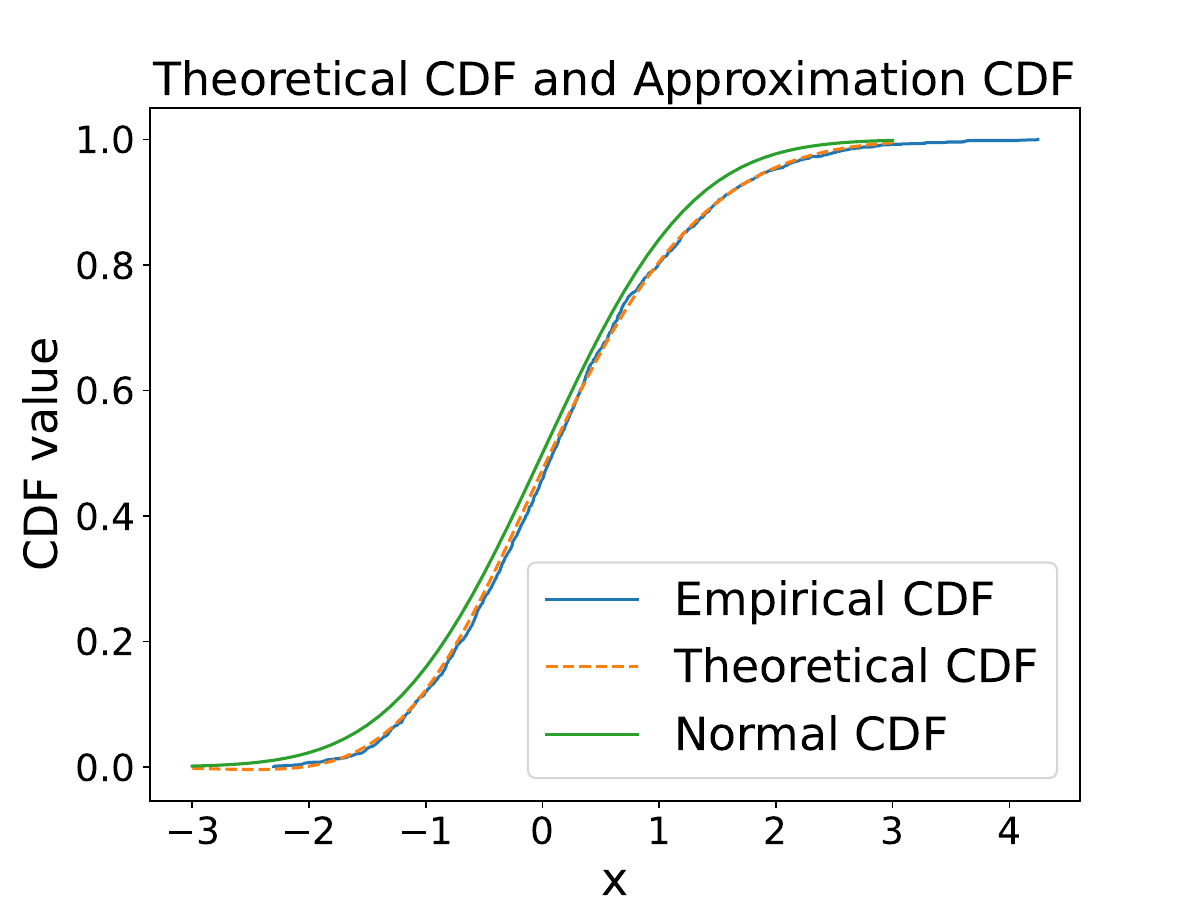}
        \label{fig:TwoPolyN200T200Rho08S08InferidxmidFrac10CombineCDF}
    \end{subfigure}
    \hspace{0.05\textwidth}
    \begin{subfigure}{0.48\textwidth}
        \centering
        \includegraphics[width=1\linewidth]{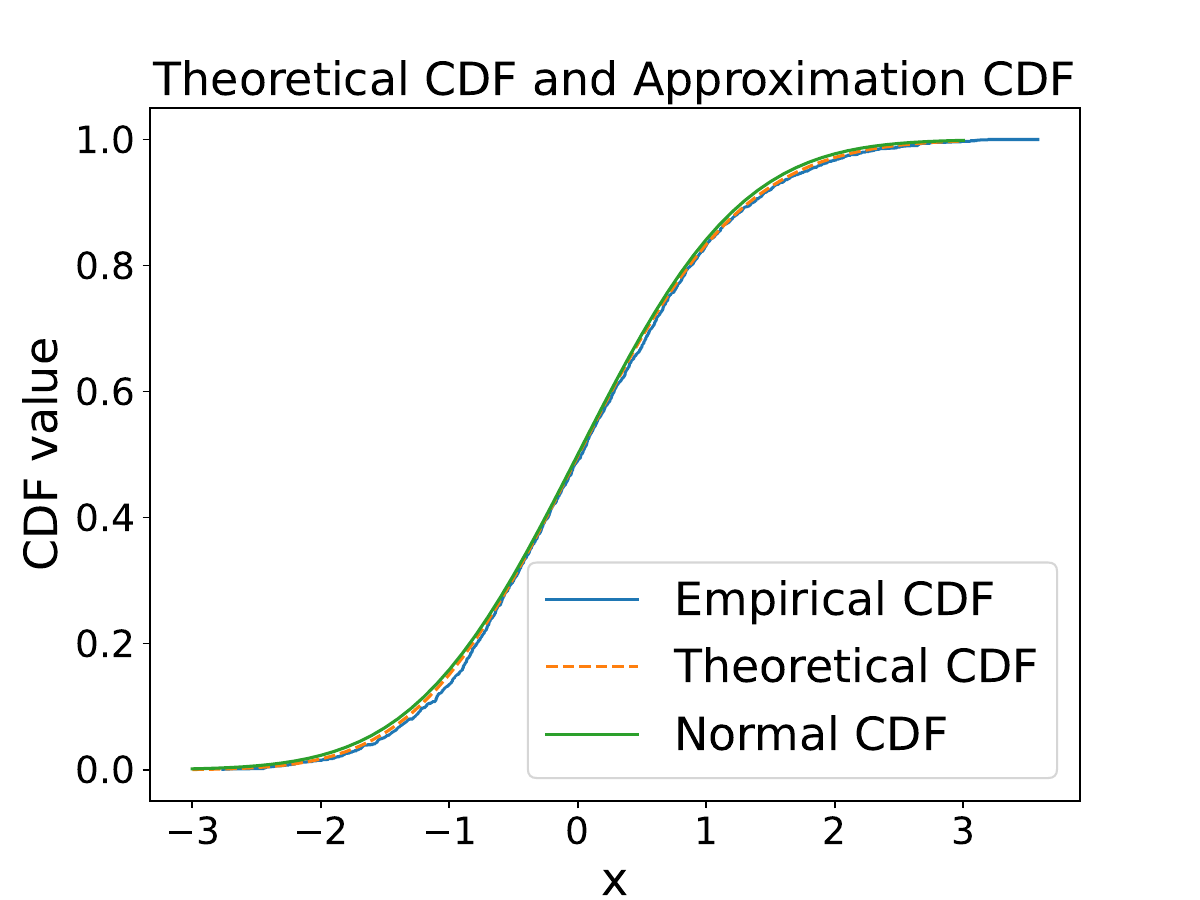} 
        \label{fig:BalanceN200T200Rho025S1Inferidx80Frac10CombineCDF}
    \end{subfigure}
    \caption{
    Comparison between the empirical CDF of the studentized ratio of network moments and its theoretical approximations. Each panel displays the empirical CDF, the CDF of the standard normal distribution, and the Edgeworth Expansion approximation derived in \Cref{main thm}.
    The left panel corresponds to Setting~\ref{two-poly} with the target inference time point $t^*=0$. The right panel corresponds to Setting~\ref{balance} with the target inference time point $t^*=-0.1$.
    }
    \label{fig:CDF.png}
\end{figure}
As shown in \Cref{fig:CDF.png}, compared with the standard normal approximation, the Edgeworth expansion based approximation provides a more accurate characterization of the sampling distribution of the test statistic.

\subsection{Impact of neighbor fraction for tuning}
\label{mcalF}
We further conduct simulation studies tp examine the effect of the fraction $\tau$ used for bandwidth selection in \eqref{opt_select_h}. 
In particular, we consider two settings of $T$ under the graphons defined in Setting \ref{two-poly}. In the first setting, we take $T=100$ and $n=200$, vary $\tau\in\{0.05,0.1,0.5,1\}$, and set the sparsity level to $\rho_n=0.25$. In the second setting, we increase $T$ to $1000$ while keeping all other parameters unchanged.
\begin{figure}[!h]
    \centering
    \includegraphics[width=\linewidth]{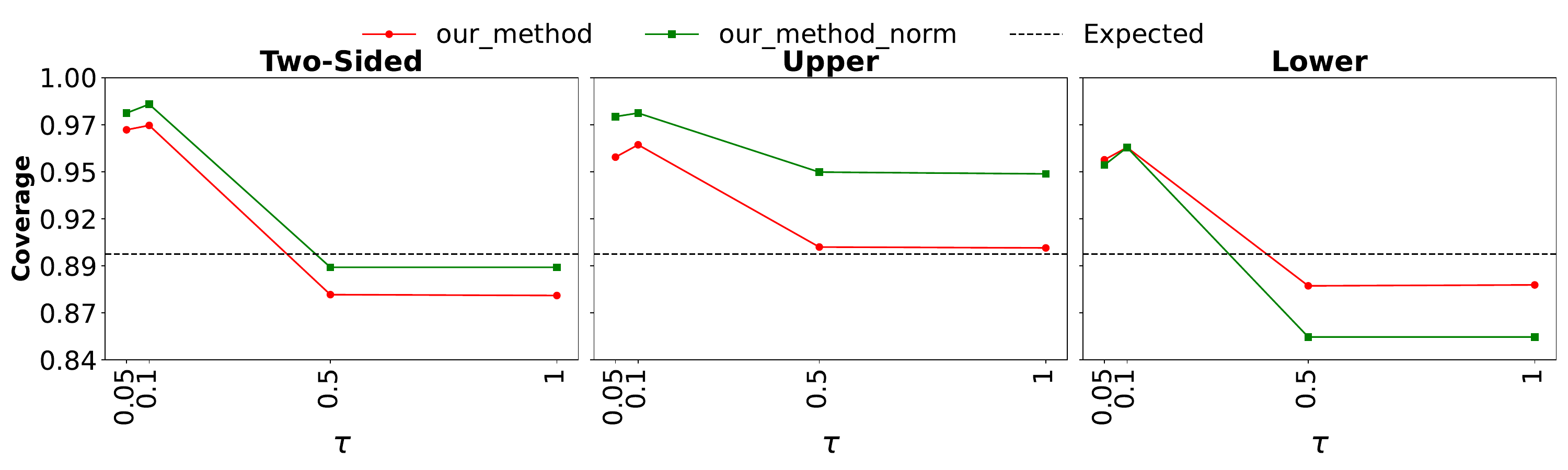}
    \caption{Coverage for Setting \ref{two-poly} with $n=200$, $T=100$ across different fractions $\tau$.}
    \label{CoverageAcrossFractionUnderTwoPolyN200T100T0RhoS08}
\end{figure}

\begin{figure}[!h]
    \centering
    \includegraphics[width=\linewidth]{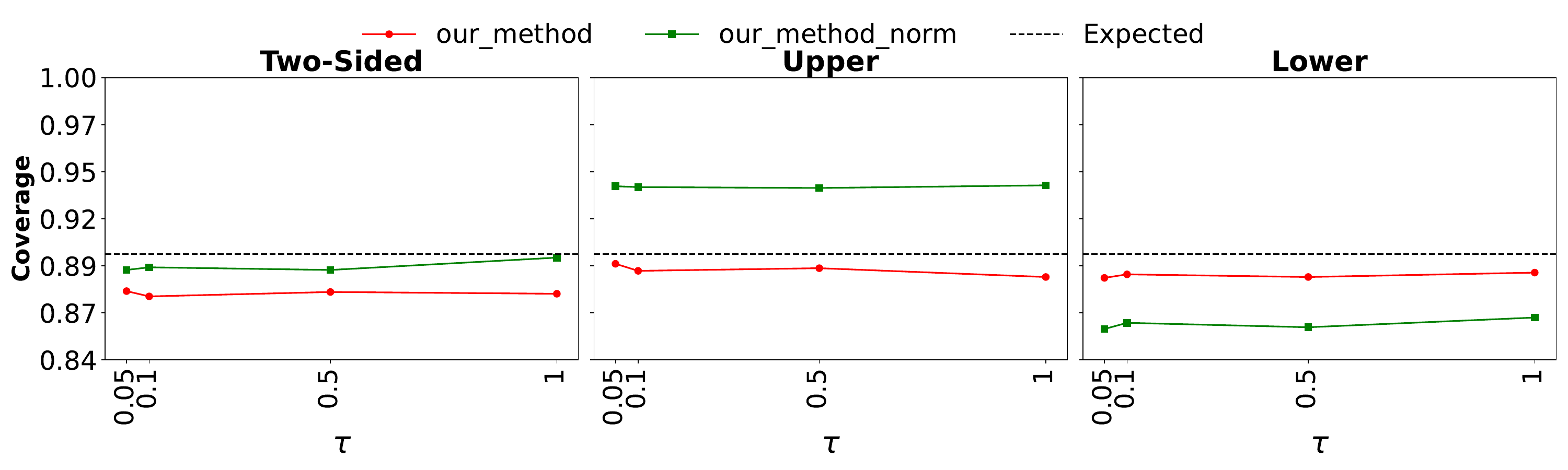}
    \caption{Coverage for Setting \ref{two-poly} with $n=200$, $T=1000$ across different fractions $\tau$.}
    \label{CoverageAcrossFractionUnderTwoPolyN200T1000T0RhoS08}
\end{figure}

As shown in \Cref{CoverageAcrossFractionUnderTwoPolyN200T100T0RhoS08}, when $T$ is moderately large, the coverage becomes closer to the nominal level as $\tau$ increases. 
As shown in
\Cref{CoverageAcrossFractionUnderTwoPolyN200T1000T0RhoS08}, for larger $T$, the coverage is more stable across different $\tau$ values, and further increasing $\tau$ has a smaller effect. These results suggest that using a larger tuning fraction is preferred when the number of observed time points $T$ is small, whereas the procedure is more robust to the choice of $\tau$ when $T$ is large.

\end{document}